\newif\iffull
\fulltrue

\iffull
\documentclass[a4paper,english]{article}

\usepackage{fullpage}

\else
\documentclass[a4paper,english]{socg-lipics-v2018}

\ArticleNo{71}
\fi

\usepackage{amsmath,amsfonts,amsthm}
\usepackage{graphicx}
\usepackage{mathtools}
\usepackage[hidelinks]{hyperref}

\usepackage{subcaption}

\usepackage[usenames,dvipsnames]{xcolor}

\iffull

\newtheorem{theorem}{Theorem}
\newtheorem{lemma}[theorem]{Lemma}

\theoremstyle{definition}
\newtheorem{definition}[theorem]{Definition}

\fi

\newcommand{\myremark}[3]{\textcolor{blue}{\textsc{#1 #2:}} \textcolor{SeaGreen}{\textsf{#3}}}

\newcommand{\maarten}[2][says]{\myremark{Maarten}{#1}{#2}}

\newcommand{\real}{\mathbb{R}}
\newcommand{\inte}{\textrm{int}\;}
\newcommand{\mydef}{:=}
\newcommand{\eps}{\varepsilon}

\newcommand{\comment}[1]{}

\title{Geometric Multicut}

\iffull

\author{Mikkel Abrahamsen\footnote{Basic Algorithms Research Copenhagen (BARC), University of Copenhagen, Universitetsparken 1, DK-2100 Copenhagen \O, Denmark. miab@di.ku.dk. Supported by the Innovation Fund Denmark through the DABAI project. BARC is supported by the VILLUM Foundation grant 16582.} \\
Panos Giannopoulos\footnote{giCenter, Department of Computer Science, City University of London, EC1V 0HB, London, United Kingdom. panos.giannopoulos@city.ac.uk.} \\
Maarten L\"{o}ffler\footnote{Department of Information and Computing Sciences, Utrecht University, The Netherlands. m.loffler@uu.nl. Partially supported by the Netherlands Organisation for Scientific Research (NWO); 614.001.504.} \\
G\"unter Rote\footnote{Institut f\"ur Informatik, 
Freie Universit\"at Berlin, Takustraße 9, 14195 Berlin, Germany. rote@inf.fu-berlin.de.}}

\else
\author{Mikkel Abrahamsen}{BARC, University of Copenhagen, Universitetsparken 1, DK-2100 Copenhagen \O, Denmark}{miab@di.ku.dk}{https://orcid.org/0000-0003-2734-4690}{Supported by the Innovation Fund Denmark through the DABAI project. MA is also a part of BARC, Basic Algorithms Research Copenhagen, supported by the VILLUM Foundation grant 16582.}
\author{Panos Giannopoulos}{giCenter, Department of Computer Science, City University of London, EC1V 0HB, London, United Kingdom}{panos.giannopoulos@city.ac.uk}{}{}
\author{Maarten L\"{o}ffler}{Department of Information and Computing Sciences, Utrecht University, The Netherlands}{m.loffler@uu.nl}{}{Partially supported by the Netherlands Organisation for Scientific Research (NWO); 614.001.504.}
\author{G\"unter Rote}{Institut f\"ur Informatik, 
Freie Universit\"at Berlin, Takustraße 9, 14195 Berlin, Germany}
{rote@inf.fu-berlin.de}{https://orcid.org/0000-0002-0351-5945}{}

\authorrunning{Mikkel Abrahamsen, Panos Giannopoulos, Maarten L\"offler, and G\"unter Rote}
\Copyright{Mikkel Abrahamsen, Panos Giannopoulos, Maarten L\"offler, and G\"unter Rote}


\acknowledgements{This work was initiated at the workshop
on \emph{Fixed-Parameter Computational Geometry}
at the Lorentz Center in Leiden
in May 2018. We thank the organizers and the Lorentz Center for a nice workshop and Michael Hoffmann for useful discussions during the workshop.}
\fi

\begin{document}

\maketitle

\begin{abstract}
We study the following separation problem:
Given a collection of colored objects in the plane, compute a shortest ``fence'' $F$, i.e., a union of curves of minimum total length, that separates every two objects of different colors.
Two objects are separated if $F$ contains a simple closed curve that has one object in the interior and the other in the exterior. We refer to the problem as GEOMETRIC $k$-CUT, where $k$ is the number of different colors, as it can be seen as a geometric analogue to the well-studied multicut problem on graphs.
We first give an $O(n^4\log^3 n)$-time algorithm that computes an optimal fence for the case where the input consists of polygons of two colors and $n$ corners in total.
We then show that the problem is NP-hard for the case of three colors.
Finally, we give a $(2-4/3k)$-approximation algorithm.
\end{abstract}

\clearpage
\section{Introduction}
\subparagraph{Problem Definition.}

We are given $k$ pairwise interior-disjoint, not necessarily connected, sets $B_1,B_2,\ldots,B_k$ in the
plane. We want to find a covering of the plane
 $\real^2 = \bar B_1\cup \bar B_2 \cup \cdots\cup \bar B_k$ such that the sets $\bar B_i$ are closed and interior-disjoint, $B_i\subseteq \bar B_i$ and
the total length of the boundary $F=\bigcup_{i=1}^k \partial\bar B_i$ between the different sets $\bar B_i$ is minimized. 

We think of the $k$ sets $B_i$ as having $k$
different \emph{colors} and each set $B_i$ as a union of simple geometric objects like
circular disks and simple polygons. 
Examples are shown in Figure~\ref{fig:4_circles_example1} and Figure~\ref{fig:polygon_example}.
We call $\bar B_i$ the \emph{territory} of color $i$.
The ``fence'' $F$ is the set of points that separates the territories. 
(Alternatively, $F$ is the set of points belonging to more than one territory.)
As we can see, a territory 
 can have more
than one connected component.

%
\begin{figure}[bt]
\centering
\includegraphics
{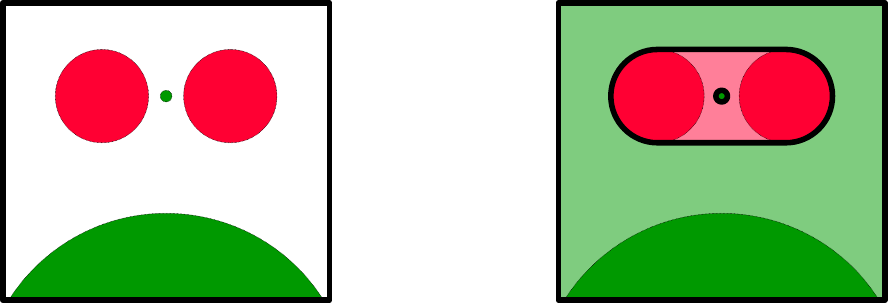}
\caption{An instance with $k=2$ sets, red and green, with two disks
  each; the big green disk is only partially shown. The optimal cover has a hippodrome-shaped red set, with the small green disk
  as a hole, and an unbounded green set. The fence $F$ has two
  components: the boundary of the hippodrome and the boundary of
  the small green disk.}
\label{fig:4_circles_example1}
\end{figure}
\begin{figure}[bt]
\centering
\includegraphics[width=\iffull 11\fi cm]
{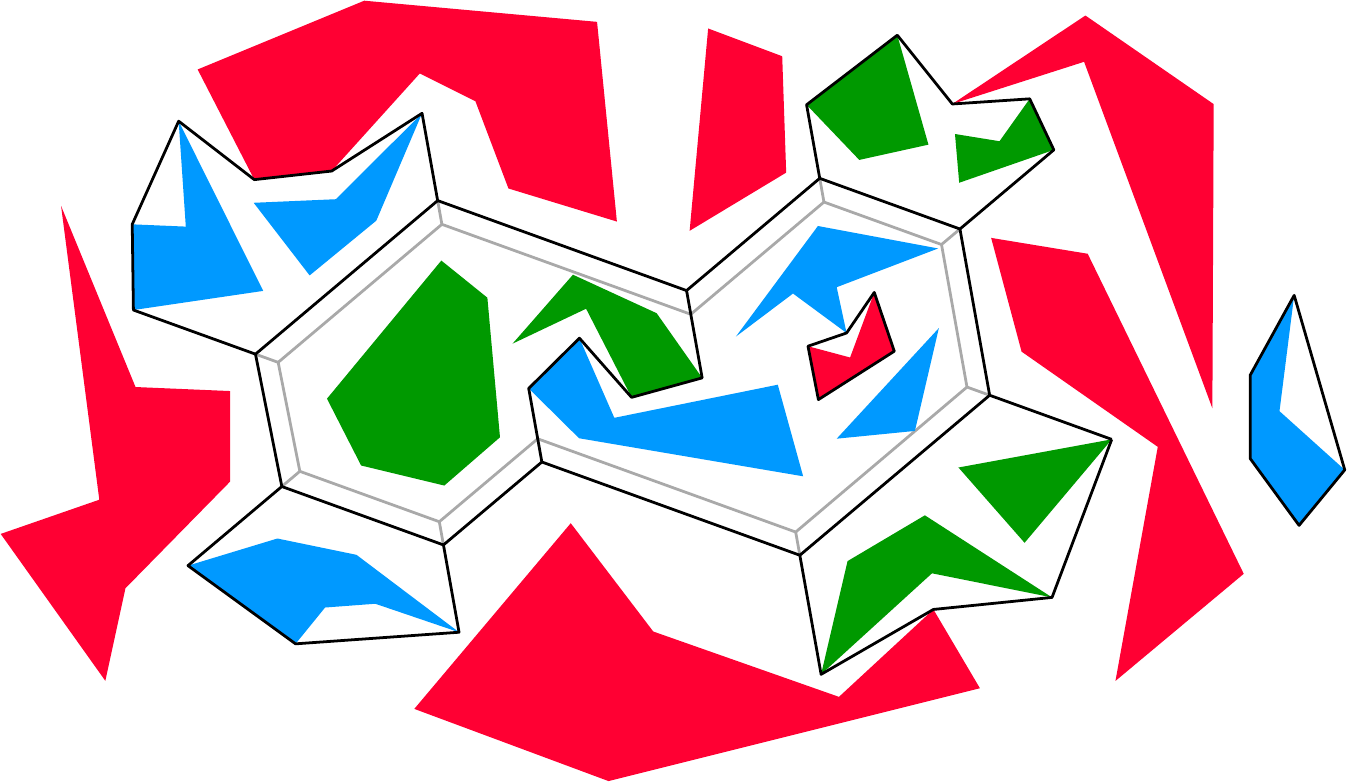}
\caption{An instance of GEOMETRIC $3$-CUT and an optimal fence in black.
The fence contains a cycle that does not touch any object.
The grey fence shows how the cycle can be shrunk without changing the
total length of the fence
(Lemma~\ref{lemma:no_cycles}).
%
}
\label{fig:polygon_example}
\end{figure}
%


An alternative view of the problem concentrates on
the \emph{fence}:
A fence is defined as
a union of curves $F$ such that each connected component of $\real^2\setminus F$ intersects at most one set $B_i$.
An interior-disjoint covering as defined above gives, by definition, such a fence.
Likewise, a fence $F$ induces such a covering,
by assigning each connected component of $\real^2\setminus F$ to an
appropriate territory $\bar B_i$.
The total length of a fence $F$ is also called the \emph{cost} of $F$ and is denoted as $|F|$.



In our paper, we will focus on the case where the input consists of
simple polygons 
(with disjoint interiors)\iffull\
where the corners have rational
coordinates\fi.
We denote this problem as \emph{GEOMETRIC $k$-CUT},
Each input polygon is called an \emph{object}.
 We use $n$ to
denote the total number of corners of the input polygons.
We count the corners with multiplicity so that if a point is a corner
of more objects, it is counted for each object
individually.

\iffull
The results can be extended to more general shapes as long as they are
reasonably well behaved.
\fi

Even in this simple setting, the problem poses both geometric and combinatorial difficulties.  A
set $B_i$ can consist of disconnected pieces, and the combinatorial challenge is to
choose which of the pieces should be grouped into the same
component of~$\bar B_i$.
The geometric
task is to construct a network of curves that surrounds the given
groups of objects and thus separates the groups from each other.
For $k=2$ colors, optimal fences consist of geodesic curves around
obstacles, which are well understood.
As soon as the number $k$ of colors exceeds~$2$, the geometry
 becomes more complicated, and
the problem acquires
traits of the geometric Steiner tree problem, as shown by the example in Figure~\ref{fig:polygon_example}.

The problem of enclosing a set of objects by a shortest system of
fences has been considered with a single set
$B_1$ by Abrahamsen et al.~\cite{abrahamsen2018fast}. The task is to ``enclose'' the
components of $B_1$ by a shortest system of fences.
This can be formulated as a special case of our problem with
$k=2$ colors: We add an additional
set $B_2$, far away from $B_1$
and large enough so that
it is never optimal to enclose
$B_2$.
Thus, we have to enclose all components of $B_1$ and separate them from
the unbounded region. In this setting, there will be no nested fences.
Abrahamsen et al.~gave an algorithm
with running time $O(n\;\text{polylog}\;n)$ for the case where the
input consists of $n$ unit disks.

\subparagraph{Our Results.}
\label{sec:results}

In Section~\ref{sec:2}, we show how to solve the case with $k=2$ colors in time $O(n^4\log^3 n)$.
The algorithm works by reducing the problem to the multiple-source
multiple-sink maximum flow problem in a planar graph.
In Section~\ref{sec:3}, we show that already the case with $k=3$ colors is NP-hard by a reduction from PLANAR POSITIVE 1-IN-3-SAT.

In Section~\ref{sec:a}, we discuss approximation algorithms.
We first compare the optimal fence $F_{\mathcal A}$ consisting of line segments between corners of input polygons to the unrestricted optimal fence~$F^*$.
We show that $|F_{\mathcal A}|\leq 4/3\cdot |F^*|$.
After applying a $(3/2-1/k)$-approximation algorithm for the $k$-terminal multiway cut problem~\cite{cualinescu2000improved}, we obtain a polynomial-time $(2-\frac{4}{3k})$-approximation algorithm for GEOMETRIC $k$-CUT (Theorem~\ref{approximation-theorem}).
\section{Structure of Optimal Fences}

%

\iffull
\fi

\begin{lemma}\label{lemma:char}
An optimal fence $F^*$ is a union of (not necessarily disjoint)
closed curves, 
disjoint from the interior of the objects.
Furthermore, $F^*$ is the union of straight line segments of positive length.
Consider two non-collinear line segments $\ell_1,\ell_2\subset F^*$
with a common endpoint~$p$.
If $p$ is not a corner of an object, then exactly three line segments
meet at $p$ and form angles of $2\pi/3$.
\end{lemma}

\begin{proof}
  \iffull
We first prove that the curves in $ F^*$ are disjoint from the interior of each object.
To this end, consider any fence $ F$ in which some open curve $\pi\subset F$ is contained in the interior of an object $O\subset B_i$.
Then the domains on both sides of $\pi$ must be part of the territory $\overline B_i$.
Hence, $\pi$ can be removed from $ F$ while the fence remains feasible.
That operation reduces the length, so $ F$ is not optimal.

\fi

We next show that $F^*$ is the union of a set of closed curves. Suppose not.
Let $F'\subset F^*$ be the union of all closed curves contained in $F^*$ and let $\pi$ be a connected component in $F^*\setminus F'$.
Then $\pi$ is the (not necessarily disjoint) union of a set of open curves, which do not contribute to the separation of any objects.
Hence, $F^*\setminus \pi$ is a fence of smaller length than $F^*$, so $ F^*$ is not optimal.

In a similar way, one can consider the union $L$ of all line segments of positive length contained in $F^*$, and if $F^*\setminus L$ is non-empty, a standard argument shows that a curve $\pi$ in $F^*\setminus L$ can be replaced by line segments, thus reducing the total length.

The last claimed property is shared with the Euclidean Steiner minimal tree on
a set of points in the plane, and it can be proved
\iffull
in the same way,
see for example Gilbert and Pollak~\cite{gilbert1968steiner}:
Suppose that the fence $F$ contains
 two non-collinear line segments $\ell_1$ and $\ell_2$ sharing an endpoint $p$ that is not a corner of an object.
If the angle between
$\ell_1$ and $\ell_2$ at $p$ is less than $2\pi/3$, then parts of
$\ell_1$ and $\ell_2$ can be replaced by three shorter segments.
Hence, the angle between segments meeting at $p$ is
at least $2\pi/3$, and there can be at most three such line segments.
If there are only two, one can make a shortcut.
Therefore, there are exactly three segments, and they form angles of~$2\pi/3$.
\fi
\end{proof}

As it can be seen in Figure~\ref{fig:polygon_example}, optimal fences may contain cycles that do not touch any object.
By the following lemma, such cycles can be eliminated without
increasing the length.
\iffull
This will turn out to be useful in our design of approximation
algorithms.
\fi

\begin{lemma}\label{lemma:no_cycles}
Let $N$ be the set of corners of the objects in an instance of GEOMETRIC $k$-CUT.
There exists an optimal fence $ F^*$ with the property that $ F^*\setminus N$ contains no cycles.
\end{lemma}

\begin{proof}
Let us look at a connected component $T$ of $ F^*\setminus N$.
 By Lemma~\ref{lemma:char}, its leaves are
 in $N$.  All other vertices have degree 3, and the incident edges meet at angles of $2\pi/3$.
 If $T$ contains a cycle $C$, we can push the edges of $C$ 
in a parallel fashion (forming an offset curve),
as shown in Figure~\ref{fig:polygon_example}. This operation does not change the
total length of $T$. This can be seen by looking at each degree-3
vertex $v$ individually: We enclose $v$ in a small equilateral
triangle whose sides cut the edges at right angles, 
see Figure~\ref{fig:triangle}. It is an easy
geometric fact that the sum of the distances from a point inside an equilateral
triangle to the three sides is constant. This implies that the length of 
the fence
 inside the triangle is unchanged by the offset operation.
The portions of $C$ 
 outside the triangles are just translated
and do not change their lengths either.
\begin{figure}[hbt]
\centering
\includegraphics
{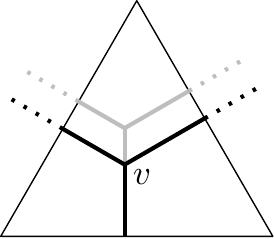}
\caption{Offsetting the cycle does not change the total length of the
  fence inside the triangle.}
\label{fig:triangle}
\end{figure}
%


As we offset the cycle $C$, an edge of $C$ must eventually hit a
corner of an object. 
Another conceivable possibility is that an edge of $C$ between two degree-3 vertices is reduced to a
point, but this can be excluded because it would lead to an optimal fence violating
 Lemma~\ref{lemma:char}.

In this way, the cycles of $T$ can be eliminated one by one.
\end{proof}

\section{The Bicolored Case}
\label{sec:2}

In this section we consider the case 
of $k=2$ different colors.
Let $N$ be the set of all corners of the objects.
A line segment is said to be \emph{free} if it is disjoint from the
interior of every object.
A vertex $v$ of an optimal fence cannot have degree $3$ or more unless $v\in N$, as otherwise two
of the regions meeting at $v$ would be part of the same territory and
could be merged, thus reducing the length.
We therefore get the following consequence of Lemma~\ref{lemma:char}.

\begin{lemma}
\label{lemma:opt_fence_polygons}
An optimal fence consists of free line segments with endpoints in $N$.\qed
\end{lemma}


\begin{figure}
\centering
\includegraphics{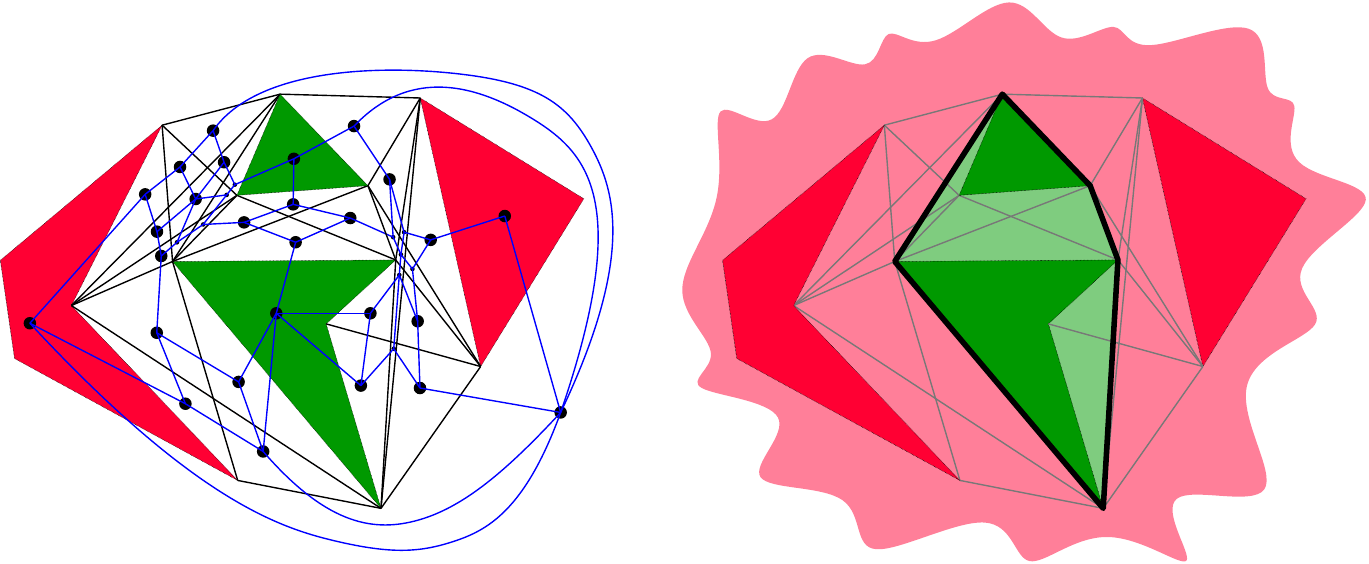}
\caption{Left: The arrangement $\mathcal A$ induced by an instance of GEOMETRIC $2$-CUT with two green and two red objects. The edges of the dual graph $G$ are blue. Right: The optimal solution.}
\label{fig:bichrom}
\end{figure}

Let $S$ be the set of all free segments with endpoints in $N$.
$S$ includes all edges of the objects.
Let $\mathcal{A}$ be the arrangement induced by $S$,
see Figure~\ref{fig:bichrom}.
Consider an optimal fence $F^*$ and the associated territories $\bar B_1$ and $\bar B_2$.
Lemma~\ref{lemma:opt_fence_polygons} implies that $F^*$ is contained in~$\mathcal{A}$.
Thus, each cell of $\mathcal{A}$ belongs entirely either to $\bar B_1$ or $\bar B_2$.
The objects are cells of $\mathcal{A}$ whose classification (i.e., membership of $\bar B_1$ versus $\bar B_2$) is fixed.
In order to find $F^*$, we need to select the territory that each of
the other cells belongs to.
Since $|S|=O(n^2)$, $\mathcal{A}$
has size $O(|S|^2)=O(n^4)$ and can be computed in $
O(|\mathcal{A}|)= O(n^4)$
time~\cite{chazelle1992optimal}.
For simplicity, we stick with the worst-case bounds.
In practice, set $S$ can be pruned by observing that the edges of
an optimal fence must be \emph{bitangents} that touch the objects in a
certain way, because the curves of the fence are locally shortest.

 Finding an optimal fence amounts to minimizing the boundary between $\bar B_1$ and $\bar B_2$.
This can be formulated as a minimum-cut problem in the dual graph $G(V, E)$ of the arrangement $\mathcal{A}$.  There is a node in $V$ for each cell and a weighted edge in $E$ for each pair of adjacent cells: the weight of the edge is the length of the cells' common boundary.
Let $S_1,S_2\subset V$ be the sets of cells that contain the objects of $B_1,B_2$, respectively.
We need to find the minimum cut that separates $S_1$ from $S_2$.
This can be obtained by finding the maximum flow in $G$ from the sources $S_1$ to the sinks $S_2$, where the capacities are the weights.
As $G$ is a planar graph, we can use the algorithm by Borradaile et al.~\cite{borradaile2017multiple} with running time $O(|V|\log^3|V|)$.
The running time has since then been improved to $O(\frac {|V|\log^3|V|}{\log^2\log |V|})$~\cite{gawrychowski2016improved}.
As $|V|=O(|S|^2)=O(n^4)$, we obtain the following theorem.

\begin{theorem}
GEOMETRIC $2$-CUT can be solved in time $O(\frac{n^4\log^3
  n}{\log^2\log n})$.
\qed
\end{theorem}

A similar algorithm has been described before in a slightly different context: image segmentation~\cite{greig1989exact}, see also~\cite{borradaile2017multiple}.
Here, we have a rectangular grid of pixels, each having
 a given gray-scale value.
Some pixels are known to be either black or white.
The remaining pixels have to be assigned either the black or the white
color.
Each pixel has edges to its (at most four) neighbors.
The weights of these  edges 
 can be chosen in such a way  that the minimum cut problem corresponds to minimizing a cost function consisting of two parts:
One part, the \emph{data component}, has a term for each pixel, and it
measures the discrepancy between the gray-value of the pixel and the assigned value.
The other part, the \emph{smoothing component}, penalizes neighboring
pixels
 with similar gray-values
 that are assigned different colors. 


\section{Hardness of the Tricolored Case}
\label{sec:3}

We show how to construct an instance $I$ of GEOMETRIC $3$-CUT from an instance $\Phi$ of PLANAR POSITIVE 1-IN-3-SAT.
For ease of presentation, we first describe the reduction
geometrically, allowing irrational coordinates. 
We prove that if $\Phi$ is satisfiable, then $I$ has a fence of cost $M^*$, whereas if $\Phi$ is not satisfiable, then the cost is at least $M^*+1/50$.
We then argue that the corners 
can be slightly moved to make a new instance $I'$ with rational
coordinates while still being able to distinguish whether $\Phi$ is
satisfiable or not, based on the cost of an optimal fence. 

In order to make the proof as simple as possible, we introduce a new specialized problem COLORED TRIGRID POSITIVE 1-IN-3-SAT in the following.

\subsection{Auxiliary NP-complete problems}

\begin{definition}
In the POSITIVE 1-IN-3-SAT problem, we are given a collection $\Phi$ of clauses containing exactly three distinct variables (none of which are negated).
The problem is to decide whether there exists an assignment of truth values to the variables of $\Phi$ such that exactly one variable in each clause is true.
\end{definition}





\begin{definition}
In the TRIGRID POSITIVE 1-IN-3-SAT problem, we are given an instance $\Phi$ of POSITIVE 1-IN-3-SAT together with a planar embedding of an associated graph $G(\Phi)$ with the following properties:
\begin{itemize}
\item
$G(\Phi)$ is a subgraph of a regular triangular grid,

\item
for each variable $x$, there is a simple cycle $v_x$,

\item
for each clause $C=\{x,y,z\}$, there is a path $c_C$ and three vertical paths $\ell^C_x,\ell^C_y,\ell^C_z$ with one endpoint at a vertex of $c_C$ and one at a vertex of each of $v_x,v_y,v_z$,

\item
except for the described incidences, no edges share a vertex,

\item
all vertices have degree $2$ or $3$,

\item
any two adjacent edges form an angle of $\pi$ or $2\pi/3$,

\item
the number of vertices is bounded by a quadratic function of the size of $\Phi$.
\end{itemize}
The problem is to decide whether $\Phi$ has a satisfying assignment.
\end{definition}

\begin{figure}
\centering
\includegraphics[page=1]{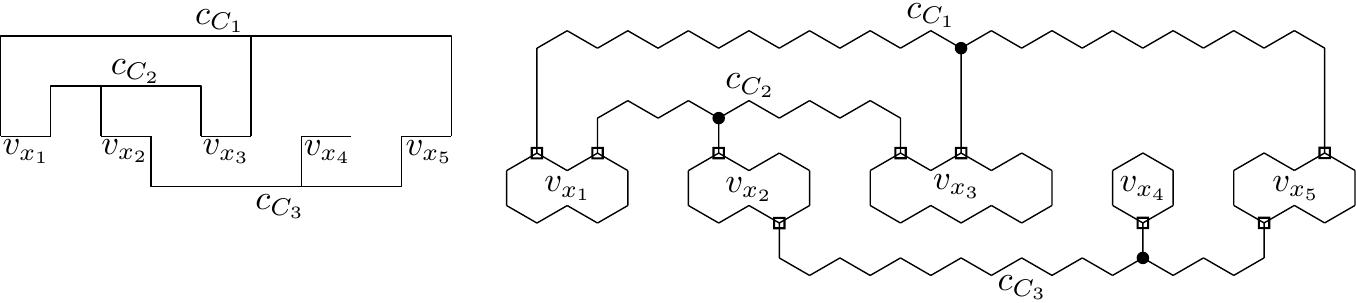}
\caption{Left: An instance of PLANAR POSITIVE 1-IN-3-SAT for the formula $\Phi=C_1\land C_2\land C_3$ for $C_1=x_1 \lor x_3 \lor x_5$, $C_2=x_1 \lor x_2 \lor x_3$, and $C_3=x_2 \lor x_4 \lor x_5$.
Right: A corresponding instance of TRIGRID POSITIVE 1-IN-3-SAT.
Clause vertices are drawn as dots and branch vertices as boxes.}
\label{trigrid:fig}
\end{figure}

Mulzer and Rote~\cite{mr-mwtnh-08} showed that another problem, PLANAR POSITIVE 1-IN-3-SAT, is NP-complete, which is similar but uses a slightly different embedding with axis-parallel segments.
It trivially follows that TRIGRID POSITIVE 1-IN-3-SAT is also NP-complete, see Figure~\ref{trigrid:fig}.



Consider an instance $(\Phi,G(\Phi))$ of TRIGRID POSITIVE 1-IN-3-SAT.
There are some vertices of degree three on the cycles $v_x$ corresponding to each variable $x$ in $\Phi$, and these we denote as \emph{branch vertices} of $G(\Phi)$.
There is also one vertex of degree three on the path $c_C$ corresponding to each clause $C$ in $\Phi$, which we denote as a \emph{clause vertex}.
Except for branch and clause vertices, at most two edges meet at each vertex.

Let $\mathcal C$ be the set of all clause vertices (considered as geometric points).
Removing $\mathcal C$ from $G(\Phi)$ (considered as a subset of $\mathbb R^2$) splits $G(\Phi)$ into one connected component $E_x$ for each variable $x$ of $\Phi$.
The idea of our reduction to GEOMETRIC $3$-CUT is to build a \emph{channel} on top of $E_x$ for each variable $x$.
The channel has constant width $1/2$ and contains $E_x$ in the center.
The channel contains small \emph{inner} objects and is bounded by larger \emph{outer objects} of another color.
There will be two equally good ways to separate the inner and outer objects, namely taking an individual fence around each inner object and taking long fences along the boundaries of the channel that enclose as many inner objects as possible.
Any other way of separating the inner from the outer objects will require more fence.
These two optimal fences play the roles of $x$ being true and false, respectively.

At the clause vertices where three regions $E_x,E_y,E_z$ meet, we make a clause gadget that connect the three channels corresponding to $x,y,z$.
The objects in the clause gadget can be separated using the least amount of fence if and only if one of the channels is in the state corresponding to true and the other two are in the false state.
Therefore, this corresponds to the clause in $\Phi$ being satisfied.

In order to make this idea work, we first assign every edge of $G(\Phi)$ an \emph{inner} and an \emph{outer} color among $\{\text{red},\text{green},\text{blue}\}$.
These will be used as the colors of the inner and outer objects of the channel later on.
We require the following of the coloring:
\begin{enumerate}
\item
The inner and outer colors of any edge are distinct.\label{col:req1}

\item
Any two adjacent collinear edges have the same inner or outer color.\label{col:req2}

\item
Any two adjacent edges that meet at an angle of $2\pi/3$ at a non-clause vertex have the same inner and the same outer color.\label{col:req3}

\item
The inner colors of the three edges meeting at a clause vertex are red, green, blue in clockwise order, while the outer colors of the same edges are blue, red, green, respectively.\label{col:req4}
\end{enumerate}

We now introduce the problem COLORED TRIGRID POSITIVE 1-IN-3-SAT, which we will reduce to GEOMETRIC $3$-CUT.

\begin{figure}
\centering
\includegraphics[page=2,width=0.75\textwidth]{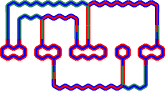}
\caption{An instance of COLORED TRIGRID POSITIVE 1-IN-3-SAT based on the instance from Figure~\ref{trigrid:fig}.}
\label{fig:2zigzag}
\end{figure}

\begin{definition}\label{def:CZZPP}
In the COLORED TRIGRID POSITIVE 1-IN-3-SAT problem, we are given an instance $(\Phi,G(\Phi))$ of TRIGRID POSITIVE 1-IN-3-SAT together with a coloring of the edges of $G(\Phi)$ satisfying the above requirements.
We want to decide whether $\Phi$ has a satisfying assignment.
\end{definition}

\iffull
\begin{lemma}\label{lemma:ctp}
The problem COLORED TRIGRID POSITIVE 1-IN-3-SAT is NP-complete.
\end{lemma}

\begin{proof}
Membership in NP is obvious.
For NP-hardness, we reduce from TRIGRID POSITIVE 1-IN-3-SAT.
Let $(\Phi,G(\Phi))$ be an instance of the latter and $G'(\Phi)$ be the graph obtained from $G(\Phi)$ by expanding the vertical paths $\ell^C_x$, so that they have length at least $4$.
The cycles $v_x$ and paths $c_C$ are shifted up or down accordingly, see Figure~\ref{fig:2zigzag}.
We specify the coloring of $G'(\Phi)$ below.

We color each triple of edges meeting at a clause vertex so that requirement~\ref{col:req4} is met.
In each of the paths $c_C$, we have colored one edge on each side of the clause vertex, and we use the colors of that edge to color the rest of the edges on that side.
Next, we choose two distinct colors that we use as the inner and outer colors of all the cycles $v_x$ containing the branch vertices.
For each vertical path $\ell^C_x$, the edge adjacent to the cycle $v_x$ has to be colored with the same two colors.

It remains to color some edges of each vertical path $\ell^C_x$.
Since the vertical paths have length at least $4$ and only the first and last edges have been colored, it is possible to change inner and outer color three times along each of them.
The maximum number of changes is needed when the edge adjacent to a clause vertex has the inner and outer colors swapped as compared to the edge adjacent to a branch vertex, in which case exactly three changes are needed.
Therefore, it is possible to adjust the colors so that the entire path gets colored.
We have hence constructed an instance of COLORED TRIGRID POSITIVE 1-IN-3-SAT based on the same instance $\Phi$ of POSITIVE 1-IN-3-SAT that we were given.
\end{proof}
\else
\fi

\subsection{Building a GEOMETRIC $3$-SAT instance from tiles}\label{sec:tileDefs}

Consider an instance $(\Phi,G(\Phi))$ of COLORED TRIGRID POSITIVE 1-IN-3-SAT that we will reduce to GEOMETRIC $3$-CUT.
We make the construction using hexagonal \emph{tiles} of six different types, namely \emph{straight}, \emph{inner color change}, \emph{outer color change}, \emph{bend}, \emph{branch}, and \emph{clause} tiles.
Each tile is a regular hexagon with side length $1/\sqrt 3$ and hence has width $1$.
The tiles are rotated such that they have two horizontal edges.

The tiles are placed so that each tile is centered at a vertex $p$ of $G(\Phi)$.
Let $G_p$ be the part of $G(\Phi)$ within distance $1/2$ from $p$ (recall that each edge of $G(\Phi)$ has length $1$).
Figure~\ref{fig:tiles} shows the tiles and how they are placed according to the shape and colors of $G_p$.

\begin{figure}
\centering
\halign{&\hfil#\hfil\cr
straight\ &inner color change\ & outer color change\ \cr
\includegraphics[width=0.33\textwidth]{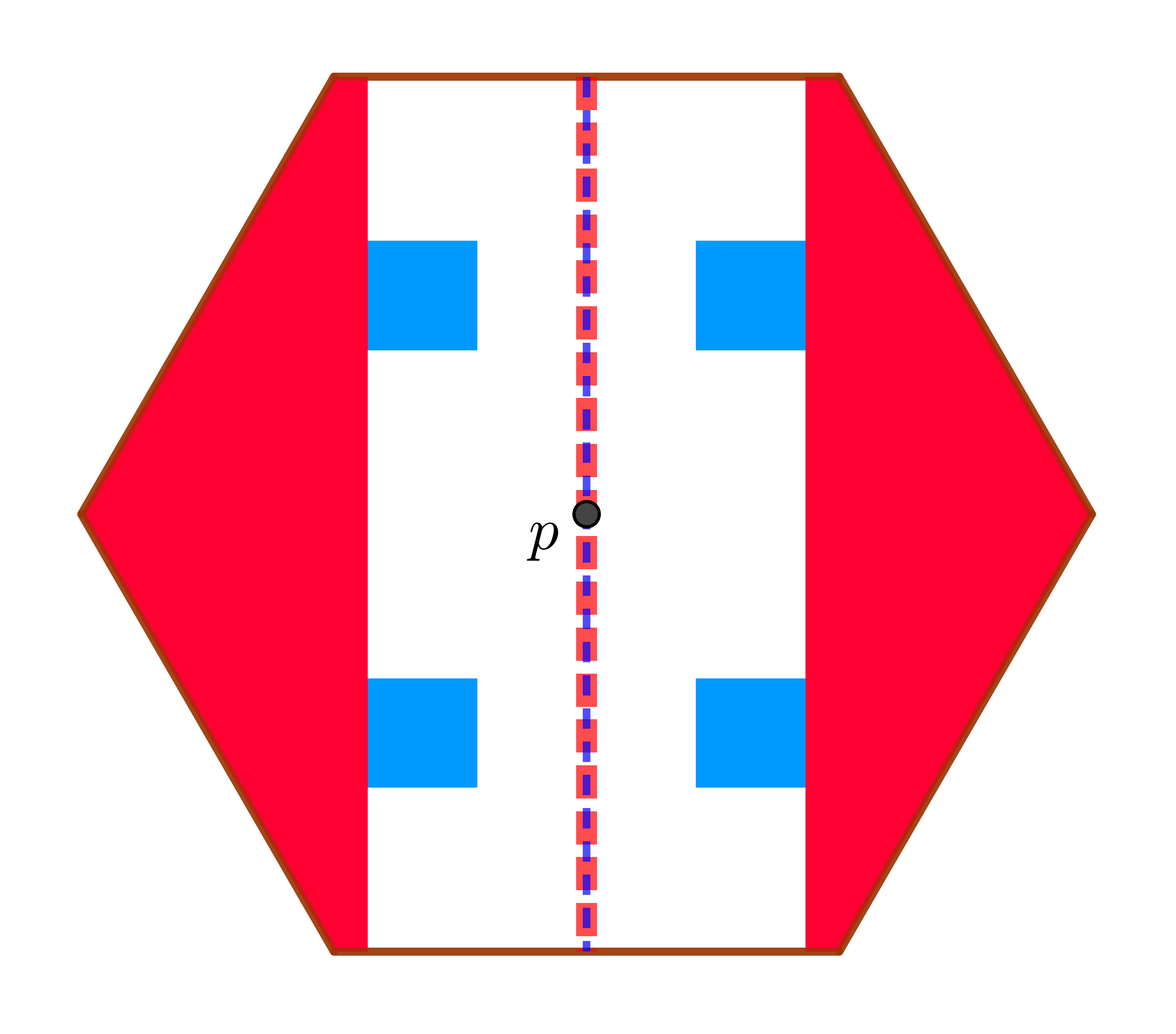}
&
\includegraphics[width=0.33\textwidth]{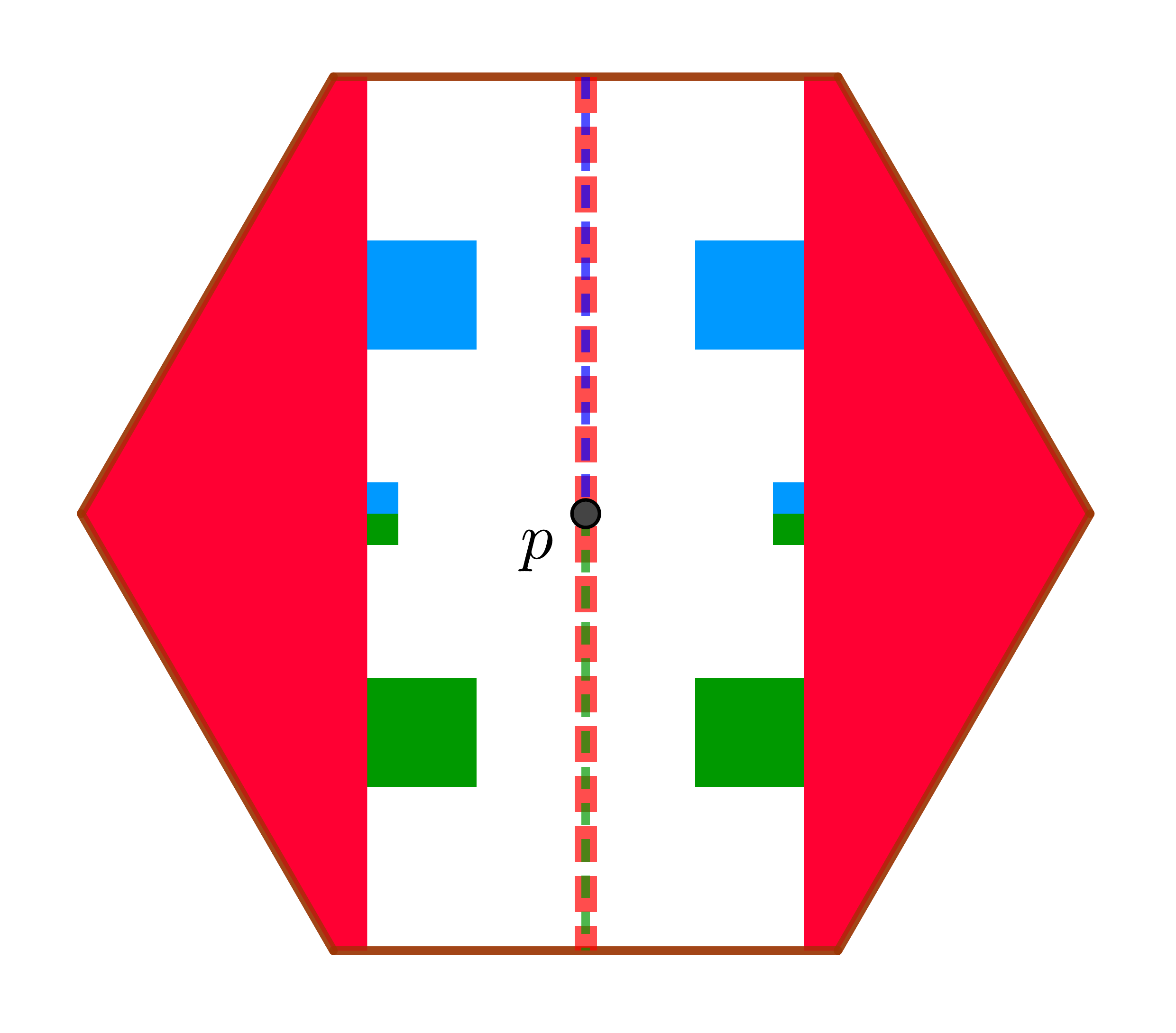}
&
\includegraphics[width=0.33\textwidth]{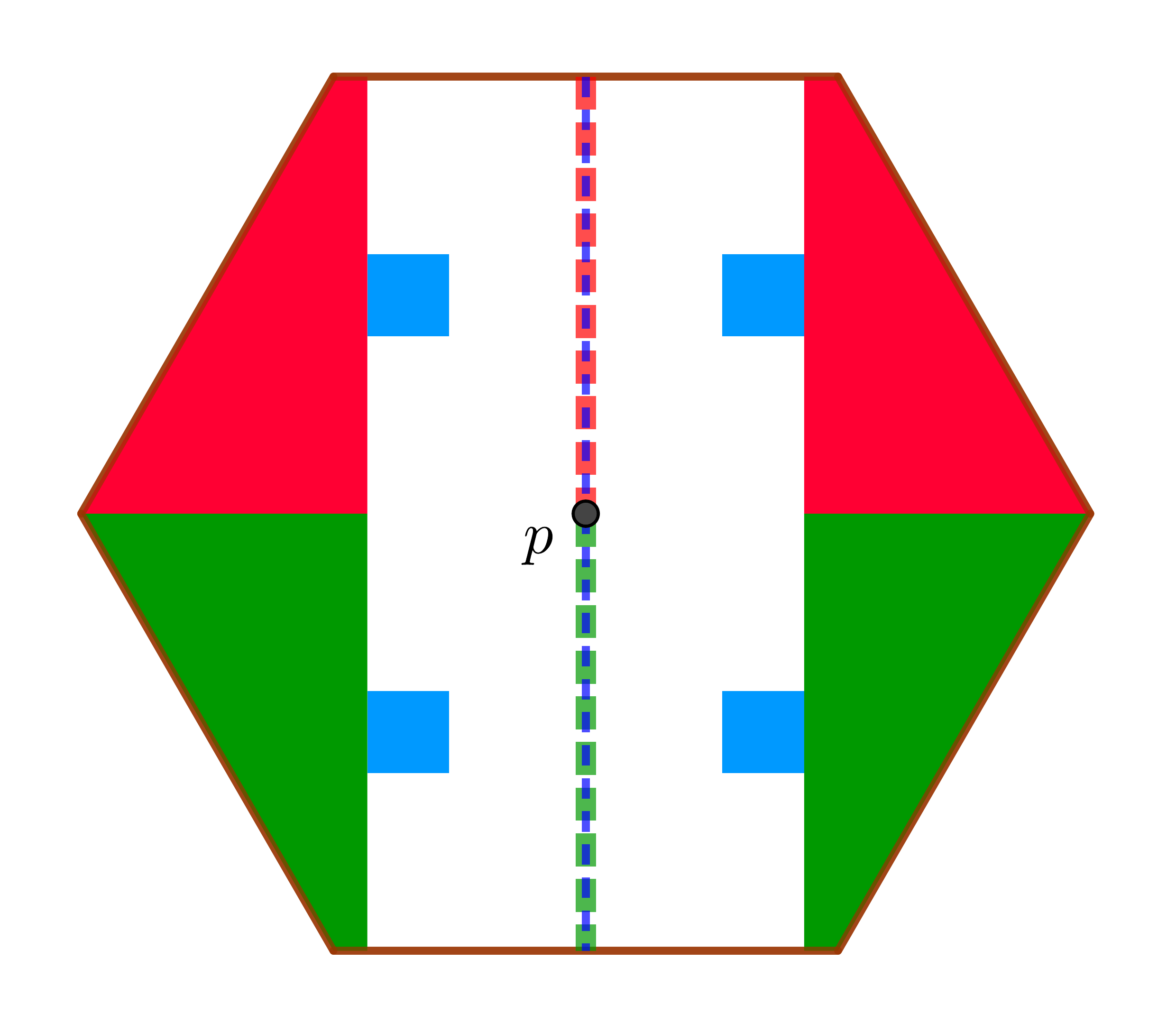}
\cr
 bend & branch & clause \cr
\includegraphics[width=0.33\textwidth]{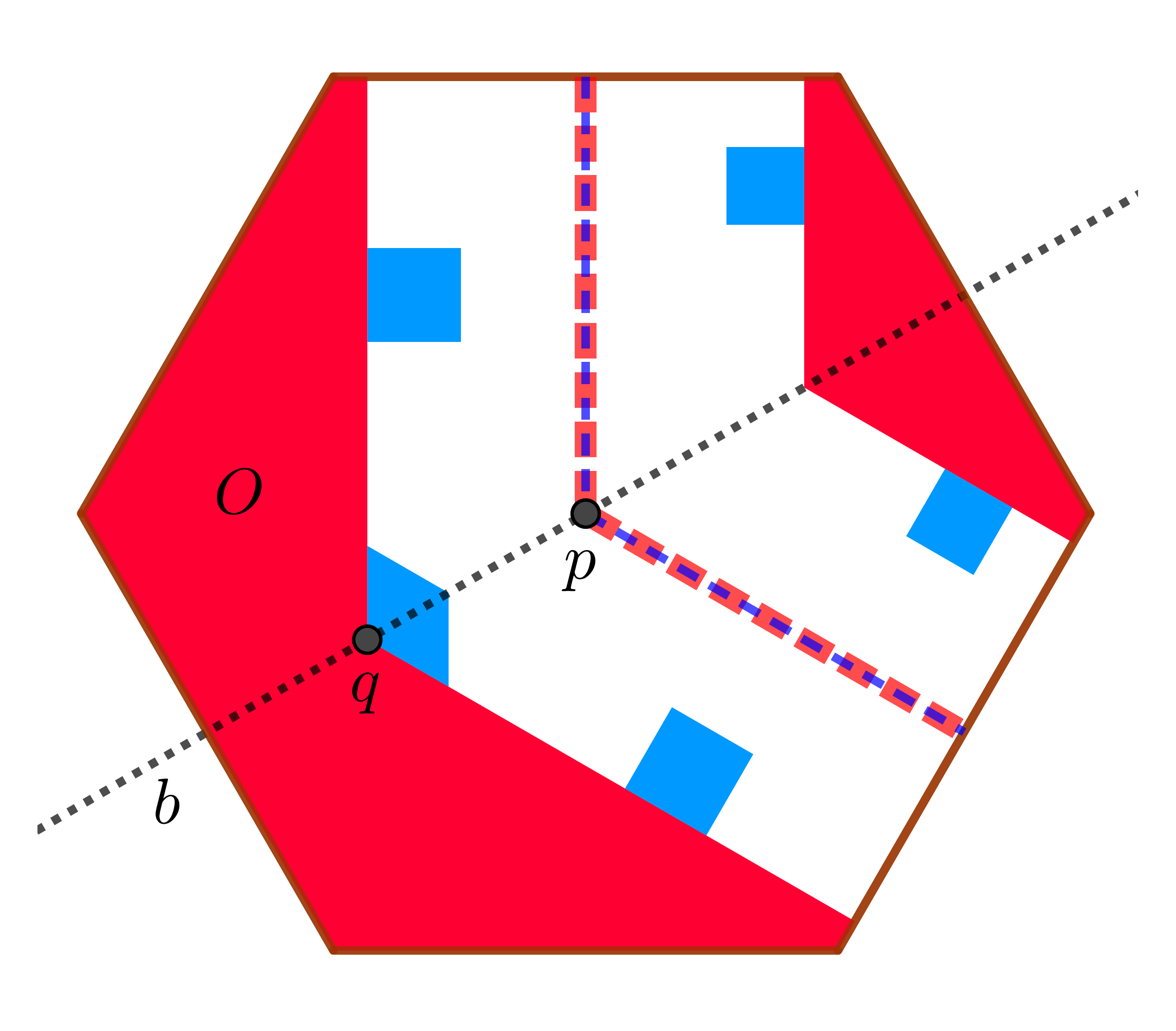}
&
\includegraphics[width=0.33\textwidth]{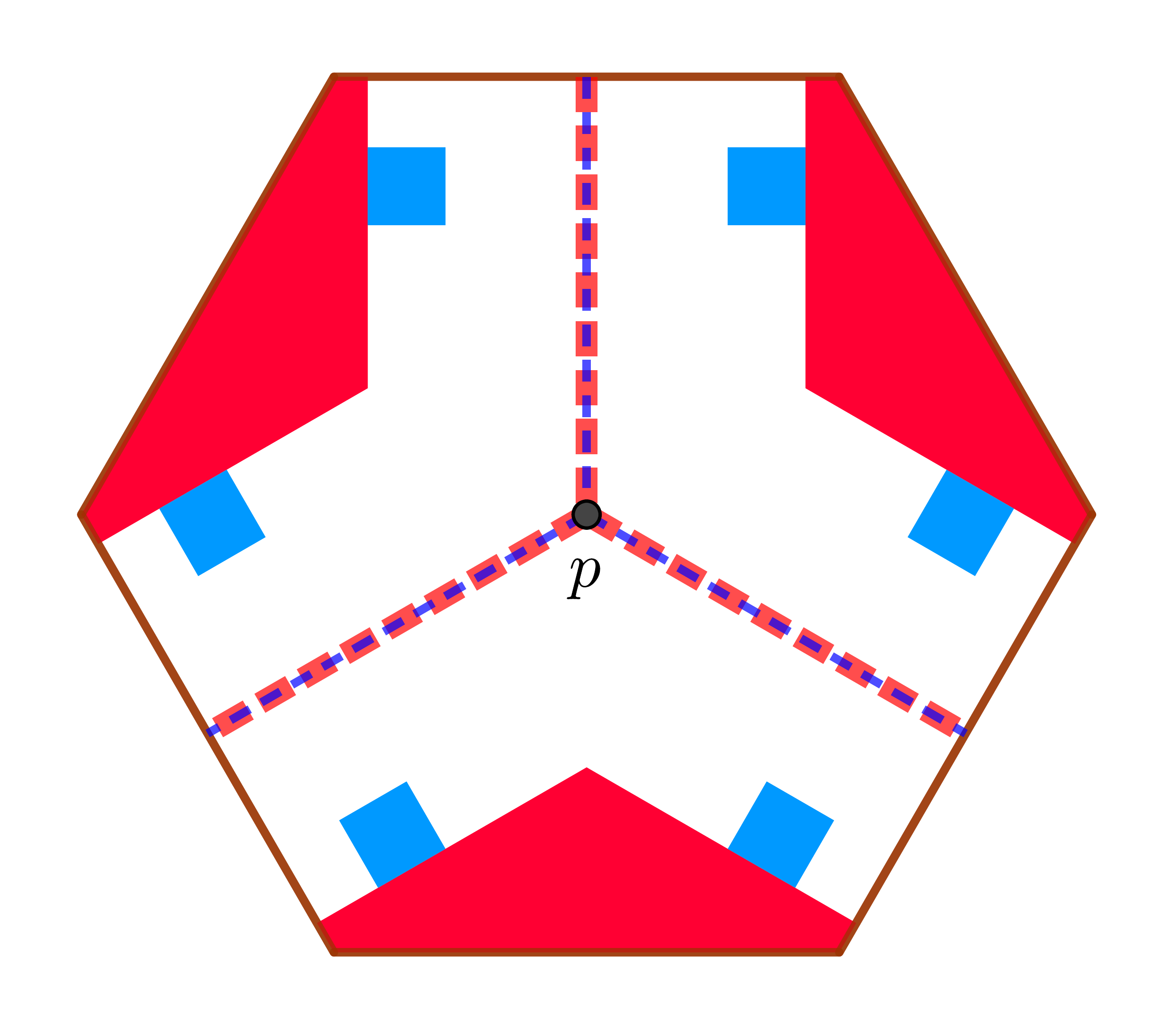}
&
\includegraphics[width=0.33\textwidth]{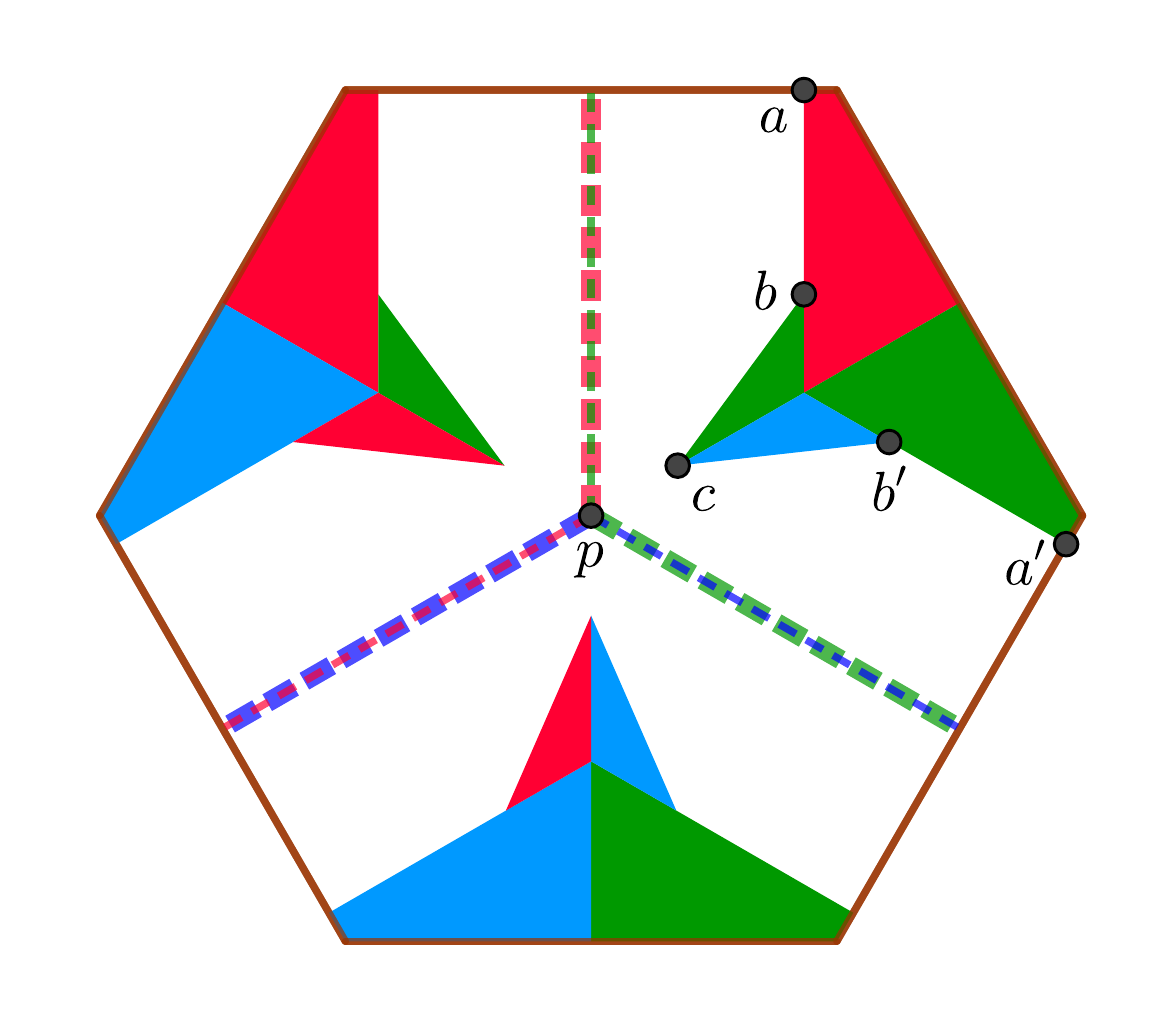}
\cr
}

\caption{Different kinds of tiles used in the reduction to GEOMETRIC
  $3$-CUT.
The dashed colored segments show $G_p$ and the inner and outer color of $G_p$.
The tiles are colored accordingly.
The points in the clause tile are defined so that $\|ab\|=\|a'b'\|=6/25=0.24$ and $\|bc\|=\|b'c\|=1/4=0.25$.
Point $c$ has coordinates $(x,x/\sqrt 3)$, where $x= {\frac {13\sqrt {3}}{200}}+3/16-{\frac {\sqrt {-459+3900\sqrt {3}}}{400}}$ is a solution to $10000x^2+(-1300\sqrt 3-3750)x+507=0$.
The remaining points in the tile are given by rotations by angles $2\pi/3$ and $4\pi/3$ around $p$.}
\label{fig:tiles}
\end{figure}

In order to define the outer objects of a tile, we consider the straight skeleton offset~\cite{Aichholzer1996} of $G_p$ at distance $1/4$.
With the exception of the bend tile, this offset is the same as the Euclidean offset.
By the \emph{outer} and \emph{inner region}, we mean the region of the tile outside, resp.~inside, this offset.
The outer objects cover the outer region, and every point is colored with the outer color of a closest edge in $G_p$.
The inner region is empty except for the inner objects described in each case below.
We suppose that $p=(0,0)$.

\textbf{The straight tile.}
If two collinear edges meet at $p$ with the same inner and outer color, we use a straight tile.
Suppose in this and the following two cases that $G_p$ is the vertical line segment from $(0,-1/2)$ to $(0,1/2)$---tiles for edges of other slopes are obtained by rotation of the ones described here.
There are four axis-parallel squares of the inner color of $G_p$ with side length $1/8$ centered at $(\pm (1/4-1/16),\pm 1/4)$.
This size is chosen so their total perimeter is $2$, which is the length of the common boundary of the inner and outer regions.

\textbf{The inner color change tile.}
If two collinear edges meet at $p$ with different inner colors, we use an inner color change tile.
There are again four squares colored in the inner color of the closest point in $G_p$.
There are also four smaller axis-parallel squares with side length $1/28$ centered at $(\pm (1/4-1/56),\pm 1/56)$, likewise colored in the inner color of the closest point in $G_p$.
The size of these small squares is chosen so that they can be  individually enclosed using fences of total length $14\cdot 1/28=1/2$, which is the width of the inner region.

\textbf{The outer color change tile.}
If two collinear edges meet at $p$ with different outer colors, we use an outer color change tile.
There are four axis-parallel squares of the inner color of $G_p$ with side length $3/32$.
Their centers are $(\pm (1/4-3/64),\pm 1/4)$.
The size of these squares is chosen so that their total perimeter is $2-1/2=3/2$.

\textbf{The bend tile.}
If two non-collinear edges meet at $p$, we use a bend tile.
Consider the case where $G_p$ is the vertical line segment from $p$ to $(0,1/2)$ and the segment of length $1/2$ from $p$ with direction $(\cos \pi/6,-\sin\pi/6)$.
The other cases are obtained by a suitable rotation of this tile.
There is an axis parallel square of side length $x=\frac{6+\sqrt 3}{72}$ with center $(-(1/4-x/2),1/4)$ and another with side length $y=\frac{6-\sqrt 3}{48}$ centered at $(1/4-y/2, 3/8)$.
The tile is symmetric with respect to the angular bisector $b$ of $G_p$, and so the reflections of the described squares with respect to $b$ are also inner objects.
Note that there are two outer objects, one of which, $O$, has a concave corner $q$ with exterior angle $2\pi/3$.
We place a parallelogram with side length $x$, a corner at $q$, and two edges contained in the edges of $O$ incident at $q$.
It is easy to verify that the common boundary of the inner and outer regions has a total length of $2$; the inner objects are chosen such that their total perimeter is also $2$.

\textbf{The branch tile.}
If $p$ is a branch vertex, we use the branch tile.
There are two cases: $G_p$ either contains the vertical segment from $p$ to $(0,1/2)$ or that from $p$ to $(0,-1/2)$.
We specify the tile in the first case---the other can be obtained by a rotation of $\pi$.
There are axis-parallel squares of side length $y=\frac{6-\sqrt 3}{48}$ centered at $(\pm(1/4-y/2), 3/8)$ and their rotations around $p$ by angles $2\pi/3$ and $4\pi/3$.
The common boundary of the inner and outer regions has a total length of $\frac{6-\sqrt 3}{2}$, and the total perimeter of the inner objects is also $\frac{6-\sqrt 3}{2}$.

\textbf{The clause tile.}
If $p$ is a clause vertex, we use the clause tile (defined in Figure~\ref{fig:tiles}).
The other clause tiles are given by rotations of the described tile by angles $k\pi/3$ for $k=1,\ldots,5$.

\subsection{Solving the tiles}

Let an instance $I$ of GEOMETRIC $3$-SAT be given together with an associated fence $\mathcal F$.
Consider the restriction of $I$ to a convex polygon $P$ and the part
of the fence $\mathcal F\cap P$ inside $P$.
Note that $\mathcal F\cap P$ consists of (not necessarily disjoint) closed curves and open curves with endpoints on the boundary $\partial P$, such that no two objects in $P$ of different color can be connected by a path $\pi\subset P$ unless $\pi$ intersects $\mathcal F$.
(An open curve is a subset of a larger closed curve of $\mathcal F$ that continues outside $P$.)
We say that a set of closed and open curves in $P$ with that property is a \emph{solution} to $I\cap P$.
In the following, we analyze the solutions to the tiles defined in Section~\ref{sec:tileDefs} in order to characterize the solutions of minimum cost.
We say that two closed curves (disjoint from the interiors of the objects) are \emph{homotopic} if one can be continuously deformed into the other without entering the interiors of the objects.
Two open curves with endpoints on the boundary of the tile are homotopic if they are subsets of two homotopic closed curves (that extends outside the tile).

\begin{figure}
\centering
\includegraphics[width=0.24\textwidth]{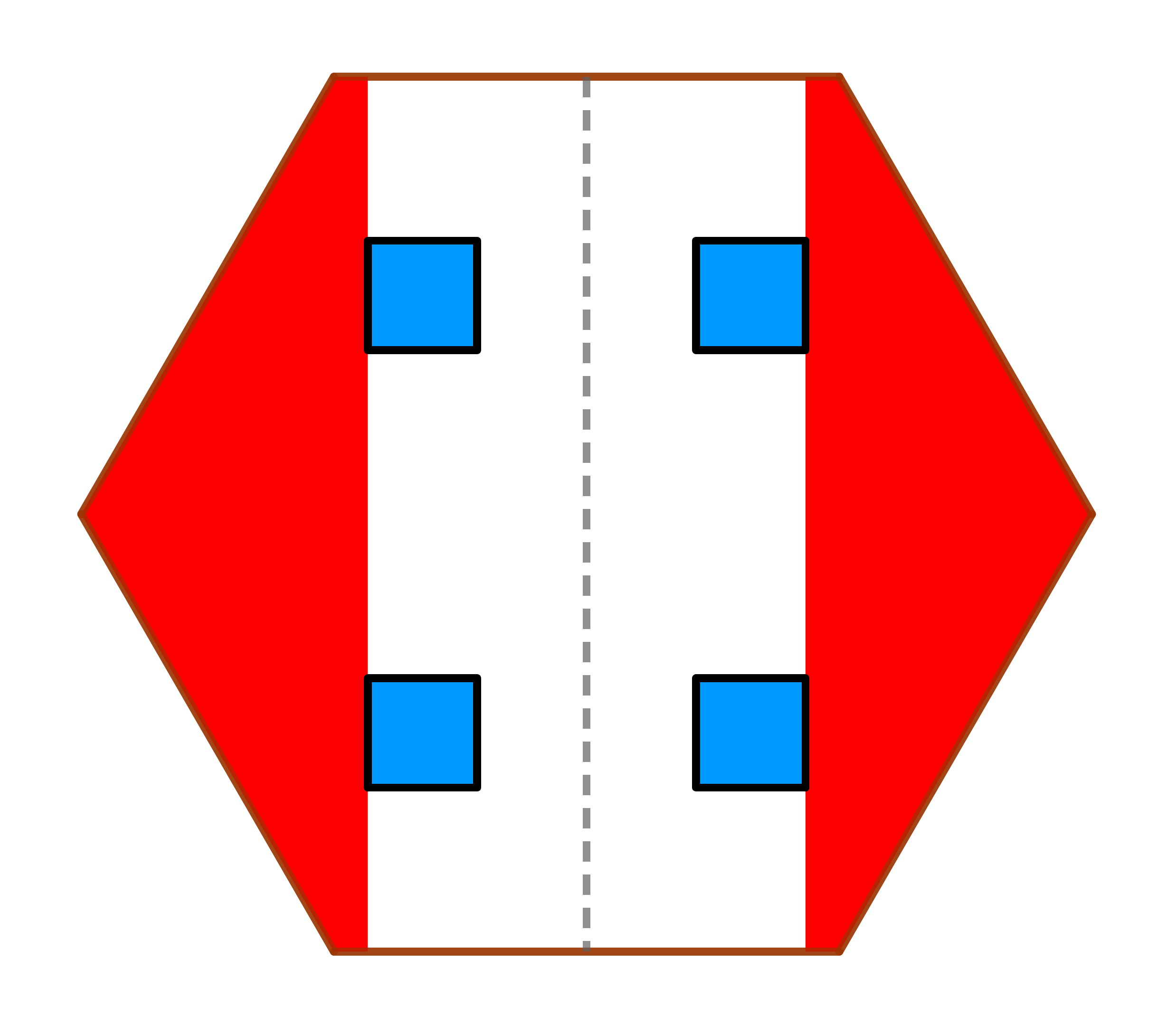}
\includegraphics[width=0.24\textwidth]{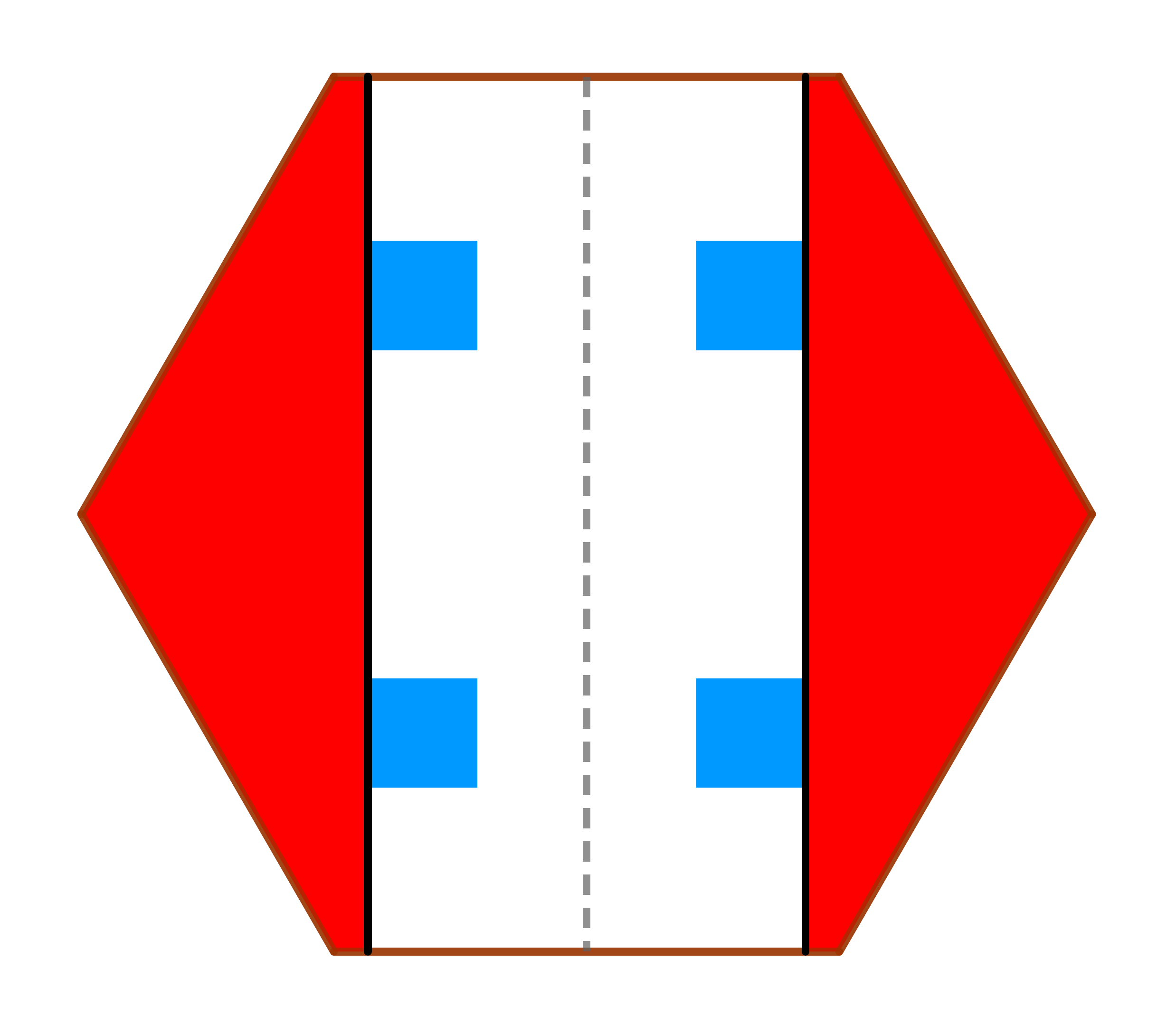}
\includegraphics[width=0.24\textwidth]{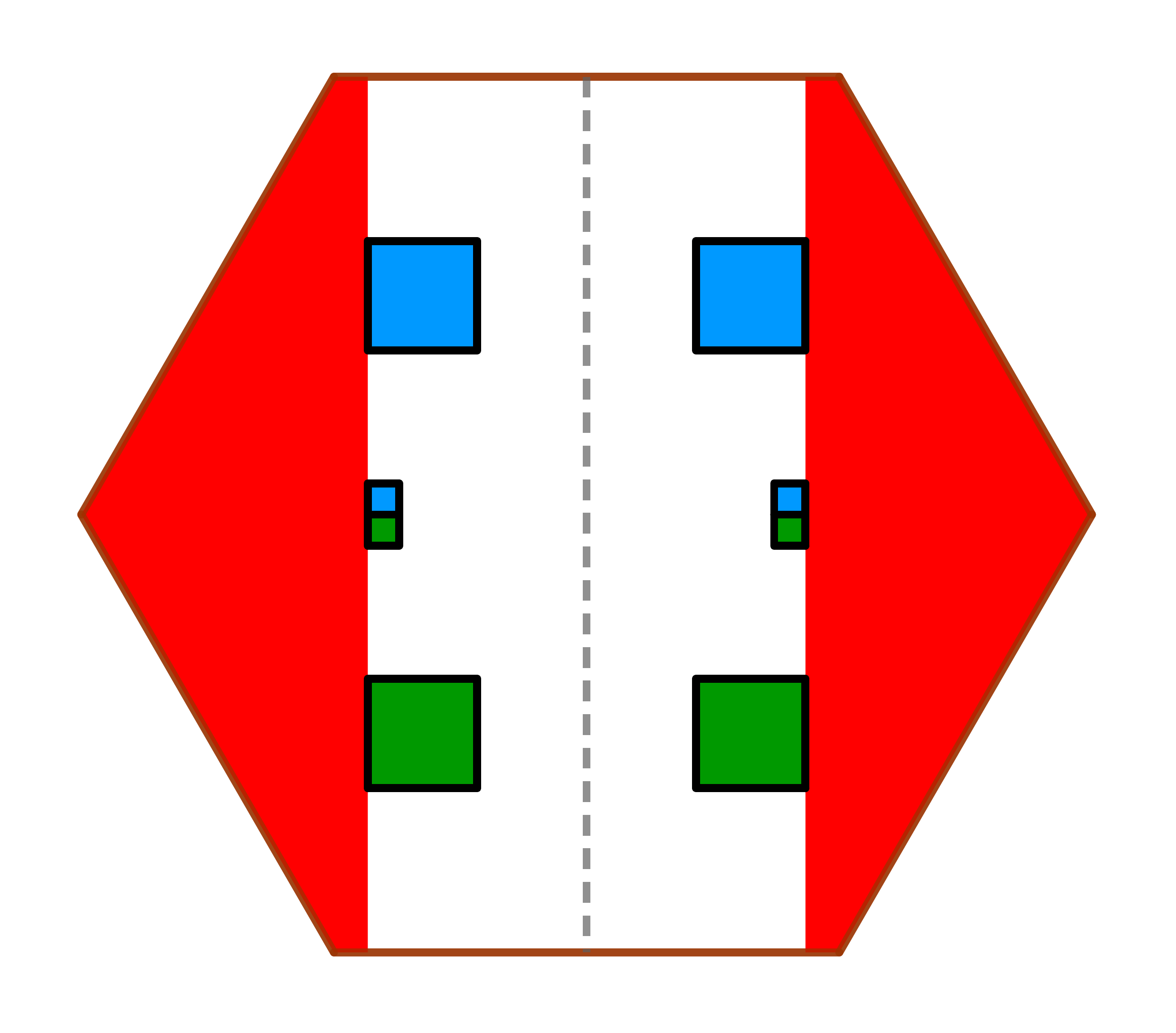}
\includegraphics[width=0.24\textwidth]{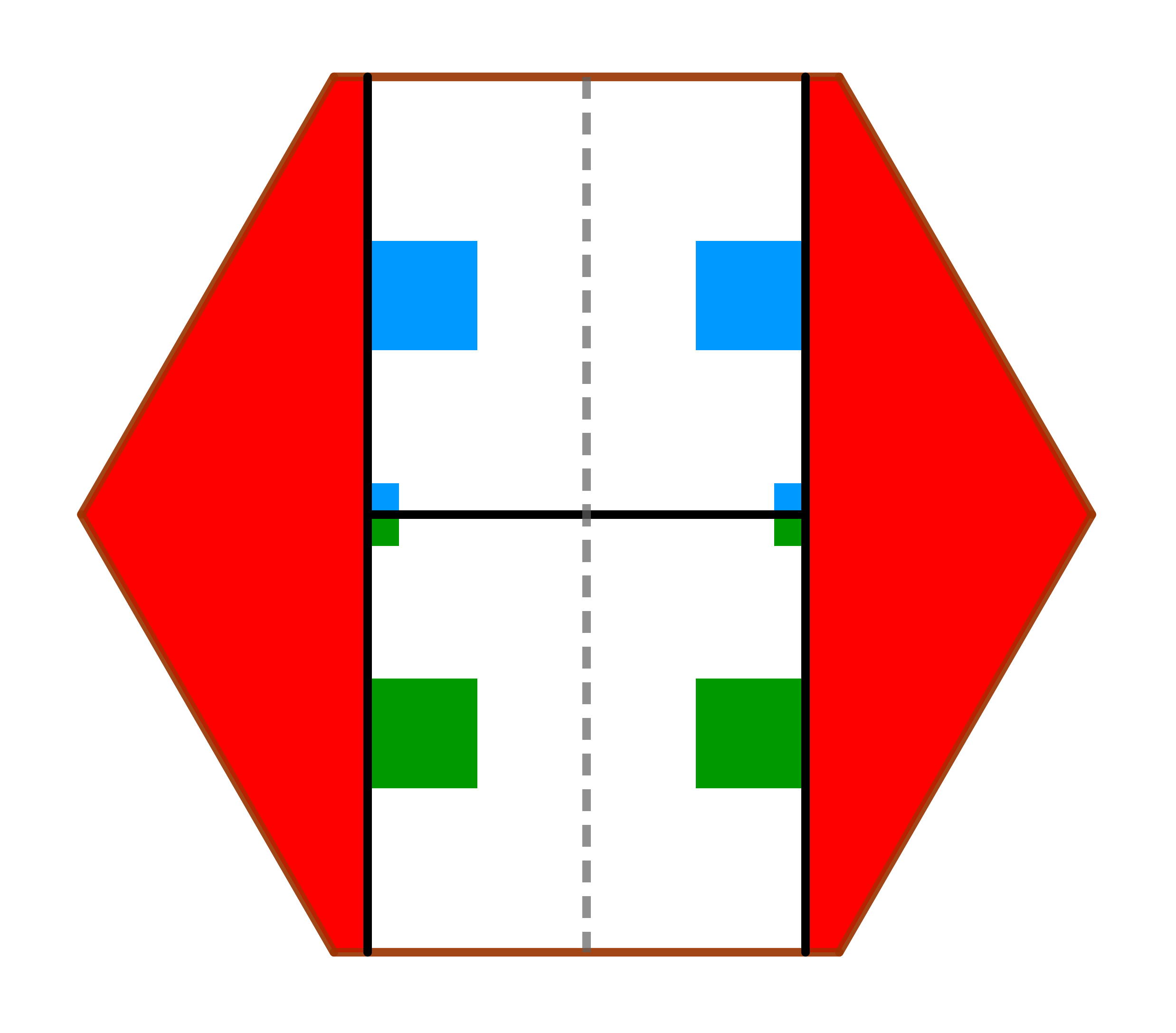}
\\
\includegraphics[width=0.24\textwidth]{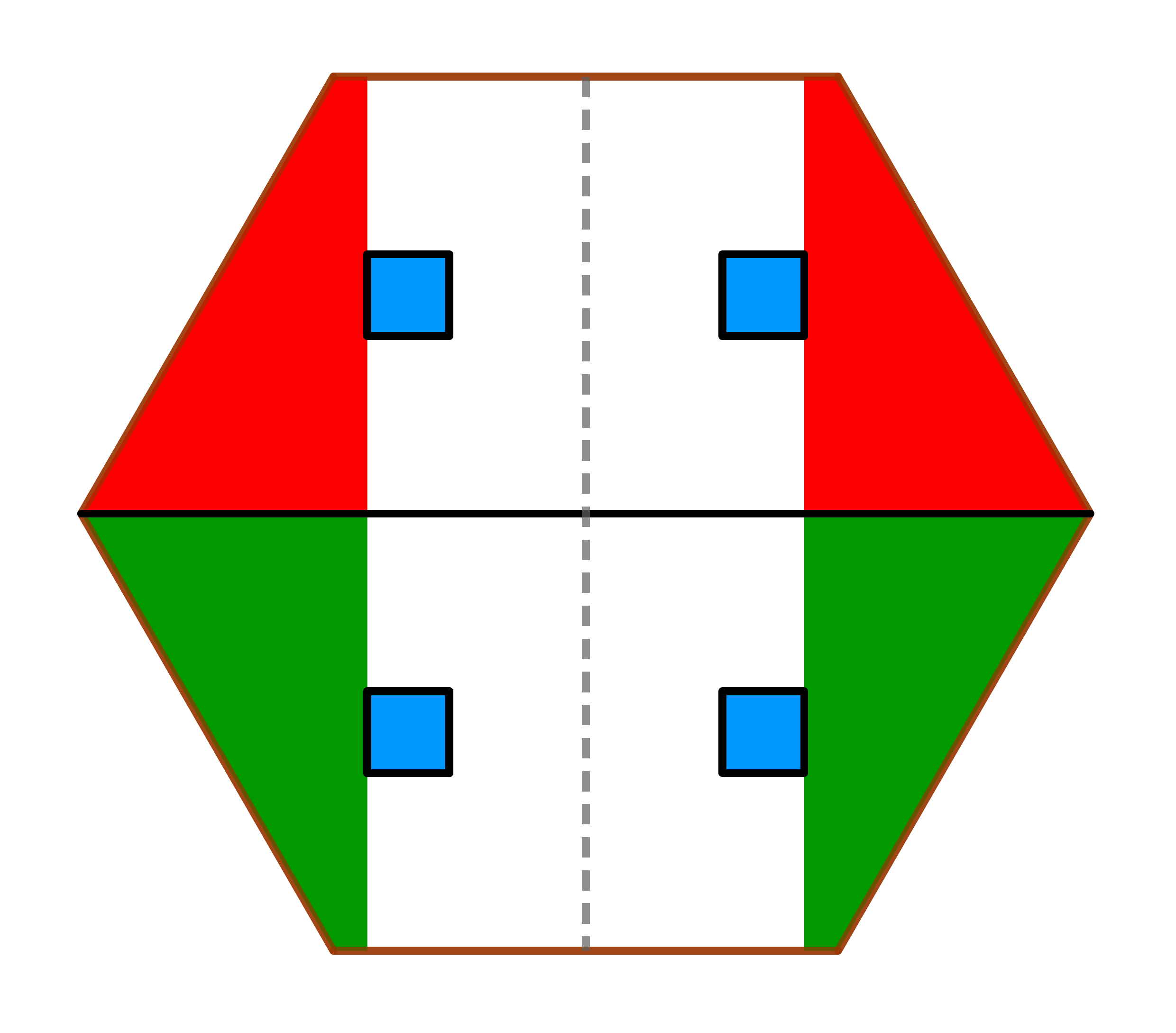}
\includegraphics[width=0.24\textwidth]{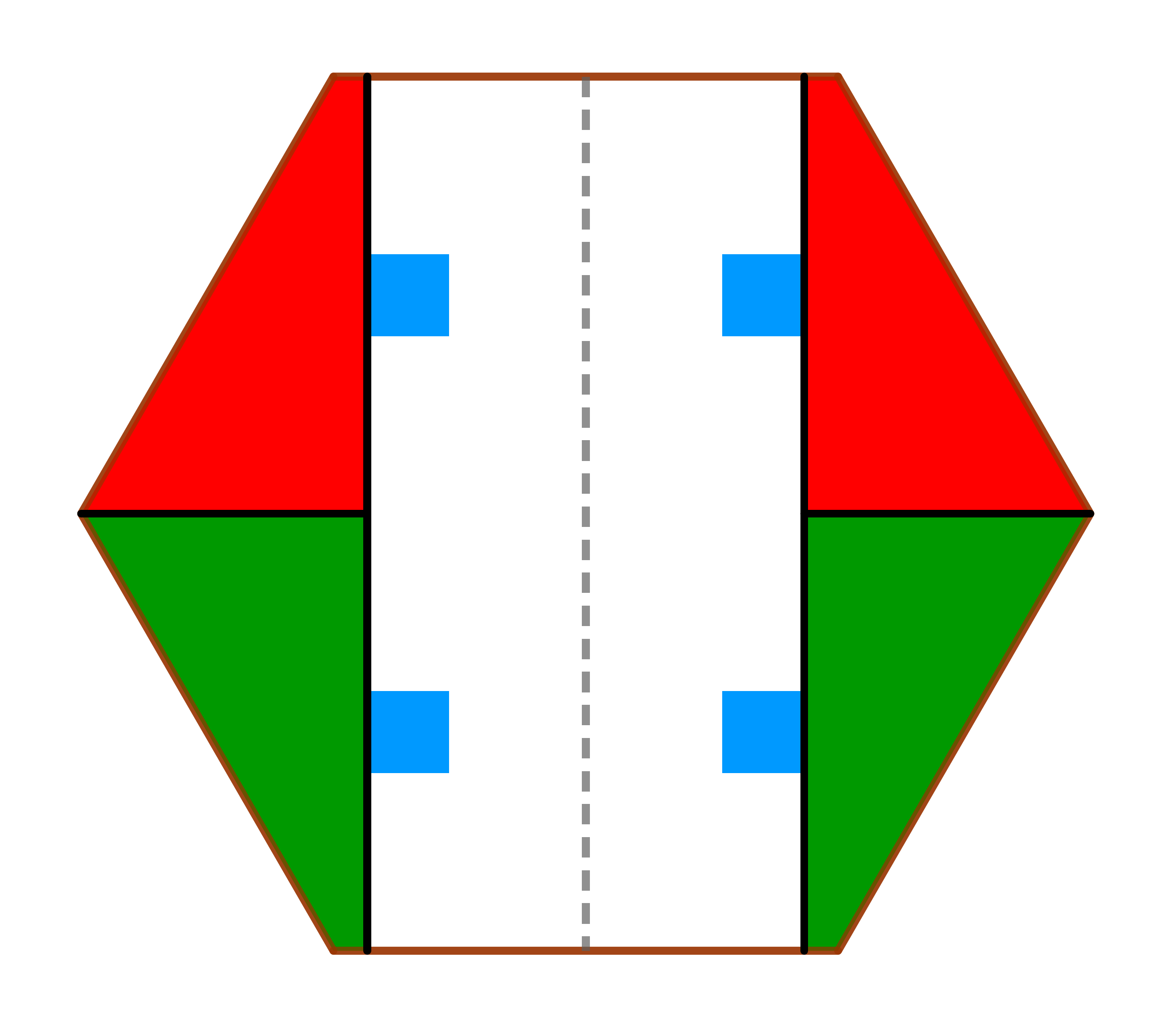}
\includegraphics[width=0.24\textwidth]{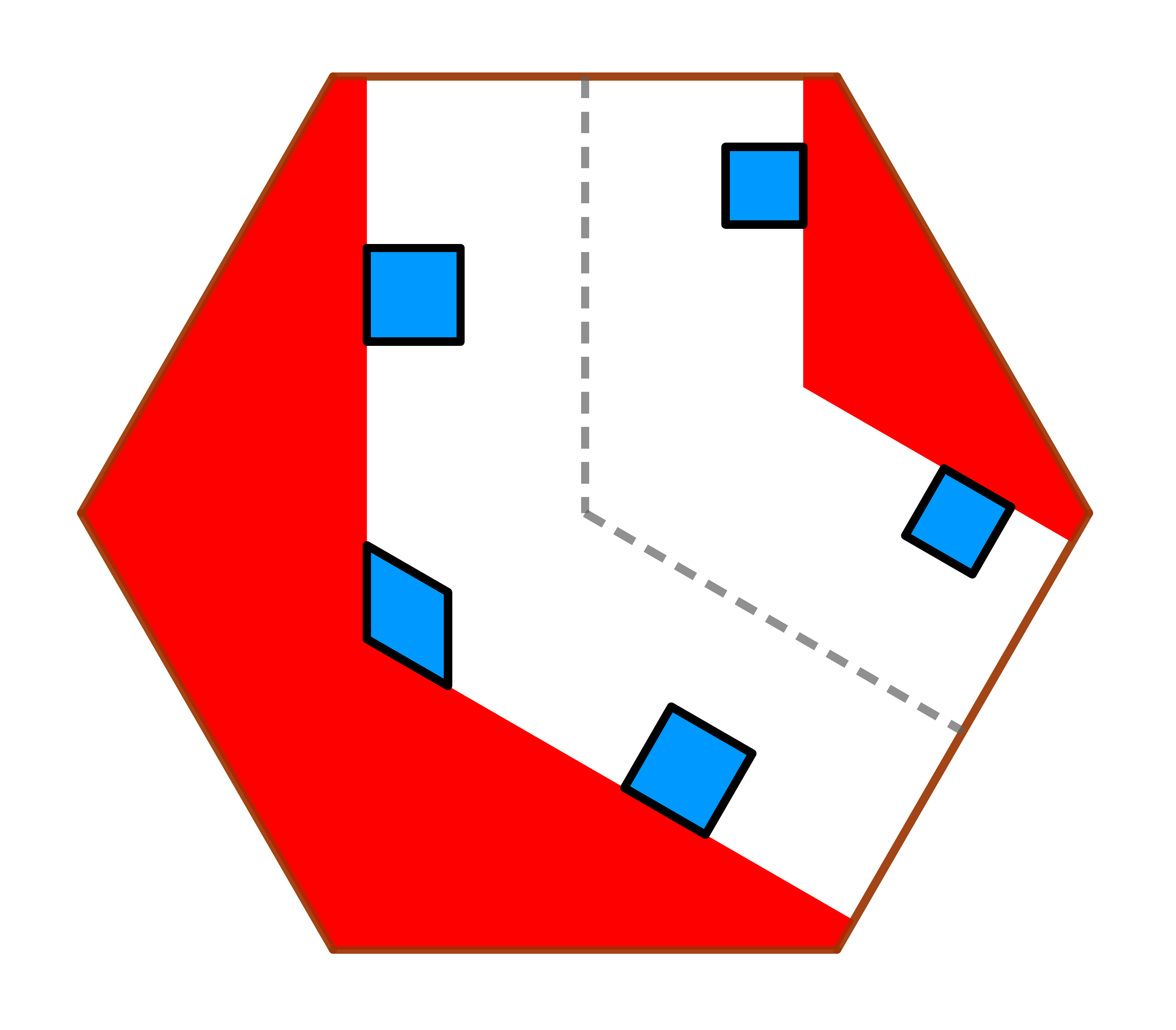}
\includegraphics[width=0.24\textwidth]{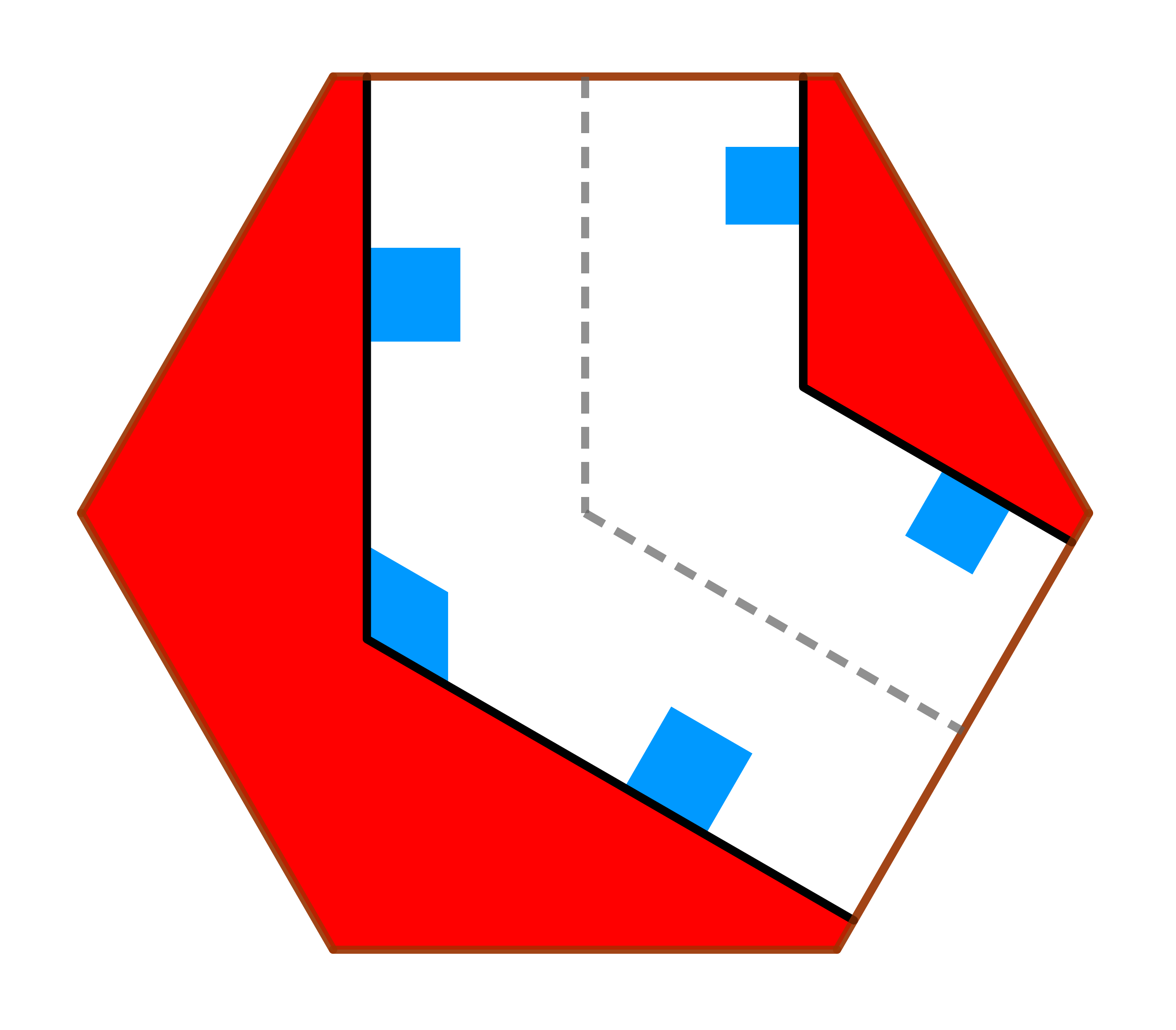}
\\
\includegraphics[width=0.24\textwidth]{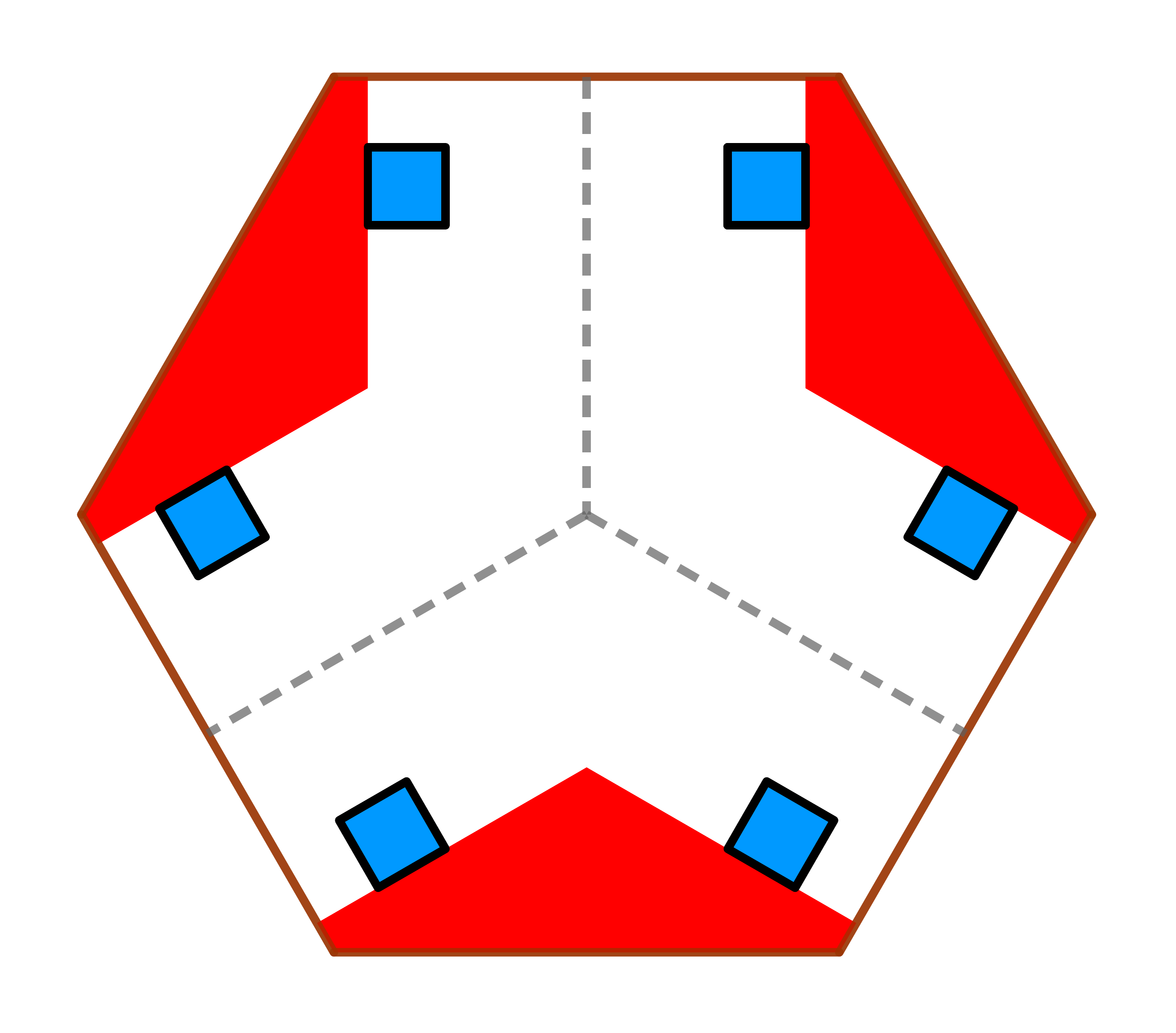}
\includegraphics[width=0.24\textwidth]{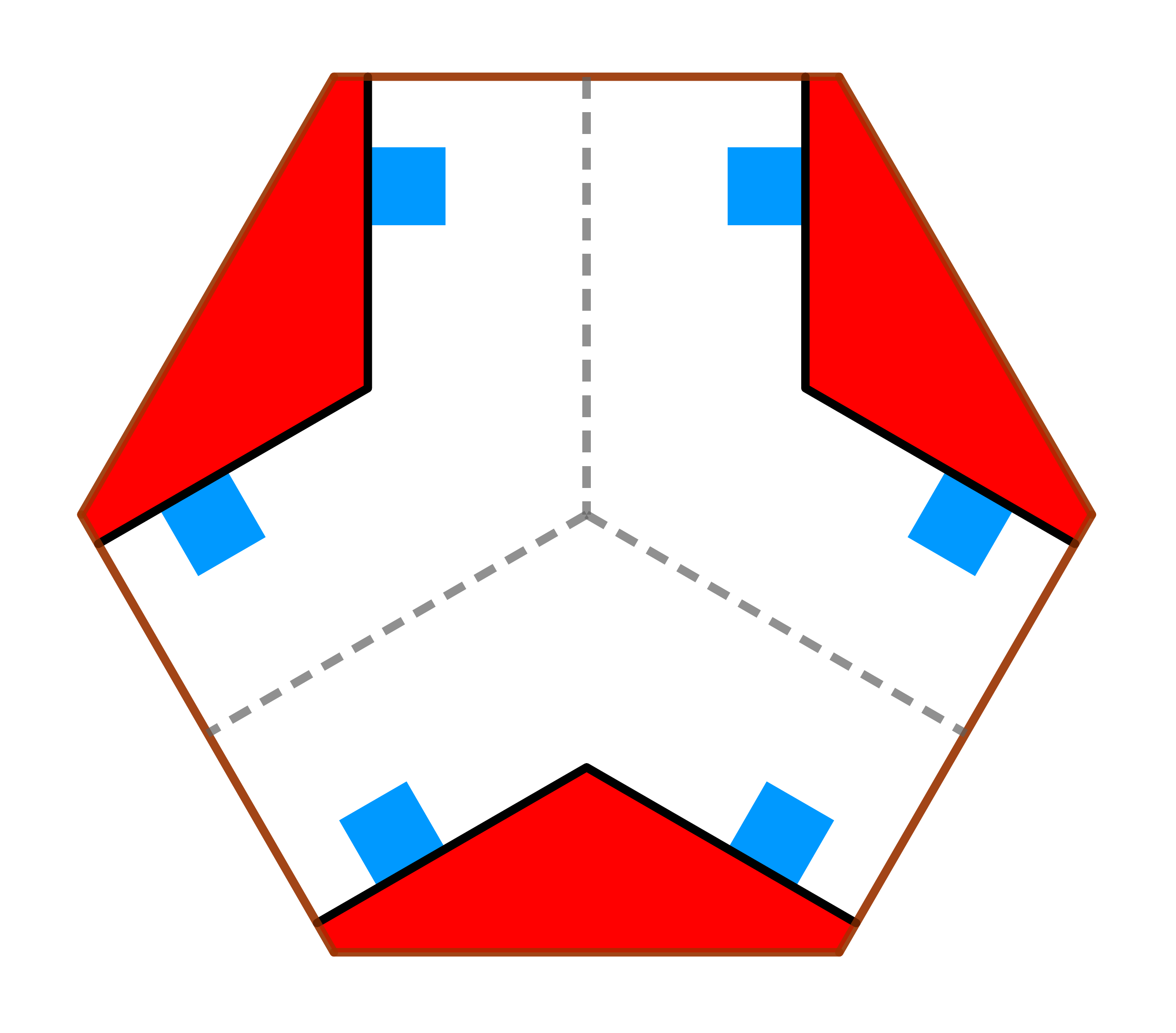}
\\
\includegraphics[scale=0.75]{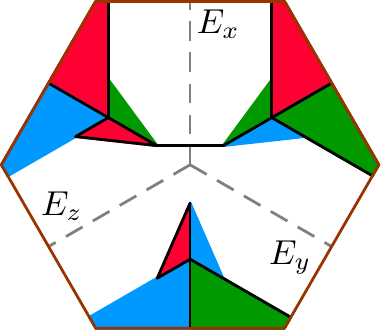}
\quad
\includegraphics[scale=0.75]{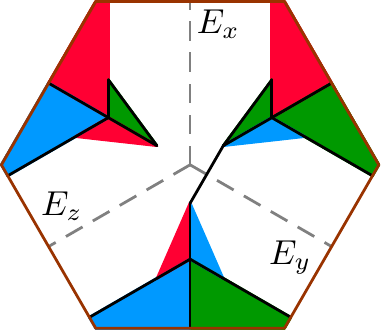}
\quad
\includegraphics[scale=0.75]{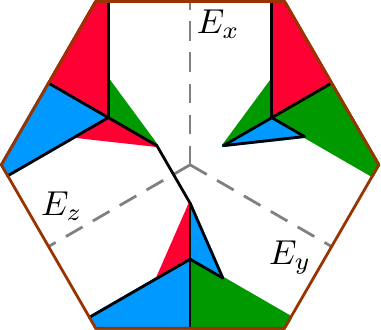}
\caption{The optimal solutions to each type of tile.
The edges in $G_p$ are shown in dashed grey.
We denote the left solution of each of the first five types of tiles as the \emph{outer} solution and the other as the \emph{inner} solution.
For the clause tile, we define the solution as the \emph{$z$-outer}, \emph{$x$-outer}, and \emph{$y$-outer} solution in order from left to right, respectively.
}
\label{fig:solved_states}
\end{figure}

\begin{lemma}\label{lemma:solved:states}
Figure~\ref{fig:solved_states} shows optimal solutions to each kind of tile.
The cost in each case is:
\iffull
\begin{itemize}
\item
Straight tile: $2$.

\item
Inner color change tile: $5/2$.
\item
Outer color change tile: $\left(\frac 2{\sqrt 3}-\frac 12\right)+2\approx 2.65$.
\item
Bend tile: $2$.
\item
Branch tile: $\frac{6-\sqrt 3}2\approx 2.13$.
\item
  Clause tile: $\approx 3.51$.\footnote{The exact value is complicated due to the coordinates and of no use.}
\end{itemize}
\fi

If the cost of a solution $\mathcal F$ to a tile $T$ exceeds the optimum by less than $1/50$, then $\mathcal F$ is homotopic to one of the optimal solutions $\mathcal F^*$ of $T$ in the following sense:
For each curve $\pi^*$ in $\mathcal F^*$, there is a curve $\pi$ in $\mathcal F$ homotopic to $\pi^*$.
If $\pi$ is closed, the distance from any point on $\pi$ to the closest point on $\pi^*$ is less than $\sqrt{(1/8+1/100)^2-(1/8)^2}<0.06$.
If $\pi$ is open and $\pi^*$ has an endpoint $f^*$, there is a corresponding endpoint $f$ of $\pi$ with $\|f^*f\|< 1/10$.
\end{lemma}

\begin{proof}
We assume that the center of the tile $T$ is $p\mydef (0,0)$ and in each case, that $T$ is rotated as in the description of the tiles and have colors as in Figure~\ref{fig:tiles}.
We also define $G_p$ to be the two or three half-edges of $G(\Phi)$ meeting at $p$ as in the description of the tiles.

Consider first the case that $T$ is any of the tiles except the clause tile.
Note that $G_p$ separates $T$ into two or three pieces.
The pieces are two pentagons for the straight, inner color change, and outer color change tiles; a pentagon and a heptagon for the bend tile; and three pentagons for the branch tile.
We consider each such piece $T'$ individually and check the minimum cost of a solution to $T'$.
It is easy to verify that for each such piece $T'$, there are two solutions, and they are exactly as shown in Figure~\ref{fig:solved_states}.
One solution corresponds to the outer state and the other to the inner state, and in order to be combined to a solution for all of $T$, each of the two or three pieces $T'$ needs to be in the same state.
It therefore follows that the solutions shown in Figure~\ref{fig:solved_states} are all the optimal solutions.

\begin{figure}
\centering
\includegraphics[width=0.45\textwidth]{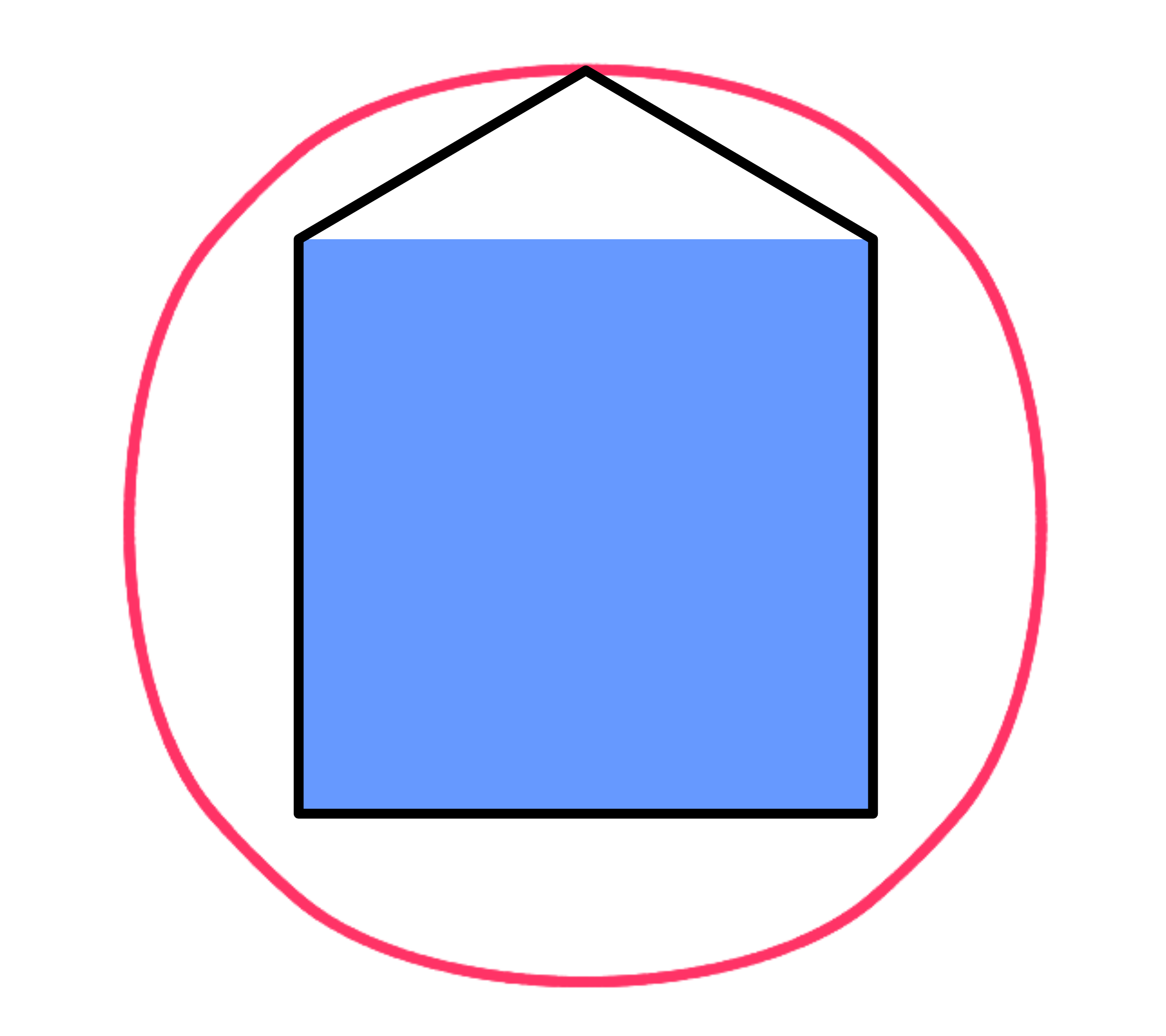}
\includegraphics[width=0.45\textwidth]{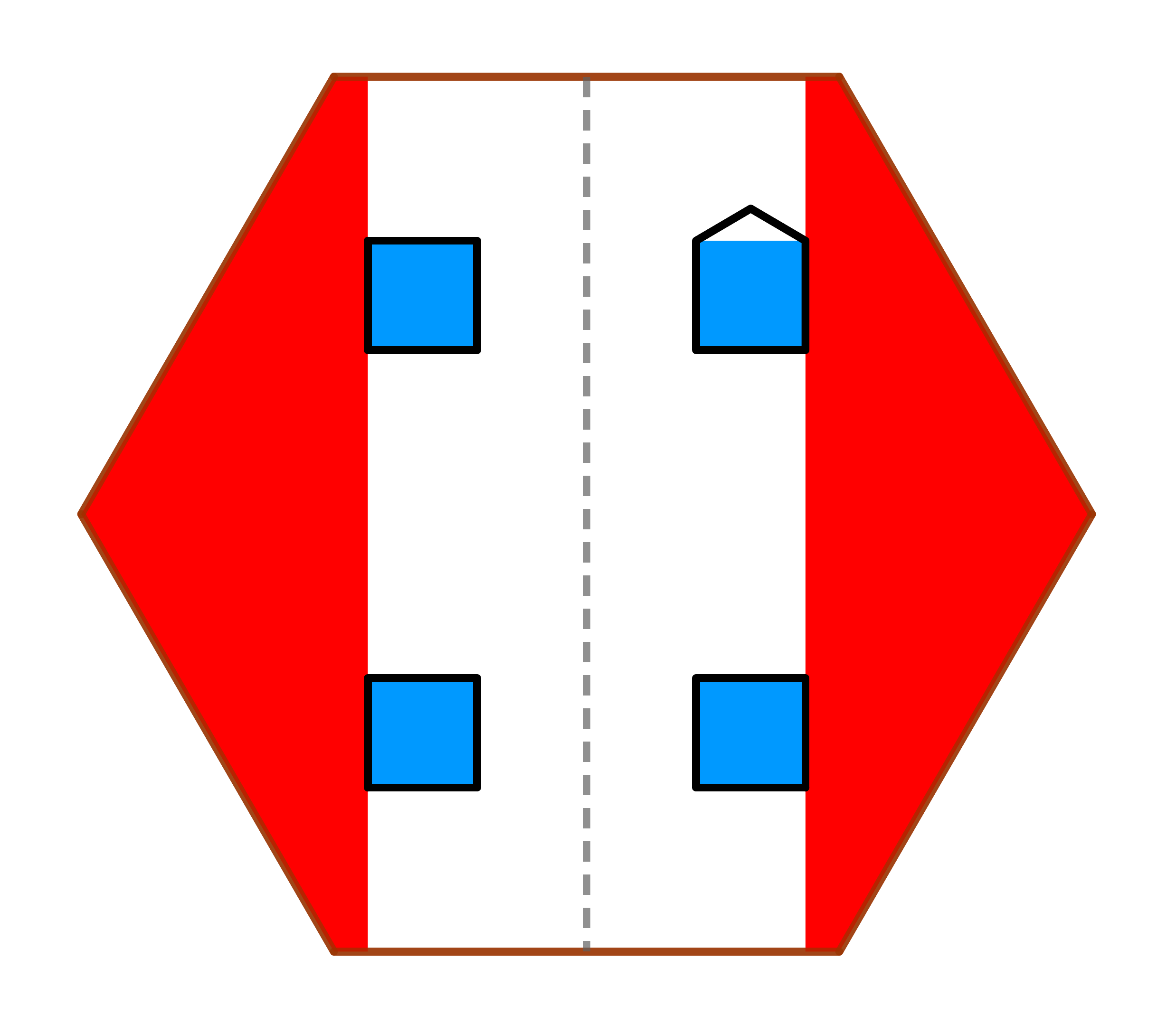}
\includegraphics[width=0.45\textwidth]{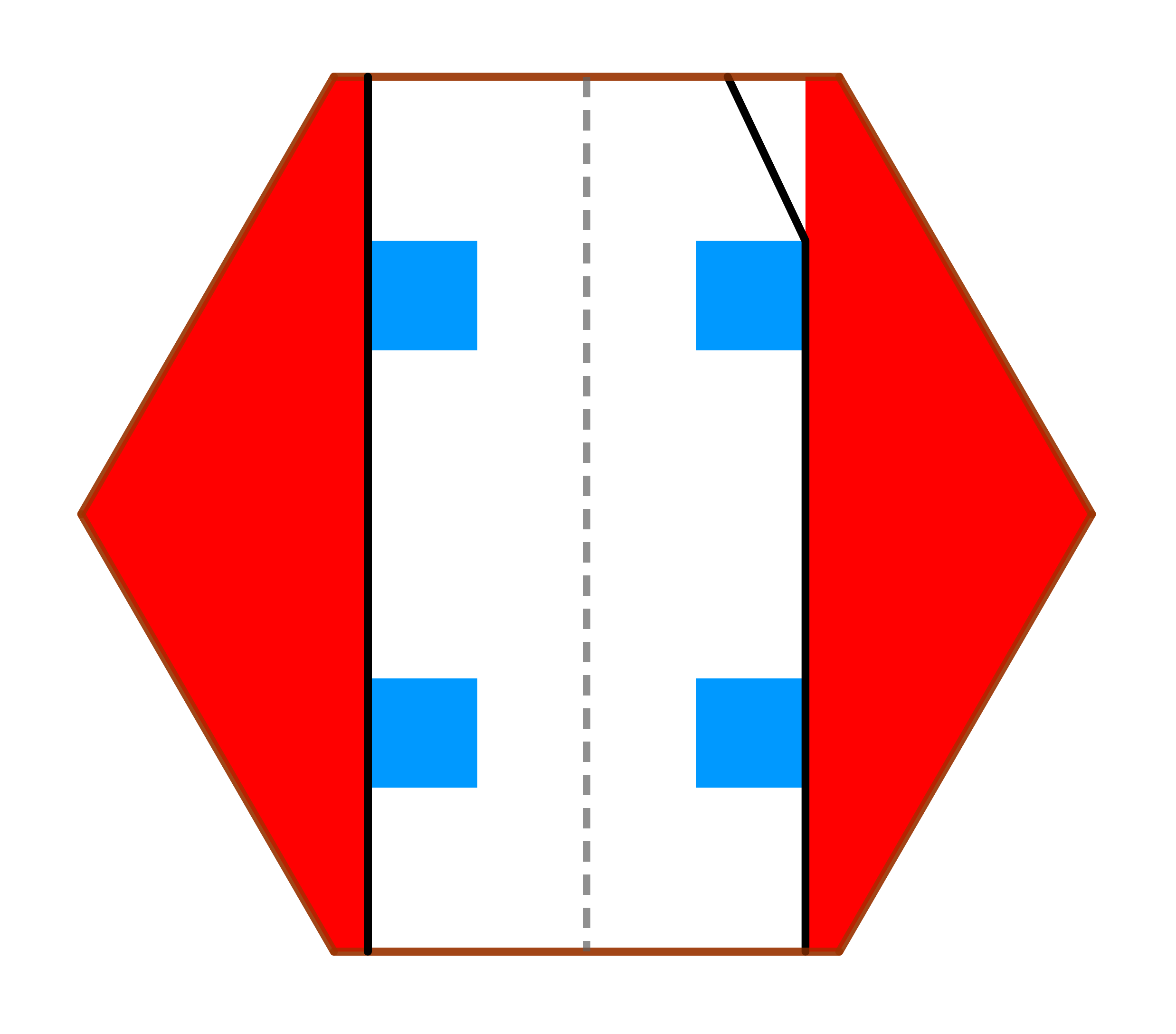}
\caption{Left:
A square of side length $1/8$.
The red curve encloses all curves of length at most $1/2+1/50$ that enclose the square.
Such a curve with maximum deviation from the boundary of the square is drawn in black.
The red curve consists of a piece of each of eight different ellipses.
Middle resp.~right:
A solution to the straight tile in the outer resp.~inner state with a cost that exceeds the optimum by $1/50$.}
\label{fig:max_dev}
\end{figure}

One can also easily verify that any solution $\mathcal F$ that is not homotopic to an optimal solution has a cost that exceeds the optimal cost by more than $1/50$.
Suppose that the cost of a solution $\mathcal F$ exceeds the cost of a homotopic optimal solution $\mathcal F^*$ by less than $1/50$.
In order to decide how much $\mathcal F$ can deviate from $\mathcal F^*$, consider the straight tile as an example, see Figure~\ref{fig:max_dev}.
In the outer state, each curve enclosing an inner object has length at least $1/2$.
Since the total cost is less than $2+1/50$, each curve has length less than $1/2+1/50$.
An elementary analysis gives that a closed curve of length at most $1/2+1/50$ which encloses a square of side length $1/8$ is within distance $\sqrt{(1/16+1/100)^2-(1/16)^2}<0.04$ from the boundary of the square.
For the inner state, consider the curve $\pi\subset\mathcal F$ in the right side of the tile that has the inner objects to the left.
The length of $\pi$ has to be less than $1+1/50$ in order for the total cost to be less than $2+1/50$.
Note that $\pi$ has to pass through the upper right corner $(1/4,5/16)$ of the upper right square.
Therefore, $\pi$ has to meet the top edge of $T$ at a point within distance $\sqrt{(3/16+1/50)^2-(3/16)^2}<0.09$ from the corresponding endpoint $(1/4,1/2)$ of $\pi^*$.
The other non-clause tiles are analyzed in a similar way.

The analysis of the clause tile is not as simple, since one does not get a solution to the complete tile by combining optimal solutions of smaller pieces.
The proof has been deferred to
\iffull
Lemma~\ref{lemma:clause:analysis}
\fi
and relies on an extensive case analysis.

The largest possible deviation between a closed curve in $\mathcal F$ and $\mathcal F^*$ can appear for the clause tile, since it contains an inner object with the longest edge of all tiles, namely a triangle with an edge of length $1/4$.
That leads to a deviation of less than $\sqrt{(1/8+1/100)^2-(1/8)^2}<0.06$.
Likewise, the largest possible deviation between open curves is $1/10$, as realized in the clause tile and described in
\iffull
Lemma~\ref{lemma:clause:analysis}.
\fi
\end{proof}

\iffull
We now analyze the optimal solutions of the clause tile.
Here, we define a \emph{domain} as a connected component of a territory.
The names of objects in the clause tile referred to in the following lemma are defined in Figure~\ref{fig:case1}.
Indices are taken modulo $3$.
The optimal solutions are covered by case~2.3.2 in the proof.

\begin{figure}[bht]
\centering
    \includegraphics[width=0.45\textwidth]
    {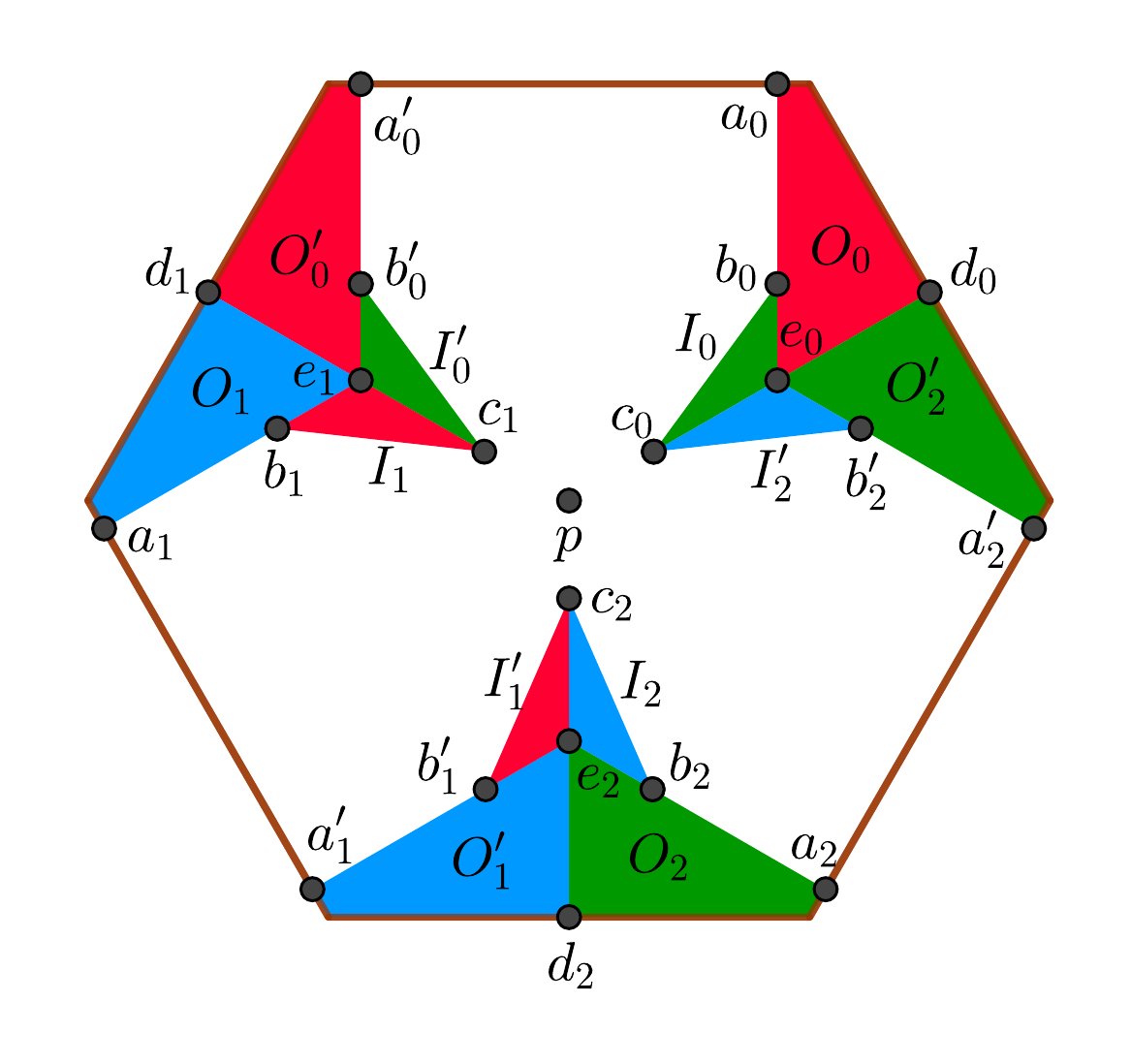}
    \includegraphics[width=0.45\textwidth]
    {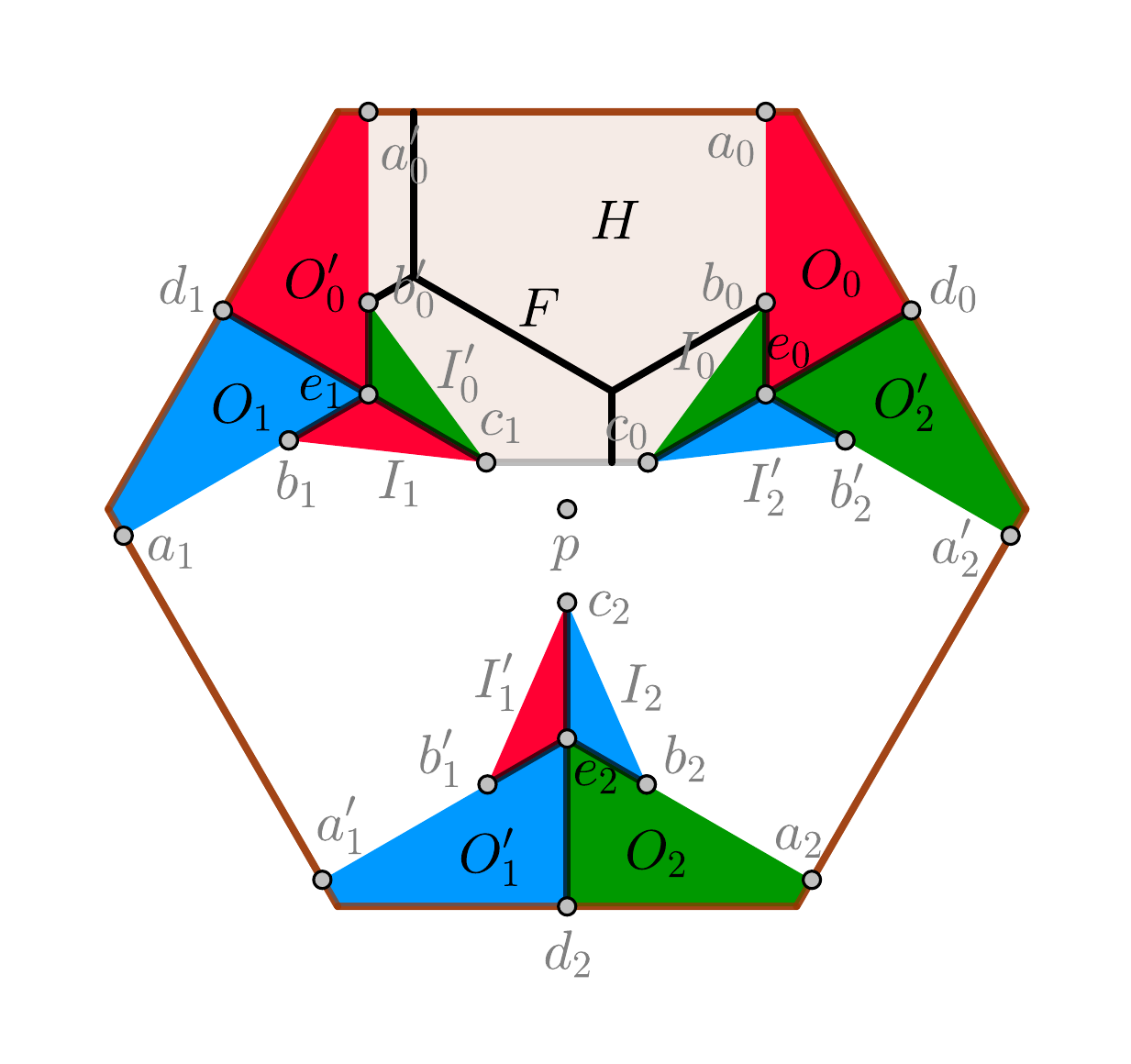} \\
    \includegraphics[width=0.45\textwidth]
    {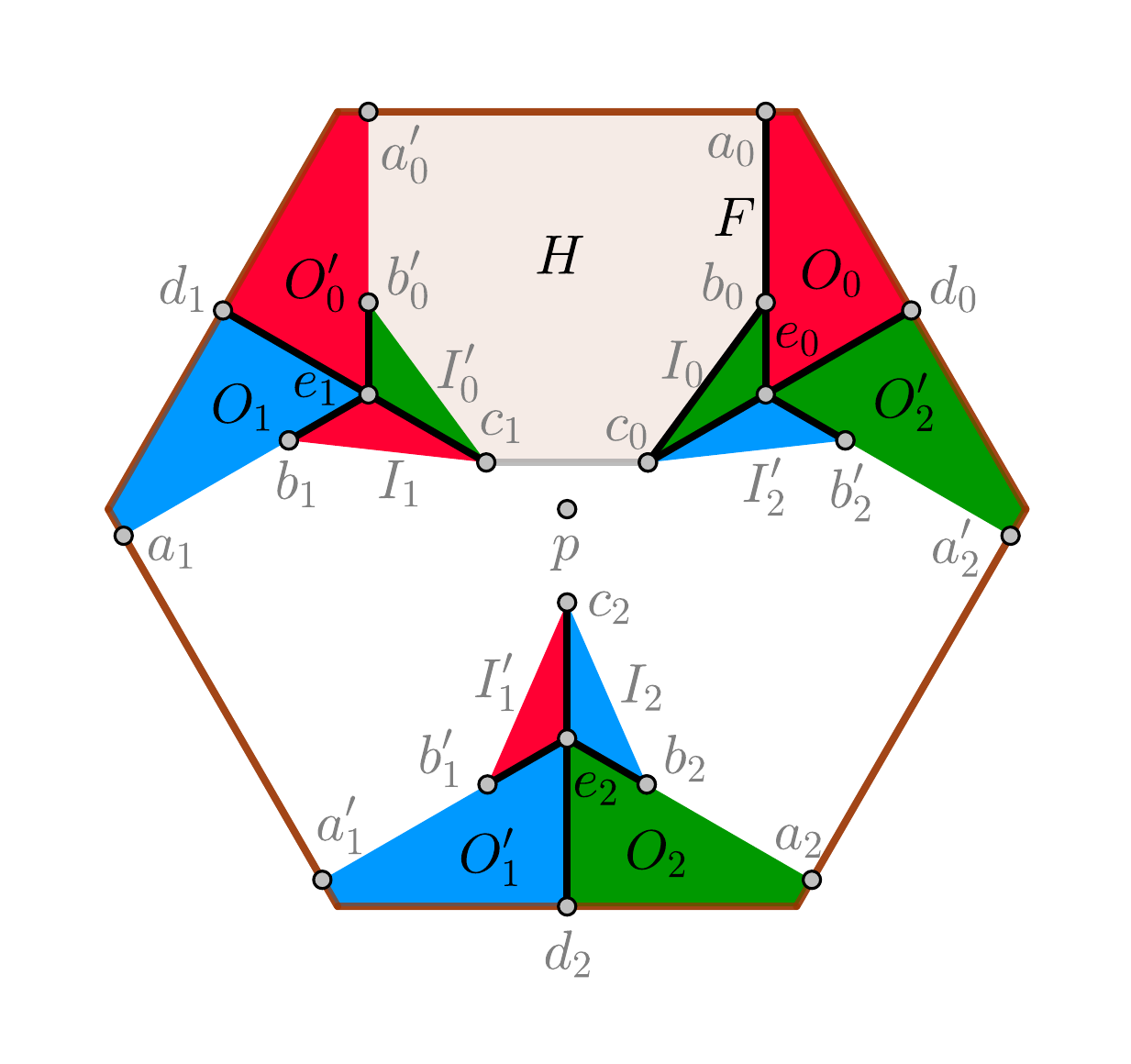}
    \includegraphics[width=0.45\textwidth]
    {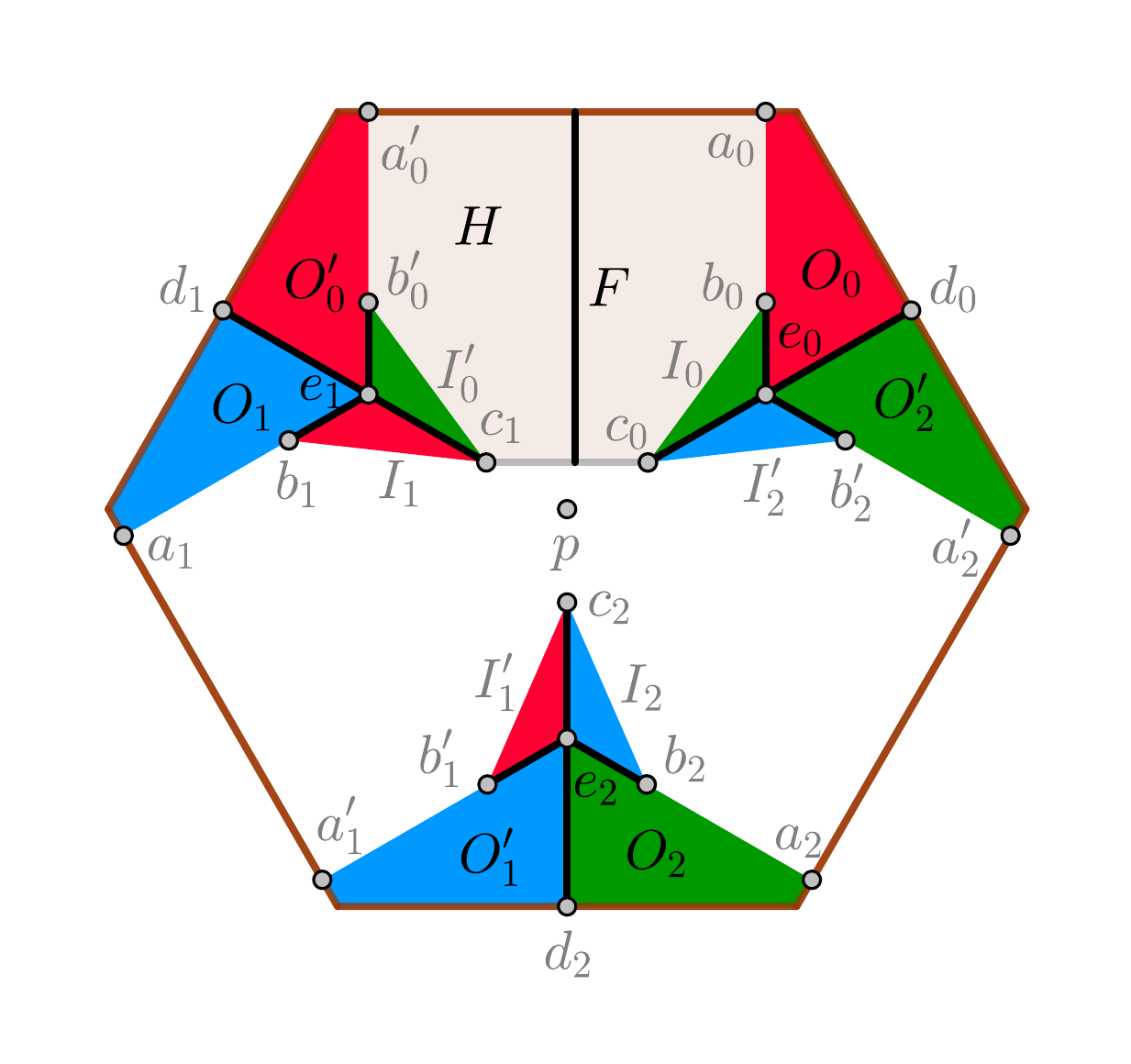}
    \caption{Top: Annotations. Case~1.1. Bottom: Case~1.2. Case~1.3.}
    \label{fig:case1}
\end{figure}

\begin{lemma}\label{lemma:clause:analysis}
The cost of a solution $\mathcal F$ to a clause tile is at least $M:= 3\|d_0c_0\|+6\|e_0b_0\|+4\|a_0b_0\|+2\|b_0c_0\|+\|c_0c_1\|$.
Moreover, if the cost is less than $M+1/50$, then there is $i\in\{0,1,2\}$ such that $\mathcal F$ contains
\begin{itemize}
\item
a curve from $f_i\in a_ia'_i$ to $b_i$, where $\|f_ia_i\|<\sqrt{(6/25+1/50)^2-(6/25)^2}=1/10$,

\item
a curve from $f'_i\in a_ia'_i$ to $b'_i$, where $\|f'_ia'_i\|<1/10$,

\item
a curve from $f_{i+1}\in a_{i+1}a'_{i+1}$ to $b_{i+1}$, where $\|f_{i+1}a_{i+1}\|<1/10$,

\item
a curve from $f'_{i+1}\in a_{i+1}a'_{i+1}$ to $b'_{i+1}$, where $\|f'_{i+1}a'_{i+1}\|<1/10$,

\item
a curve from $c_{i+1}$ to $c_{i+2}$,

\item
a curve from $c_i$ to $b'_{i+2}$,

\item
a curve from $c_{i+2}$ to $b_{i+2}$.
\end{itemize}
\end{lemma}

\begin{proof}
Clearly, $\mathcal F$ must contain segments $d_ic_i$, $e_ib_i$, and $e_ib'_{i+2}$ for any $i\in\{0,1,2\}$ since each of these segments are on the shared boundary of two objects of different color.
In total, this amounts for the cost $3\|d_0c_0\|+6\|e_0b_0\|$.
In the following, we argue about the fence needed in addition to that, i.e., the part of $\mathcal F$ contained in the closed pentadecagon $T'=a_0a'_0b'_0c_1b_1a_1a'_1b'_1c_2b_2a_2a'_2b'_2c_0b_0$.
We characterize how the solution looks when the additional cost in $T'$ is at most $4\|a_0b_0\|+2\|b_0c_0\|+\|c_0c_1\|+0.02<1.67$.
When we say that the solution must contain a curve or a tree with certain properties (such as connecting two specific points), we mean such a curve or tree contained in $T'$.

We divide into cases after which objects are in the same domains.
After making enough assumptions in one branch of the case analysis, we can either give a lower bound of the cheapest solution above $1.67$, or it follows that all objects of different colors are separated, and then we state what the optimal solution satisfying the specific assumptions is.
We have used Geogebra~\cite{hohenwarter2018geogebra} to construct the solutions and estimate their costs.
Each estimate differs from the exact value by at most $0.005$.

Note first that for any $i\in\{0,1,2\}$, in order to separate $O_i$ from $I_i$, the solution must contain a curve (in $T'$) starting at $b_i$ that has a length of at least $0.24$, and similarly one from $b'_i$ in order to separate $O'_i$ from $I'_i$.
The prefixes of length $0.24$ of these six fences are disjoint.
We therefore charge $0.24$ to each $b_i$ and $b'_i$.

\paragraph*{Case 1: For a value of $i$, neither $O_i$ and $O'_i$ nor $I_i$ and $I'_i$ are in a domain together.}
Without loss of generality, suppose that the condition holds for $i=0$.
We consider the solution $\mathcal F$ restricted to the hexagon $H=a_0a'_0b'_0c_1c_0b_0$.
In order to separate $O_0$ and $O'_0$ in $H$, there must be a connected component $F$ of $\mathcal F\cap H$ that connects $a_0a'_0$ and $c_0c_1$.
The individual cases are shown in Figure~\ref{fig:case1}.

\paragraph*{Case 1.1: $F$ also separates $O_0$ from $I_0$ and $O'_0$ from $I'_0$.}
Note that $F$ connects $b_0$ and $b'_0$.
The shortest connected system of curves that connects $b_0$, $b'_0$, $a_0a'_0$ and $c_0c_1$ is a Steiner minimal tree with vertical edges meeting $a_0a'_0$ and $c_0c_1$.
All such trees have the same cost, which is $\approx 0.87$.
But adding the $0.24$ charged to $b_1,b'_1,b_2,b'_2$, we get more than $1.67$.

\paragraph*{Case 1.2: $F$ separates $O_0$ from $I_0$, but not $O'_0$ from $I'_0$.}
(The case where $F$ separates $O'_0$ from $I'_0$, but not $O_0$ from $I_0$, is analogous.)
Note that $F$ has minimal length if $F=a_ib_i\cup b_ic_i$, which has length $\approx 0.49$.
In addition to that comes $1.2$ charged to $b'_i,b_{i+1},\ldots$, and the total is $1.69>1.67$.

\paragraph*{Case 1.3: $F$ separates neither $O_0$ from $I_0$, nor $O'_0$ from $I'_0$.}
In this case, $F$ has cost at least $\approx 0.44$, which is the distance between $a_ia'_i$ and $c_ic_{i+1}$.
In addition comes $1.44$ charged to $b_i,b'_i,\ldots$, which in total is more than $1.67$.

\begin{figure}[htb]
    \centering
    \includegraphics[width=0.45\textwidth]
    {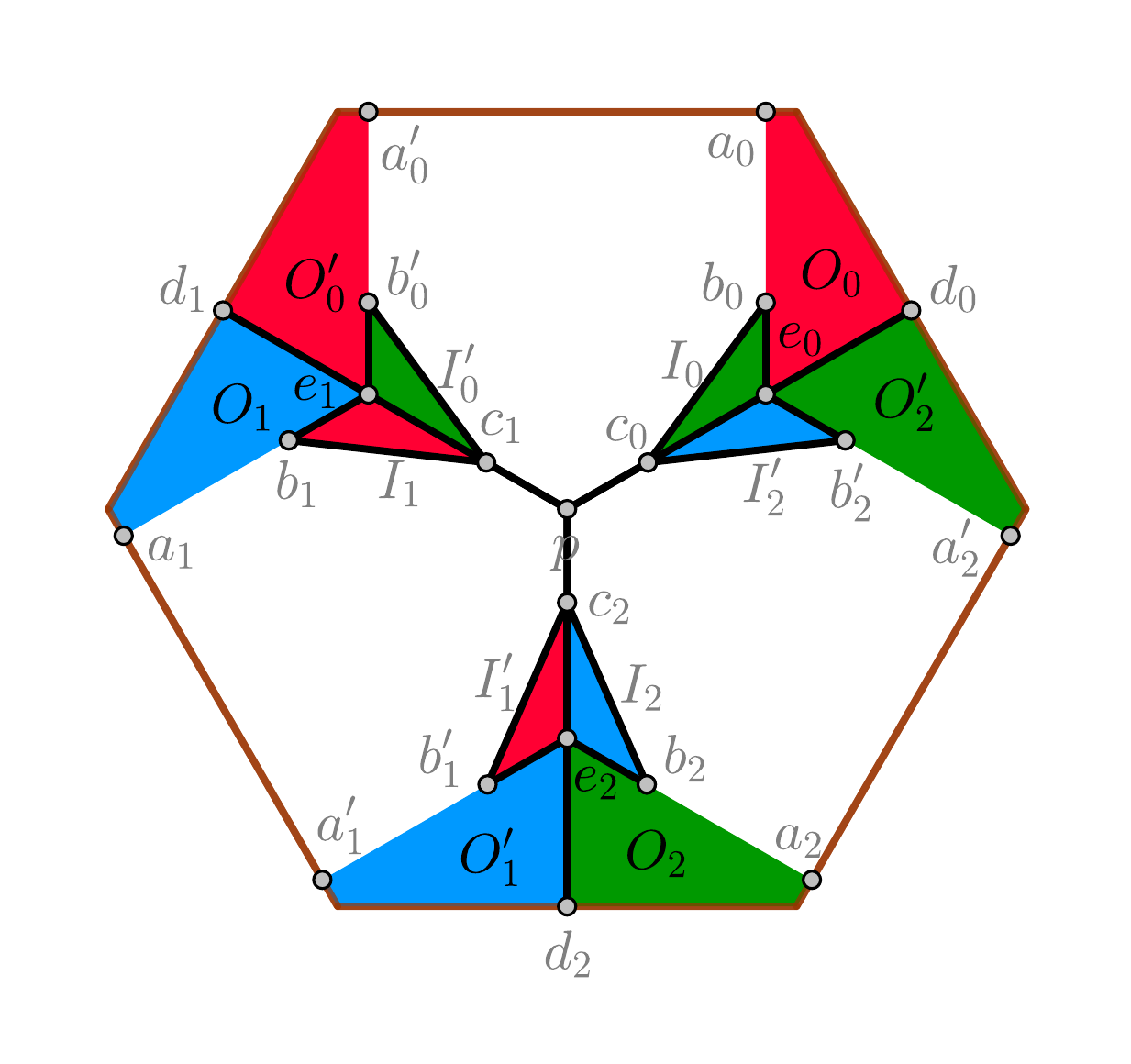}
    \includegraphics[width=0.45\textwidth]
    {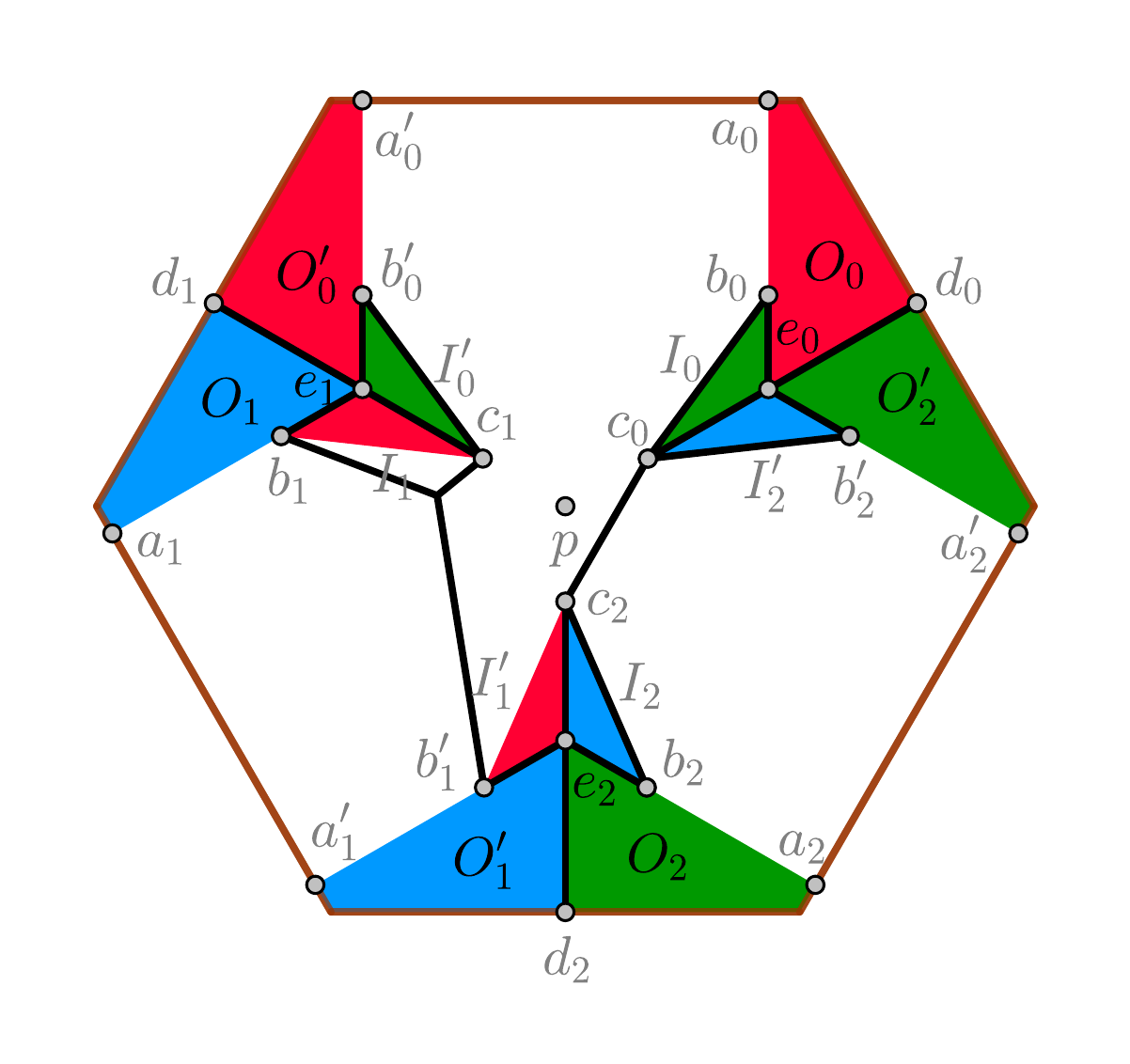} \\
    \includegraphics[width=0.45\textwidth]
    {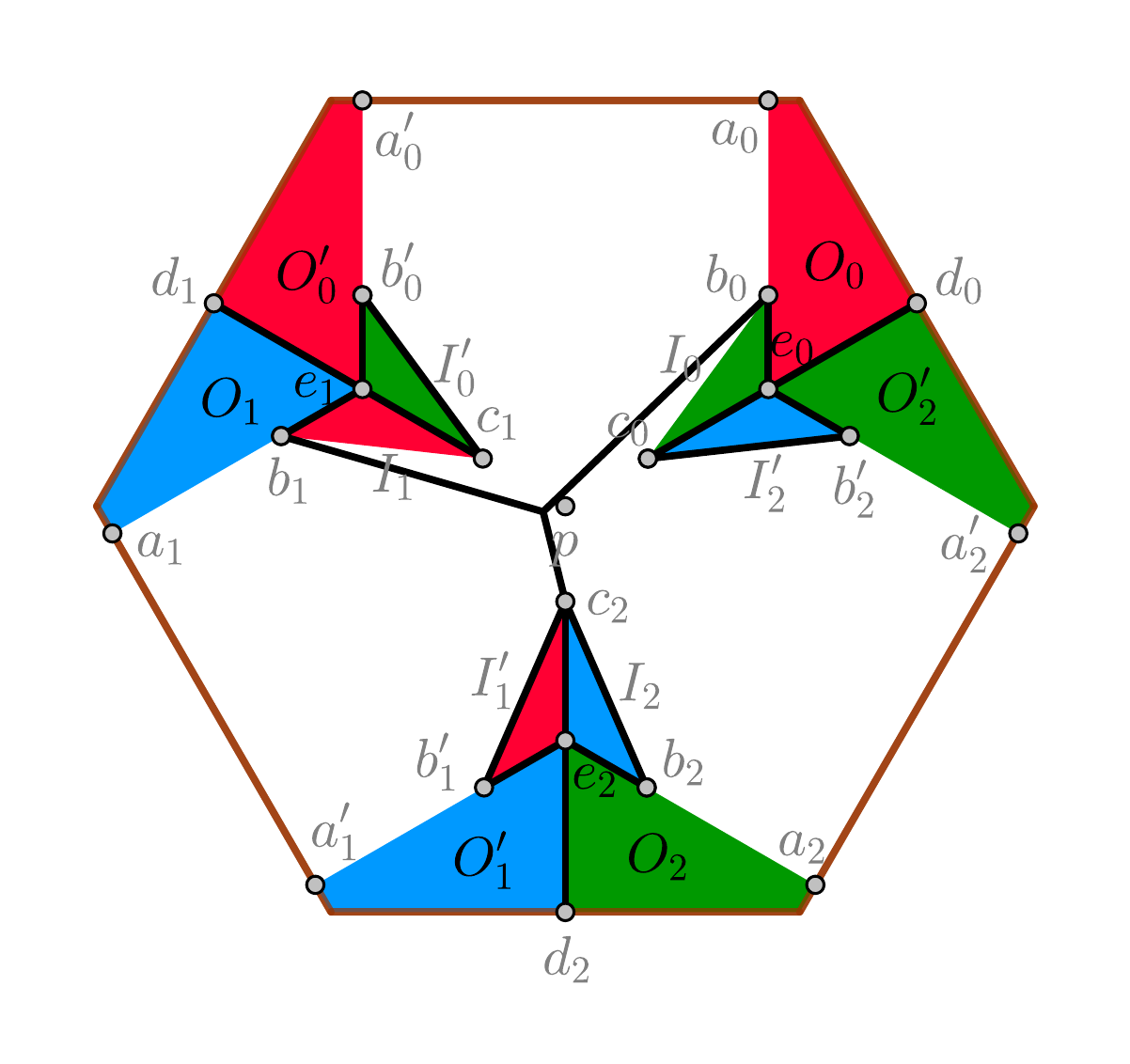}
    \includegraphics[width=0.45\textwidth]
    {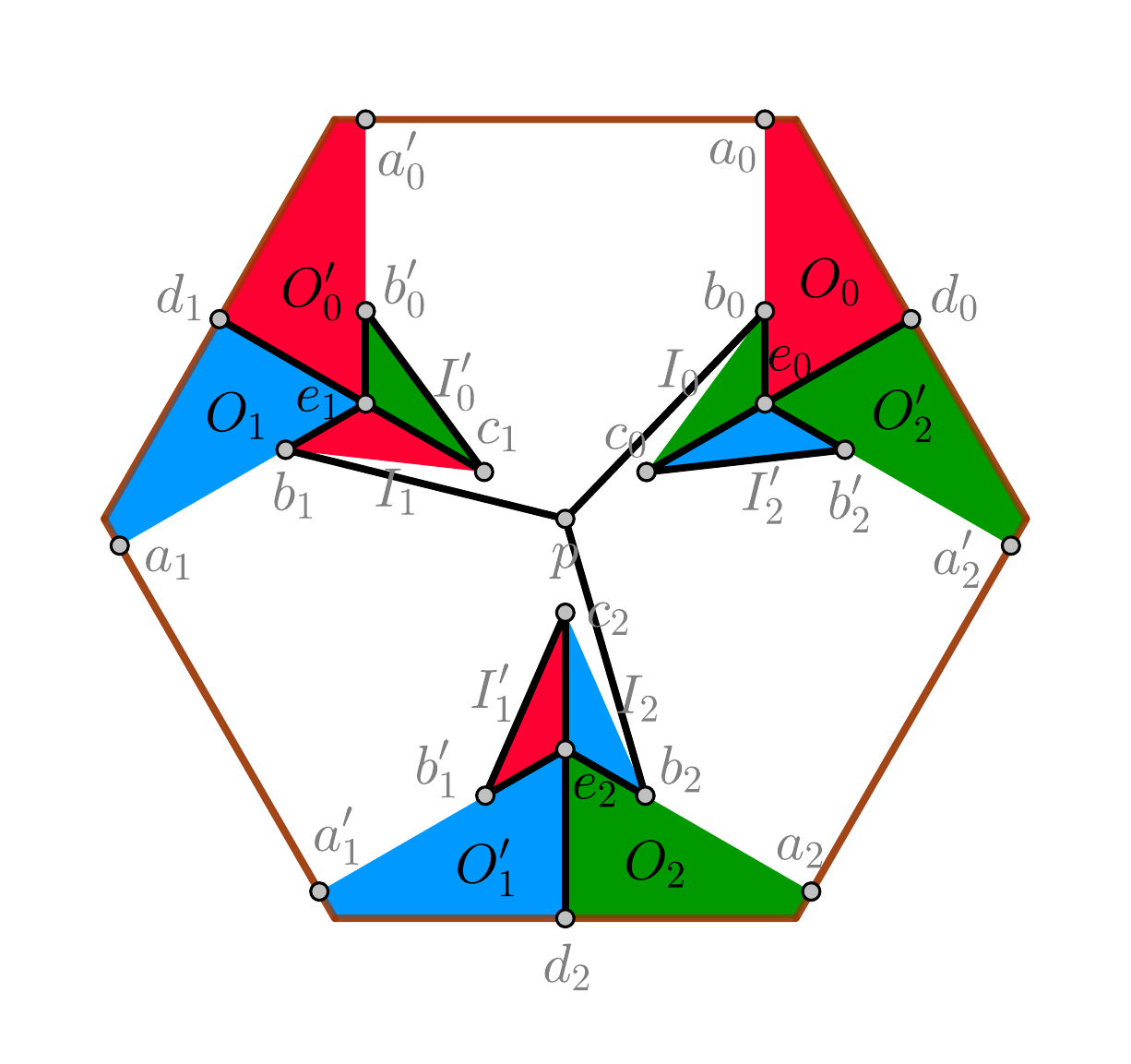}
    \caption{Top: Case~2.1.1. Case~2.1.2. Bottom: Case~2.1.3.1. Case~2.1.3.3.}
    \label{fig:case21}
\end{figure}

\paragraph*{Case 2: For any value of $i$, $O_i$ and $O'_i$ or $I_i$ and $I'_i$ are in a domain together.}

We may now divide into cases after how many values of $i$ satisfy that $I_i$ and $I'_i$ are in the same domain.
Let $c$ denote this number.
\paragraph*{Case 2.1: $c=0$.}
In this case, $O_i$ and $O'_i$ are in the same domain for each $i$.
The individual cases are shown in Figure~\ref{fig:case21}.

\paragraph*{Case 2.1.1: For no value of $i$ is $O_i\cup O'_i$ in a domain with $I_{i+1}$ or $I'_{i+1}$.}
In this case, the solution contains curves from $b_i$ to $c_i$ and $b'_i$ to $c_{i+1}$ for each $i$, since the objects $I_i$ and $I'_i$ are separated from other objects of the same color.
Furthermore, there must be a curve from $b_i$ to $b'_i$ bounding the domain containing $O_i\cup O'_i$.
It follows that the solution connects any two of the nine points $\bigcup_{i=0}^2\{b_i,b'_i,c_i\}$.
The cheapest solution that satisfies this is $\bigcup_{i=0}^2 b_ic_i\cup b'_ic_{i+1}\cup c_ip$, which has cost $\approx 1.85>1.67$.

\paragraph*{Case 2.1.2: $O_i\cup O'_i$ is in a domain with $I'_{i+1}$.}
Assume without loss of generality that $O_0\cup O'_0$ is in a domain with $I'_1$.
The mentioned domain separates $O_1\cup O'_1$ from $I_2$ and $I'_2$, and it separates $O_2\cup O'_2$ from $I'_0$.
The boundary of that domain contains a curve connecting $b'_0$ and $b'_1$ and one connecting $b_0$ and $c_2$.
It follows that the solution contains
\begin{itemize}
\item
a tree connecting $b'_0,c_1,b_1,b'_1$ ($b_1$ and $c_1$ are connected since $I_1$ is isolated from the other red objects),

\item
a tree connecting $b_0,c_2,b_2$ ($b_2$ and $c_2$ are connected since $I_2$ is separated from the other blue objects),

\item
a curve connecting $b'_2$ and $c_0$ (same reasons as above).
\end{itemize}

The optimal solution of this form consists of segments from $b_1,b'_1,c_1$ to their Fermat point and the segments $b_0c_0,b'_0c_1,b'_2c_0,c_0c_2,b_2c_2$, and it has cost $\approx 1.85>1.67$.

\paragraph*{Case 2.1.3: $O_i\cup O'_i$ is in a domain with $I_{i+1}$.}
Assume without loss of generality that $O_0\cup O'_0$ is in a domain with $I_1$.
That domain separates $O_2\cup O'_2$ from $I'_0$, and thus any solution contains a curve from $b'_0$ to $c_1$.
There must also be a curve connecting $b_0$ and $b_1$ on the boundary of the domain.
The solution also contains a curve from $b'_1$ to $c_2$, as $I'_1$ is isolated from the other red objects.

\paragraph*{Case 2.1.3.1: $O_1\cup O'_1$ is not in a domain with $I_2$ or $I'_2$.}
Any solution contains a curve from $b_1$ to $b'_1$ (bounding the domain of $O_1\cup O'_1$), one from $b_2$ to $c_2$ (bounding the domain of $I_2$), and one from $b'_2$ to $c_0$ (bounding the domain of $I'_2$).
To summarize, the solution contains
\begin{itemize}
\item
a tree connecting $b_0,b_1,b'_1,b_2,c_2$,

\item
a curve connecting $b'_0$ and $c_1$, and

\item
a curve connecting $b'_2$ and $c_0$.
\end{itemize}

The shortest such solution consists of segments from $b_0,b_1,c_2$ to their Fermat point and the segments $b'_0c_1,b'_1c_2,b_2c_2,b'_2c_0$, and it has cost $\approx 1.83>1.67$.

\paragraph*{Case 2.1.3.2: $O_1\cup O'_1$ is in a domain with $I'_2$.}
This case is covered by case~2.1.2.

\paragraph*{Case 2.1.3.3: $O_1\cup O'_1$ is in a domain with $I_2$.}
There is a curve bounding that domain connecting $b_1$ and $b_2$.
There are thus curves connecting all pairs $b_i,b_{i+1}$ and $b'_i,c_{i+1}$.
The shortest solution of that form is $\bigcup_{i=0}^2 b_ip\cup b'_ic_{i+1}$, which has total cost $\approx 1.83>1.67$.

\begin{figure}
    \centering
    \includegraphics[width=0.45\textwidth]
    {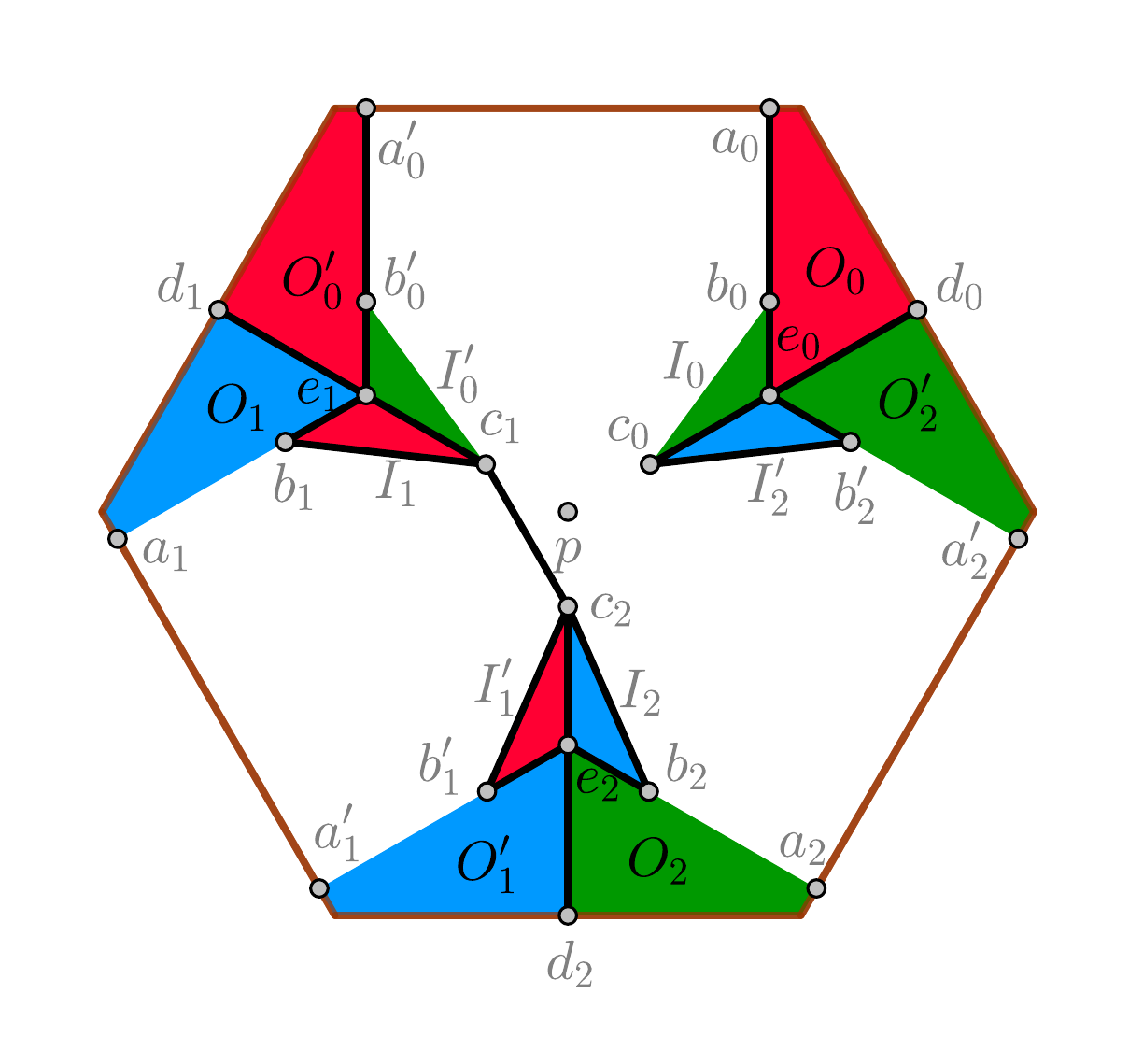}
    \includegraphics[width=0.45\textwidth]
    {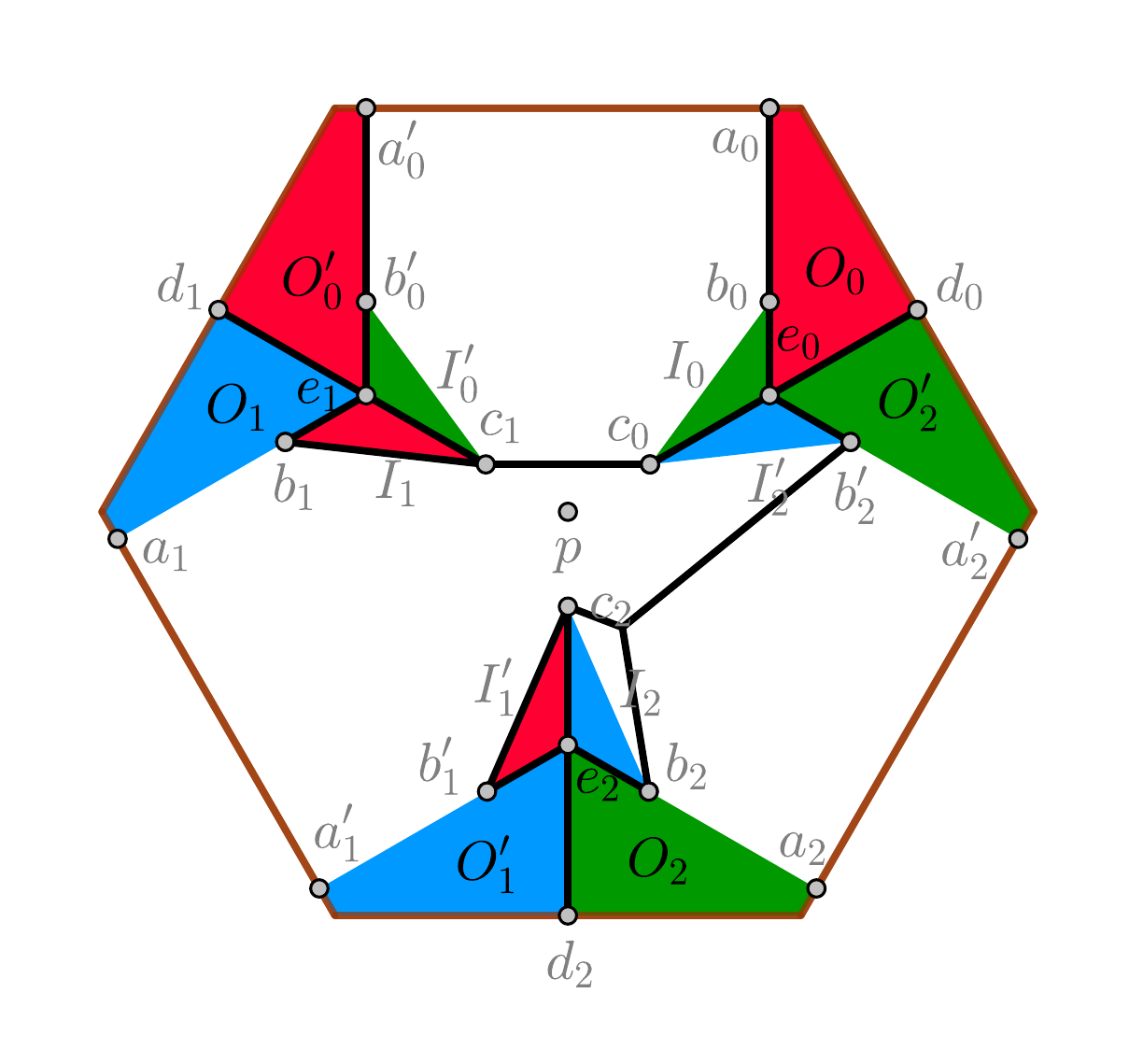} \\
    \includegraphics[width=0.45\textwidth]
    {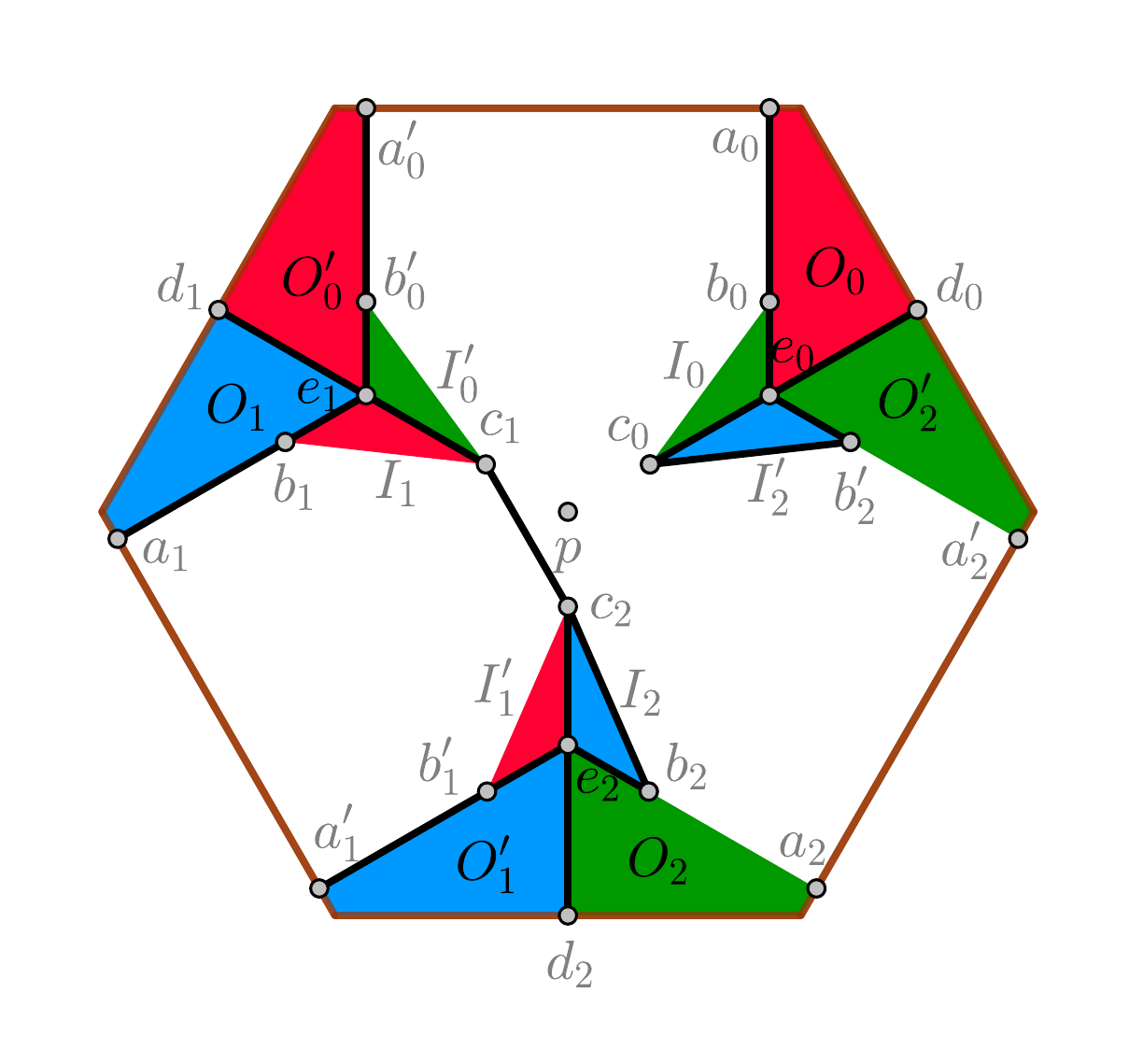}
    \includegraphics[width=0.45\textwidth]
    {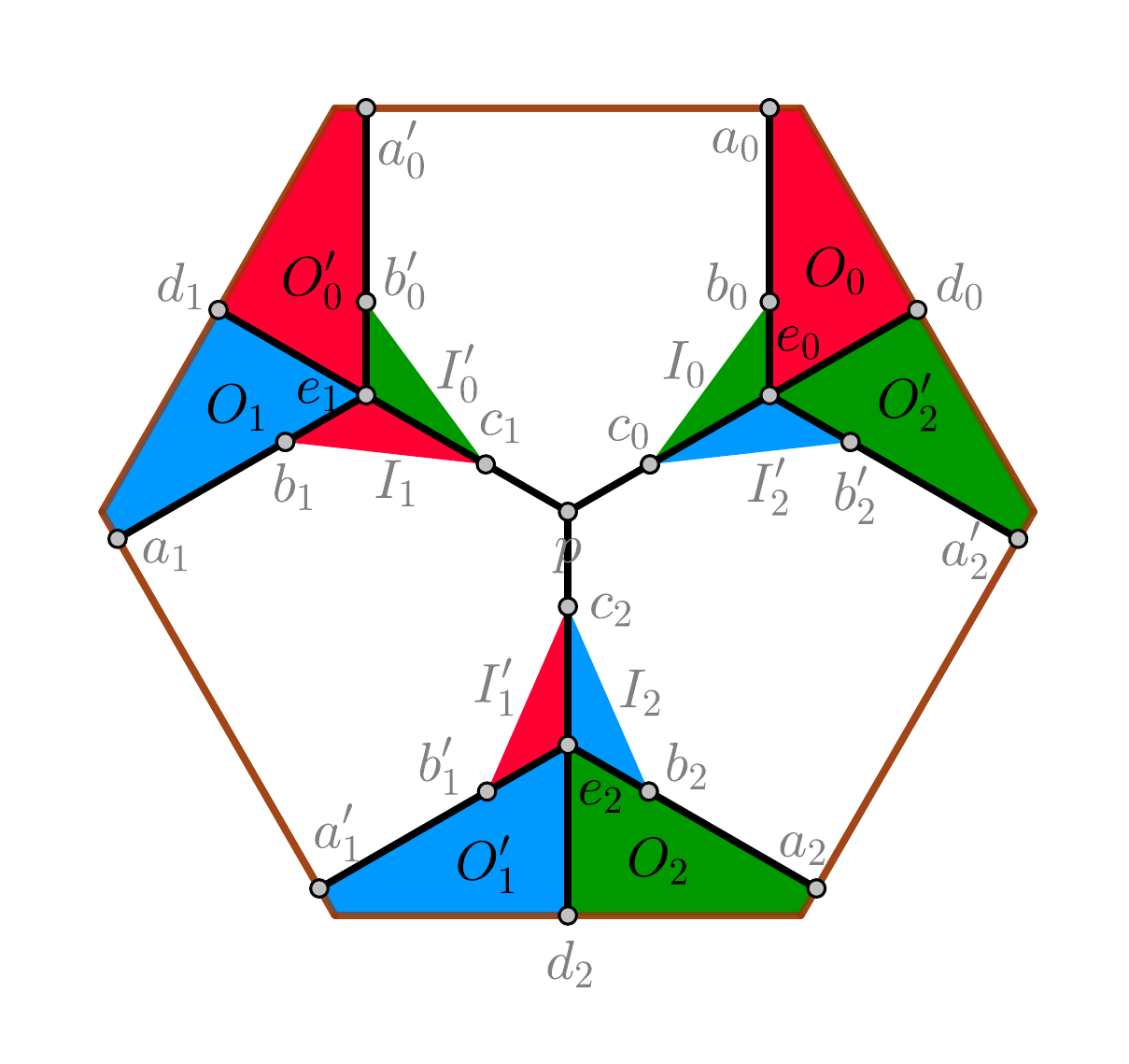}
    \caption{Top: Case~2.2.1. Case~2.2.2. Bottom: Case~2.3. Case~2.4.}
    \label{fig:case22}
\end{figure}

\paragraph*{Case 2.2: $c=1$.}
Assume without loss of generality that $I_0$ and $I'_0$ are in the same domain.
We thus also know that $O_1$ and $O'_1$ are in the same domain, as are $O_2$ and $O'_2$.
The domain of $I_0\cup I'_0$ separates $O_0$ and $O'_0$ from $I_1$ and $I'_1$, so $I_1$ and $I'_1$ are in separate domains, and the solution must contain a curve connecting $b_1$ and $c_1$ and one connecting $b'_1$ and $c_2$ bounding these domains.

Note that in order to separate $I_0\cup I'_0$ from $O_0$ and $O'_0$, the solution either contains a curve from $b_0$ to $b'_0$ or curves from $b_0$ and $b'_0$ to the boundary segment $a_0a'_0$.
We consider the latter option, which is $0.02$ cheaper.
It will follow from the analysis that even this is too expensive to get below $M+0.02$.
The individual cases as shown in Figure~\ref{fig:case22}.

\paragraph*{Case 2.2.1: $O_1\cup O'_1$ is not in a domain with $I_2$ or $I'_2$.}
The solution contains a curve from $b_1$ to $b'_1$ bounding the domain containing $O_1\cup O'_1$, and a curve connecting $b_2$ and $c_2$, and one connecting $b'_2$ and $c_0$ on the boundaries of the domains of $I_2$ and $I'_2$, respectively.
It follows that the solution contains a tree connecting $b_1,c_1,c_2,b'_1,b_2$.

The cheapest such solution is $a_0b_0\cup a'_0b'_0\cup b_1c_1\cup c_1c_2\cup b'_1c_2\cup b_2c_2\cup b'_2c_0$, which has cost $M+0.02$.
This ``second-best'' solution is the reason we have chosen the threshold $0.02$ in the lemma.

\paragraph*{Case 2.2.2: $O_1\cup O'_1$ is in a domain with $I'_2$.}
The solution contains a curve connecting $b_1$ and $c_0$ and one connecting $b'_1$ and $b'_2$ bounding that domain.
It also contains a curve connecting $b_2$ and $c_2$ bounding the domain of $I_2$, since $I_2$ is sepatated from the other blue objects.
The optimal solution consists of segments from $b_2,b'_2,c_2$ to their Fermat point and segments $a_0b_0, a'_0b'_0,b_1c_1,b'_1c_2,c_0c_1$, and it has cost $\approx 1.83>1.67$.

\paragraph*{Case 2.2.3: $O_1\cup O'_1$ is in a domain with $I_2$.}
The solution contains a curve connecting $b_1$ and $b_2$, and one connecting $b'_2$ and $c_0$ bounding the domain of $I'_2$.
The cheapest solution is similar to that in case 2.2.1, where the domain containing $O_1\cup O'_1\cup I_2$ has collapsed to zero width at $c_2$.

\paragraph*{Case 2.3: $c=2$.}
Assume without loss of generality that $I_0$ and $I'_0$ are in the same domain, as are $I_1$ and $I'_1$.
Furthermore, $I_2$ and $I'_2$ are separated, but $O_2$ and $O'_2$ are together.
As in case~2.2, we assume that $O_0$ and $O'_0$ are separated, as are $O_1$ and $O'_1$.
Otherwise, the cost of the solution will increase by at least $0.02$, and as the analysis will show, that is too much to stay below $M+0.02$.

The domain containing $I_1\cup I'_1$ separates $O_1$ and $O'_1$ from $I_2$ and $I'_2$.
It follows that there is a curve in the solution that connects $b_2$ and $c_2$, and one that connects $b'_2$ and $c_0$.
Likewise, the domain containing $I_0\cup I'_0$ separates $O_0$ and $O'_0$ from $I_1$ and $I'_1$.
Hence, the boundary of the domain containing $I_1\cup I'_1$ contains a curve connecting $c_1$ and $c_2$.
The cheapest solution is obtained as $a_0b_0\cup a'_0b'_0\cup a_1b_1\cup a'_1b'_1\cup b_2c_2\cup b'_2c'_0\cup c_1c_2$, as shown in Figure~\ref{fig:case22}, and the cost is $M$.
This is the optimal solution among all cases.

The segment $a_0b_0$ can be substituted by a curve from $f_0\in a_0a'_0$ to $b_0$, while keeping the cost below $M+0.02$, if and only if $\|f_0a_0\|<\sqrt{(0.24+0.02)^2-0.24^2}=0.1$.
Likewise for the other segments with an endpoint on the boundary of $T$.
These are exactly the solutions described in the lemma.

\paragraph*{Case 2.4: $c=3$.}
The cheapest solution is obtained when $O_i$ and $O'_i$ are separated for each $i$.
Furthermore, the solution contains a curve connecting $c_i$ and $c_{i+1}$ for each $i$ bounding the domain containing $I_i\cup I'_i$.
The cheapest such solution is $\bigcup_{i=0}^2 a_ib_i\cup a'_ib'_i\cup c_ip$ as shown in Figure~\ref{fig:case22}, which has cost $1.79>1.67$.
\end{proof}
\fi

\begin{theorem}\label{thm:nphard}
The problem GEOMETRIC $3$-CUT is NP-hard.
\end{theorem}

\begin{proof}
Let an instance $(\Phi,G(\Phi))$ of COLORED TRIGRID POSITIVE 1-IN-3-SAT be given and construct the tiles on top of $G(\Phi)$ as described.
Let $\mathcal T$ be the set of tiles and $\mathcal A$ the area that the tiles cover (i.e., $\mathcal A$ is a union of the hexagons).
We will cover any holes in $\mathcal A$ with completely red tiles, and
place red tiles all the way along the exterior boundary of $\mathcal A$.
Let $\mathcal R$ be the set of these added red tiles and let $I$ be the resulting instance of GEOMETRIC $3$-CUT.
It is now trivial how to place the fences in $I$ everywhere except in the interior of $\mathcal A$.

Consider a fence $\mathcal F$ to the obtained instance with cost $M$.
Let $M^*$ be the sum of the cost of an optimal solution to each tile in $\mathcal T$ plus the cost of the fence that must be placed along the boundaries of the added red tiles in $\mathcal R$.
We claim that if $\Phi$ is satisfiable, then a solution realizing the minimum $M^*$ exists.
Furthermore, if $M< M^*+1/50$, then $\Phi$ is satisfiable.

Suppose that $\Phi$ is satisfiable and fix a satisfying assignment.
Consider a clause tile where $E_x$, $E_y$, and $E_z$ meet.
Now, we choose the $v$-outer state, where $v\in\{x,y,z\}$ is the variable that is satisfied.
For each non-clause tile that covers a part of $E_w$ for a variable $w$ of $\Phi$, we choose the outer state if $w$ is true and the inner otherwise.
It is now easy to see that the curves form a fence of the desired cost.

On the other hand, suppose that $M<M^*+1/50$.
It follows that in each tile in $\mathcal T$, the cost exceeds the optimum by at most $1/50$.
Hence, the solution in each tile is homotopic to one of the optimal states as described in Lemma~\ref{lemma:solved:states}.
We now claim that the states of all tiles representing one variable must agree on either the inner or outer state.
Consider two adjacent tiles where one is in the inner state.
There are open curves with endpoints on the shared edge of the two tiles with a distance of more than $1/2-2\cdot 1/10=3/10$.
The other tile cannot be in the outer state, because then there would have to be an extra open curve of length at least $3/10$ to connect those endpoints.
It follows that the other tile must also be in the inner state.
Thus, both tiles are either in the inner or in the outer state, as desired.

We now describe how to obtain a satisfying assignment of $\Phi$.
Consider a clause tile where $E_x$, $E_y$, and $E_z$ meet and suppose the tile is in the $x$-outer state.
It follows from the above that each tile covering $E_x$ is in the outer state or, in the case of the clause tile, in the $x$-outer state.
Similarly, each non-clause tile covering only $E_y$ (resp.~$E_z$) is in the inner state and each clause tile covering a part of $E_y$ (resp.~$E_z$) is not in the $y$-outer (resp.~$z$-outer) state.
We now set $x$ to true and $y$ and $z$ to false and do similarly with the other clause tiles, and it follows that we get a solution to $\Phi$.

\iffull

\begin{figure}
\centering
\includegraphics{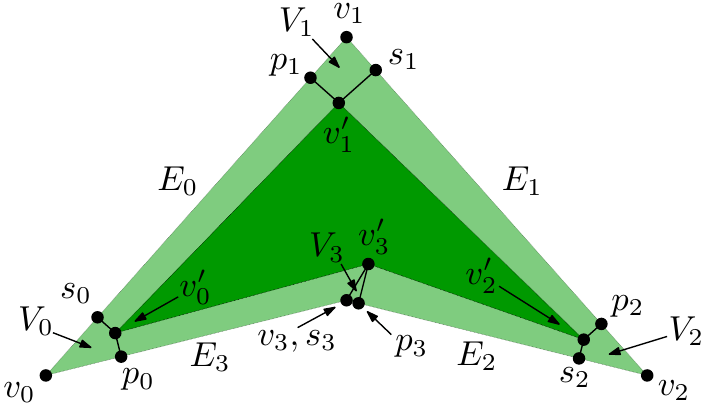}
\caption{An object $O$ and the rationalized version $O'\subset O$.}
\label{fig:rationalize}
\end{figure}

For each object $O$ with corner $v$ with an irrational coordinate, we choose a substitute $v'\in O$ with rational coordinates such that $\|vv'\|< \frac{1/50}{4n}$ and such that $v'$ only requires polynomially many bits to represent.
This results in an instance $I'$ where all objects are subsets of corresponding objects in $I$.
Let $C$ and $C'$ be the cost of the optimal solutions to $I$ and $I'$, respectively, and note that $C'\leq C$, as any solution to $I$ is also a solution to $I'$.
We claim that $C< C'+1/50$.
To prove this, consider a solution $\mathcal F$ to $I'$.
If $\mathcal F$ contains a curve $\pi'$ in the interior of an object $O$ of $I$, we move $\pi'$ to a curve $\pi$ on the boundary of $O$, as follows.

Let $O'\subseteq O$ be the object in $I'$ corresponding to $O$.
Let $v_0,\ldots,v_{k-1}$ be the vertices of $O$ in clockwise order and $v'_0,\ldots,v'_{k-1}$ the corresponding vertices of $O'$.
In the following, indices will be taken modulo $k$.
We divide the annulus $\mathcal D\mydef \inte O\setminus\inte O'$ into a region $E_i$ for each edge $v'_iv'_{i+1}$ and a region $V_i$ for each vertex $v'_i$ of $O'$, see Figure~\ref{fig:rationalize}.
We make a line segment from $v'_i$ to the closest point $p_i$ on $v_{i-1}v_i$ and one to the closest point $s_i$ on $v_iv_{i+1}$.
We then define quadrilaterals $V_i\mydef v'_is_iv_ip_i$ and $E_i\mydef v_iv_{i+1}p_{i+1}s_{i+1}$.
We now describe the modification we make on $\pi'$ in order to avoid $\mathcal D$.
If $\pi'$ intersects $V_i$, we remove $\pi'\cap V_i$ and instead add the segments $p_iv_i\cup v_is_i$.
Note that these added segments have total length less than $\frac 1{100n}$.
If $\pi'$ intersects $E_i$, we project each point in $\pi'\cap E_i$ to the closest point on $v_iv_{i+1}$.
This does not increase the length of the curve.
It follows that the modification of $\pi'$ made to avoid $\mathcal D$ increases the length by less than $\frac k{100n}$.
Performing the same operation for all objects of $I'$, we get a solution to $I$ with cost less than $C'+1/100$.
Hence, $C<C'+1/100$.

Let $M'\mydef \frac{\lceil 100M^*\rceil}{100}$, so that $M'$ is rational and $M^*\leq M'< M^*+1/100$.
We conclude by observing that if $C'\leq M'$, then $C<C'+1/100< M'+1/100<M^*+1/50$, and thus $\Phi$ is satisfiable.
On the other hand, if $\Phi$ is satisfiable, then $C'\leq C=M^*\leq M'$.
We can thus tell whether $\Phi$ is satisfiable or not by evaluating whether $C'\leq M'$. 
\fi
\end{proof}



\section{Approximation}
\label{sec:a}




The approach for $k=2$ from Section~\ref{sec:2} does not extend to $k\geq 3$ because
Lemma~\ref{lemma:opt_fence_polygons} does not apply:
The
arrangement $\mathcal A$ (formed by the free segments between the
corners $N$ of the input objects) is no longer guaranteed to contain an optimal
fence,
see Figure~\ref{fig:polygon_example}.
However, we can still search for an approximate solution in
 $\mathcal A$:
We show that the optimal fence $F_\mathcal A$ contained in
$\mathcal A$ has a cost which is at most $4/3$ times higher than the
true optimal fence $F^\star$ (Section~\ref{4/3}).
 We construct a corresponding lower-bound example with $|F_\mathcal
A|> 1.15\cdot|F^\star|$. (This factor is the conjectured Steiner
ratio, see
 Section~\ref{sec:lower}.)

 The graph-theoretic problem that we then have to solve in
 the weighted dual graph $G=(V,E)$ of $\mathcal A$ is
 the \emph{colored multiterminal cut problem}:
We have terminals of
 $k\geq 3$ different colors and want to make a cut that separates
 every pair of terminals of different colors.
 This problem is NP-hard, but we can use
 approximation algorithms, see Section~\ref{sec:approx_alg}.
 

\subsection{Upper bound  $|F_\mathcal A|/|F^\star|\le 4/3$}
\label{4/3}

\begin{theorem}\label{thm:approx}
$|F_\mathcal A|\leq 4/3\cdot |F^\star|$.
\end{theorem}
%
\begin{proof}
  
By Lemma~\ref{lemma:char} and Lemma~\ref{lemma:no_cycles}, we know
that after cutting an optimal fence $F^\star$ at all points of $N$,
the remaining components
are Steiner minimal trees with leaves in $N$ and internal \emph{Steiner vertices} of
degree $3$, where three segments make angles of $2\pi/3$.

  Consider such a Steiner tree $T$ (Figure~\ref {fig:approx2}a).
Since $T$ is embedded in the plane, the leaves can be enumerated in cyclic order as $v_1,\ldots,v_m$.
We will replace $T$ by a connected system $\bar T$ of fences that
connects the same set of leaves $v_1,\ldots,v_m$, but contains only
segments from the arrangement $\mathcal A$.
Furthermore, we prove that the total length of $\bar T$ is bounded as $|\bar T|\le \frac43 |T|$.
  Thus, carrying out this replacement for every Steiner tree leads to the fence $F_{\mathcal{A}}$ of the desired cost.
If $T$ consists of a single segment, we define $\bar T$ to be the same segment, in which case trivially $|\bar T|\leq \frac 43|T|$.
Assume therefore that $T$ has at least one Steiner vertex.

Let $T_{ij}$ be the path in $T$ from $v_i$ to $v_j$.
For each pair $\{i,j\}$, we define the path $\bar T_{ij}$
as the shortest path with the properties that
\begin{enumerate}
\item[a)]
$\bar T_{ij}$ has endpoints $v_i$ and $v_j$, and
\item[b)]
$\bar T_{ij}$ is \emph{homotopic} to $T_{ij}$: this means that $T_{ij}$ can be continuously deformed into $\bar T_{ij}$ while keeping the endpoints fixed at $v_i$ and $v_j$, without entering the interiors of the objects.
\end{enumerate}
It is clear that
\begin{itemize}
\item[c)]
$\bar T_{ij}$ is contained in the arrangement $\mathcal A$, and
\item[d)]
$\bar T_{ij}$ is at most as long as $T_{ij}$.
\end{itemize}

We will construct $\bar T$ as the union of paths $\bar T_{ij}$ that
are specified by a certain set $S$ of leaf pairs $\{i,j\}$, and we
will show that its total length is bounded $|\bar T|\le \frac43 |T|$.
The fact that $F_{\mathcal{A}}$ is a valid fence is ensured by our
choice of the set $S$, which we will now discuss.
  
%


 \begin{figure}
 \centering
 \includegraphics {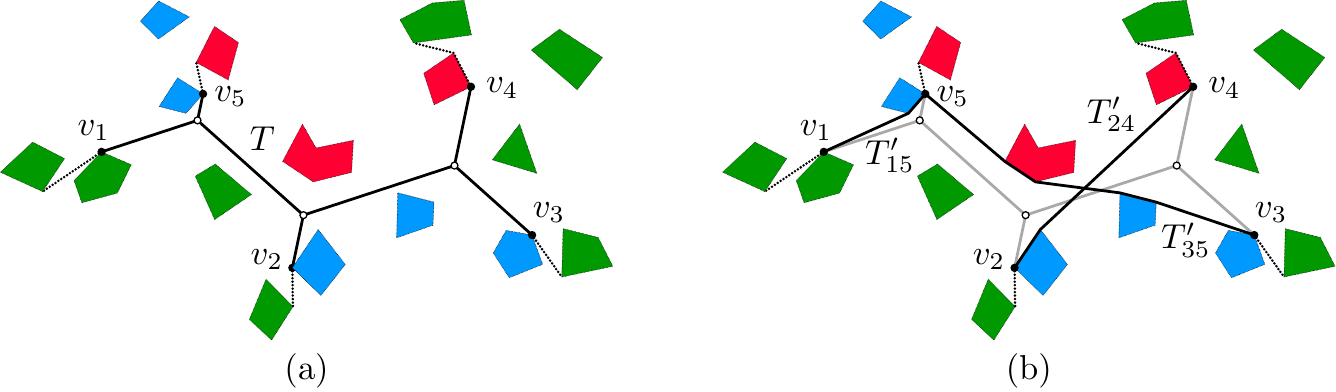}
 \caption{(a) a single Steiner tree $T$ with 5 terminals $v_1,\ldots,v_5$, part of a larger fence system $F^\star$. Steiner vertices are white, leaves are black.
   (b) The transformed graph $\bar T$, formed as the union of three
   shortest homotopic paths $\bar T_{15}$, $\bar T_{24}$, and
    $\bar T_{35}$.}
 \label{fig:approx2}
 \end{figure}

If we overlay all paths $T_{ij}$ for $\{i,j\}\in S$, we get
a multigraph $\widetilde T$, which has the same vertices as $T$ and
uses the edges of $T$, some of them multiple times.
%
%
We require 
these three properties:
\begin{enumerate}
\item\label{enum:dup1}
  Every edge of $T$ is used once or twice in $\widetilde T$.
\item\label{enum:dup2}
Every Steiner vertex of $T$ has even degree (4 or 6) in $\widetilde T$.
(By contrast, the degree in $T$ is always~$3$.)
\item\label{enum:dup3} Any two paths $T_{ij}$ and $T_{i'j'}$ that have a point of $T$ in
  common must \emph{cross} in the following sense:
If we assume, by relabeling if necessary, that $i<j$ and $i'<j'$, then 
$i\leq i'\leq j\leq j'$ or $i'\leq i\leq j'\leq j$.
\end{enumerate}
The last property is important to ensure that $\bar T$ is connected.
For this we use the following lemma, whose proof is given later.
For a path $P$ and points $x,y\in P$, we denote by $P[x,y]$ the subpath of $P$ from $x$ to $y$.
This is in general not well-defined unless $P$ is simple.
To make the notation unambiguous, we will assume that
 the points $x,y$ are associated to particular
 parameter values along the parameterization of~$P$.

\begin{lemma}
  \label{cross}
  Suppose that 
  the paths $T_{ij}$ and $T_{i'j'}$ cross in the sense of Property~\ref{enum:dup3}. 
  Then there exists a point $\bar x\in \bar T_{ij}\cap \bar T_{i'j'}$
  such that the path
 $$\bar T_{ij}[v_j,\bar x]\cup \bar T_{i'j'}[\bar x,v_{i'}]$$ is
 homotopic to the path
 \begin{math}
T_{ji'}
 \end{math}.
\end{lemma}
As we will prove shortly, 
Properties~\ref{enum:dup1} and~\ref{enum:dup3}
imply that for any two leaves $v_i$ and $v_j$ (where the pair $\{i,j\}$ is not necessarily in $S$),
 the set $\bar T$ contains a path from $v_i$ to $v_j$ that is homotopic to
 the path $T_{ij}$. 
This means that after replacing $T$ by $\bar T$ in $F^\star$, we get a
system of fences $F'$ that encloses and separates the same objects
as $F^\star$, and thus we have indeed produced a valid fence.

We now show the existence of a path in $\bar T$ homotopic to $T_{ij}$:
Let the vertices of $T_{ij}$ be $x_0,x_1,\ldots,x_{p+1}$ in order, such that $x_0\mydef v_i$ and $x_{p+1}\mydef v_j$.
For each $m=0,1,\ldots,p$, 
we select, by Property~\ref{enum:dup1}, a
path $T_{k_ml_m}$ with $\{k_m,l_m\}\in S$ that goes through the
directed edge $x_mx_{m+1}$ on the way from $v_{k_m}$ to  $v_{k_l}$.
This leads to a sequence of paths $T_{k_0l_0},T_{k_1l_1},\ldots,T_{k_pl_p}$, where $k_0=i$, $l_p=j$, and any two successive paths $T_{k_{m-1}l_{m-1}}$ and
$T_{k_ml_m}$ have the point $x_m$ in common, and hence cross, by 
Property~\ref{enum:dup3}.
Lemma~\ref{cross} implies that also the corresponding paths $\bar T_{k_{m-1}l_{m-1}}$
and $\bar T_{k_ml_m}$ have a common point $\bar x_m$ such that
\begin{align*}
& \bar U_m \mydef \bar T_{k_{m-1}l_{m-1}}[v_{l_{m-1}},\bar x_m]\cup \bar T_{k_ml_m}[\bar x_m,v_{k_m}]
\end{align*}
is homotopic
to $U_m := T_{{l_{m-1}}{k_m}}$. Now, define the paths
\begin{align*}
& W\mydef T_{k_0l_0}\cup U_1\cup T_{k_1l_1}\cup U_2\cup\ldots\cup U_p\cup T_{k_pl_p},\text{ and} \\
& \bar W\mydef \bar T_{k_0l_0}\cup \bar U_1\cup \bar T_{k_1l_1}\cup
  \bar U_2\cup\ldots\cup \bar U_p\cup \bar T_{k_p l_p}.
\end{align*}
The paths $T_{ij}$ and $W$ are homotopic, as
they have the same endpoints and $W$ is obtained by joining paths in
the simple tree~$T$.
 Also, $W$ and $\bar W$ are homotopic, as the corresponding subpaths are homotopic.
The paths $T_{ij}$ and $\bar W$ are thus homotopic, and $\bar W$ is contained in $\bar T$, so we are done.


\begin{proof}[Proof of  Lemma~\ref{cross}]
We first describe how 
 $T_{ij}$ can be continuously deformed into $\bar T_{ij}$ while remaining
 a polygonal path, moving one vertex at a time.
 We denote by
 $\hat T_{ij}$ the current path during this deformation procedure.
 
Consider the case that
 $\hat T_{ij}$ has a vertex $b$ which is not in $N$.
Let $a$ and $c$ be the neighboring vertices. 
We then move $b$ towards $c$, thus shortening the edge $bc$ while the
edge $ab$ sweeps over a region in the plane.
If $ab$ hits the corner 
of an object,
$\hat T_{ij}$ gets a new vertex $a'$ at this point.
 The segment $aa'$ will then remain static, and we continue the movement
of $b$ with $a'$ taking the role of $a$.
%
When $b$ eventually reaches $c$, the number of vertices of
$\hat T_{ij}$
that are not in $N$ has decreased by~1.
We repeat this process of contracting edges as long as there is a
vertex not in $N$.
Note that it is possible that the path crosses itself during the
deformation, or it
may have a vertex where it turns $180^\circ$ back on itself.
Such a vertex is known as a \emph{spur}, and it
can be easily eliminated by moving
it 
to the closest adjacent vertex.

 (For establishing Theorem~\ref{thm:approx}, we could already stop
the deformation procedure as soon as all vertices of
 $\hat T_{ij}$ are in $N$ and $\hat T_{ij}$ is free of spurs, because $\hat T_{ij}$
is contained in $\mathcal A$ and is at most as long as $T_{ij}$,
thus satisfying properties (c) and~(d).)
If $\hat T_{ij}$ is not yet the shortest homotopic path, it must
contain three consecutive vertices $abc$ such that the angle
at $b$ contains no object. In this case we can start the same
deformation move from $b$ towards $c$ as above. Temporarily, the vertex $b$
is an additional vertex not in $N$, but after the move,
$\hat T_{ij}$
is again a path connecting vertices of $N$. Since the number of such
paths that are not longer than the initial path
$T_{ij}$ is finite, the procedure must eventually terminate with the shortest
homotopic path~$\bar T_{ij}$.

We successively apply this procedure to the pairs ${ij}$ and ${i'j'}$.
We still have to prove the existence of a point $\bar x\in \bar T_{ij}\cap \bar T_{i'j'}$ with the property stated in the lemma.
We assume that the four corners $v_{i},v_{j},v_{i'},v_{j'}$ are
distinct because otherwise the statement follows easily if we choose
 a shared endpoint as~$\bar x$.

The proof uses that fact that the number of \emph{crossings} between
the paths $\hat T_{ij}$ and $\hat T_{i'j'}$ can only change by an even number
during a deformation. The definition of a crossing requires some
care, as the paths may share segments.
Assume that
$\hat T_{ij}$ is the path that is currently being deformed, while
$\hat T_{i'j'}$ is 
either the initial path
$ T_{i'j'}$ or the final deformed path
$\bar T_{i'j'}$.

%
Orient the paths $\hat T_{ij}$ and $\hat T_{i'j'}$ arbitrarily.
Consider a maximal subpath $Q$ that is shared between $\hat T_{ij}$
and $\hat T_{i'j'}$, possibly in opposite directions.
If $\hat T_{ij}$ enters and leaves $Q$ on the same side of $\hat T_{i'j'}$, we say that $\hat T_{ij}$ \emph{touches} $\hat T_{i'j'}$ at $Q$.
Otherwise, $\hat T_{ij}$ and  $\hat T_{i'j'}$  form a \emph{crossing}
at $Q$.
Here it is important that $\hat T_{i'j'}$ has no spurs, since at a
spur, the side on which $\hat T_{ij}$ enters or leaves $\hat T_{i'j'}$ is ill-defined.
If $Q$ contains an endpoint $q$ of one of the paths, 
we extend this path 
into the interior of the object in order
to determine the side of
$\hat T_{i'j'}$ on which $\hat T_{ij}$ enters or leaves $Q$ at $q$.

A crossing $Q$ of $\hat T_{ij}$ and  $\hat T_{i'j'}$ is a
\emph{homotopic crossing} if
it has the desired property for the lemma, namely
that $\hat T_{ij}[v_j,\hat x]\cup \hat
 T_{i'j'}[\hat x,v_{i'}]$ 
 for $\hat x\in Q$ is homotopic to
$T_{ji'}$. 
Clearly, this does not depend on the choice of $\hat x\in Q$, because
$Q$ is represented by a connected interval of parameters, both on 
$\hat T_{ij}$
and $\hat T_{i'j'}$.

When $\hat
T_{ij}$ is deformed by moving one vertex at a time, as described above,
it is easy to see that crossings can only appear or disappear in pairs:
It is not possible for a crossing $Q$ to appear or disappear by $\hat T_{ij}$ sliding over an endpoint $q'$ of $\hat T_{i'j'}$, since that would require $\hat T_{ij}$ to enter the interior of the object with $q'$ on the boundary.

Furthermore, a pair of crossings $Q_1,Q_2$ that appear or disappear will either both be homotopic crossings or non-homotopic crossings:
At the moment when the pair appears or disappears, the loop formed by the subpaths of $\hat T_{ij}$ and $\hat T_{i'j'}$ between $Q_1$ and $Q_2$ is empty and thus contains no objects.
Therefore, if $\hat x_1\in Q_1$ and $\hat x_2\in Q_2$, the paths $\hat T_{ij}[v_j,\hat x_1]\cup \hat T_{i'j'}[\hat x_1,v_{i'}]$ and $\hat T_{ij}[v_j,\hat x_2]\cup \hat T_{i'j'}[\hat x_2,v_{i'}]$ are homotopic.

During the deformation of $\hat T_{ij}$, each crossing $Q$ can move back and forth on $\hat T_{i'j'}$, expand and shrink.
However, it is clear that its character (homotopic versus non-homotopic) does not change during the deformation.

The initial number of crossings is~$1$, and the
 crossing, $Q$, is a
 homotopic crossing, since $T_{ji'}$ can be realized as a path $\hat T_{ij}[v_j,\hat x]\cup \hat
 T_{i'j'}[\hat x,v_{i'}]$ 
 for $\hat x\in Q$.
Hence the number of homotopic crossings of $\bar T_{ij}$ and $\bar T_{i'j'}$ is
odd, and in particular positive,
which establishes the claim.
\end{proof}


To bound the length of $\bar T$, we bound each path $\bar T_{ij}$, $\{i,j\}\in S$, by the corresponding path $T_{ij}$ in $T$.
 This upper estimate is simply the total length of $T$ plus the length of the duplicated edges of $T$.

Our first task is to construct the multigraph
$\widetilde T$.
By Property~\ref{enum:dup1},
this boils down to selecting which edges of $T$ to duplicate.
In order to
fulfill 
Property~\ref{enum:dup2}, we require that the degree of every inner
vertex of $\widetilde T$ becomes even.
(We show later that this is sufficient to ensure
that the edges of $\widetilde T$ can be partitioned into paths
$T_{ij}$ subject to Property~\ref{enum:dup3}.)

 \begin{lemma}\label{lemma:duplicate}
   The edges that should be duplicated can be chosen such that their total length is at most $|T|/3$.
 \end{lemma}
 \begin{proof}
For a particular tree, the optimum can be computed easily by dynamic programming, as follows.
We root $T$ at some arbitrary leaf. 
Consider a subtree $U$ rooted at some vertex $u$ of $T$ such that $u$ has one child $v$ in $U$.
We define $U_1$ and $U_2$ as the cost of the optimal set of duplicated edges in $U$, under the constraint that
the multiplicity of the edge $uv$ in $\widetilde T$ is 1 and 2, respectively.

By induction, we will establish that
\begin{equation}
  \label{tree-bound}
  2U_1+U_2 \le |U|.
\end{equation}
This gives 
$\min\{U_1,U_2\}\leq |U|/3$ and proves the lemma, since this also holds for $U=T$.

In the base case $U$ has only one edge.
Then $U_1=0$ and $U_2=\|uv\|=|U|$, and
~\eqref{tree-bound} holds.

If $U$ is larger, $v$ has degree 3, and two subtrees $L$ and $R$ are attached there.
 If $uv$ is not duplicated, then exactly one of the other edges incident to $v$ has to be duplicated in order for $v$ to get even degree in $\widetilde T$.
On the other hand, if $uv$ is duplicated, then either both or none of the other edges should be duplicated.
Hence, we can compute $U_1$ and $U_2$ by the following recursion:
\begin{align}
\label{T1}
  U_1 &= \min \{ L_1+R_2, L_2+R_1 \}
\\
\label{T2}
  U_2 &= \min \{ L_1+R_1, L_2+R_2 \} + \|uv\|
\end{align}
We therefore get
\begin{align}
  \label{au}
    U_1 &\le L_2+R_1 
\\
    U_1 &\le L_1+R_2
\end{align}
from~\eqref{T1} and
\begin{align}
  \label{way}
    U_2 &\le L_1+R_1 + \|uv\|
\end{align}
from~\eqref{T2}. Adding inequalities \thetag{\ref{au}--\ref{way}} and using the
inductive hypothesis~\eqref{tree-bound}
 for $L$ and $R$ gives
\begin{displaymath}
  2U_1+U_2 \le 2L_1+L_2+2R_1+R_2+\|uv\|
 \le |L|+|R|+\|uv\| = |U|.
\qedhere
\end{displaymath}
 \end{proof}

 \iffull
 The bound $|T|/3$ in Lemma~\ref{lemma:duplicate} cannot be improved,
 as shown by the graph $K_{1,3}$ with unit edge lengths.
This graph can appear as a Steiner tree in an optimal fence, see
Figure~\ref{fig:approx-lower}.
(But this does not mean that the factor $4/3$ in Theorem~\ref{thm:approx} cannot be improved.)
\fi

We now have a multigraph $\widetilde T$ where every internal vertex has even degree.
It follows that the edges of $\widetilde T$ can be partitioned into
leaf-to-leaf paths, much like when creating an Eulerian tour in a
graph where all vertices have even degree.

We still need to 
satisfy Property~\ref{enum:dup3}.
Whenever two paths $P_1$ and $P_2$ violate this property, we
repair this by swapping parts of the paths, without changing the number
of remaining violating pairs,
as follows:
The paths 
$P_1$ and $P_2$ 
must have a common
vertex, and thus also a common edge $uv$, because the maximum degree in
$T$ is $3$.
Orient $P_1$ and $P_2$ so that they use this edge in the direction
$uv$, and cut them at $v$ into
$P_1=Q_1\cdot R_1$
and $P_2=Q_2\cdot R_2$.
We now make a cross-over at $v$, forming the new paths
$
Q_1\cdot R_2$
and $
Q_2\cdot R_1$.
These new paths satisfy Property~\ref{enum:dup3}.
To check that we did not create any new violations, we observe that,
by Property~~\ref{enum:dup1},
no other path can use the edge $uv$, because the capacity of 2
is already taken by $P_1$ and $P_2$. Thus, all other paths can either
interact with $Q_1$ and $Q_2$, or with
$R_1$ and $R_2$. Thus, swapping the parts of $P_1$ and $P_2$ in the
other half of the tree $T$ does not affect Property~\ref{enum:dup3}.

We have thus established Theorem~\ref{thm:approx}.
\end{proof}

\subsection{Lower bound on $|F_\mathcal A|/|F^\star|$}\label{sec:lower}

We believe that the bound of Theorem~\ref{thm:approx} can be improved:
We bounded $|\bar T_{ij}|$ crudely by $|T_{ij}|$, using  only the triangle inequality,
and we did not use at all the fact that edges meet at
$120^\circ$.

We  construct an example that establishes a
 lower bound.
Gilbert and Pollack~\cite{gilbert1968steiner} conjectured that for any set of points in the plane, the ratio between the length of a minimum spanning tree and the length of a minimum Steiner tree is at most $\frac{2}{\sqrt{3}}\approx 1.15$, which is realized by the corners of an equilateral triangle.
The current best upper bound is $1.21$ by Chung and Graham~\cite{CG86}.
We show a lower bound on the ratio $|F_\mathcal A|/|F^\star|$ that corresponds to the conjectured Steiner ratio.

\begin{figure}
\centering
  \includegraphics[scale=0.8] {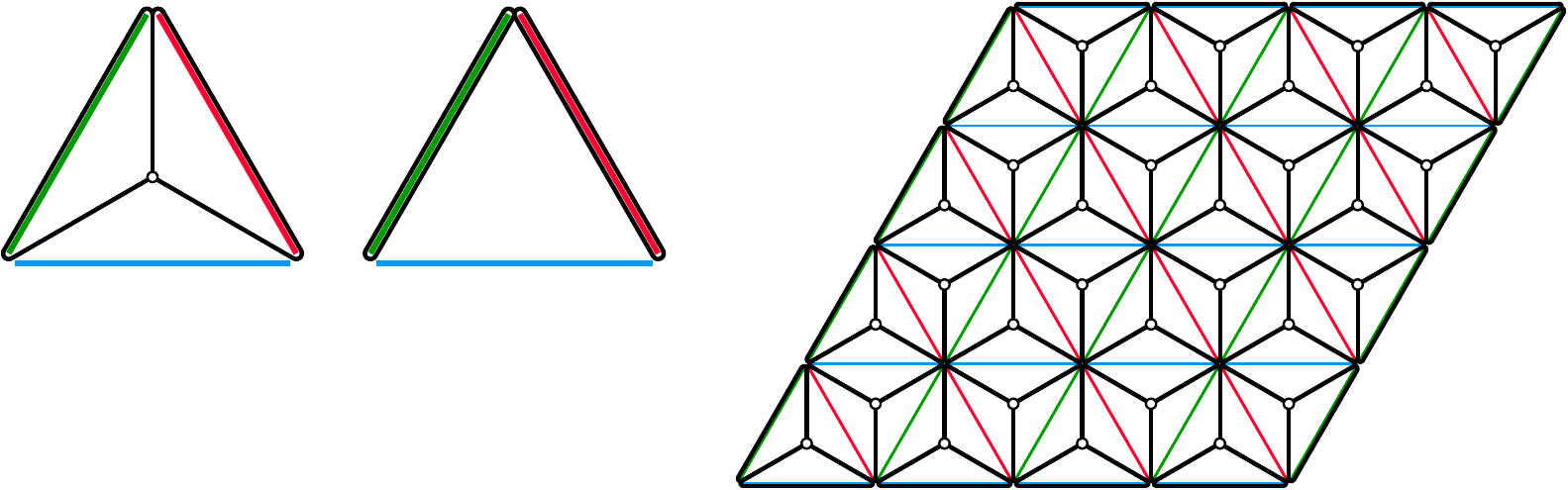} \\
\caption{The core (left) and repeated (right) construction from the proof of Lemma~\ref{lemma:lower}.}
\label{fig:approx-lower}
\end{figure}


\begin {lemma}\label{lemma:lower}
For every $\eps>0$, there is an instance of GEOMETRIC $3$-CUT for which 
  \[
    \frac {|F_{\cal A}|}{|F^\star|} \ge \frac 2{\sqrt3} - \varepsilon > 1.15-\eps.
  \]
\end {lemma}

\begin {proof}
The core idea is shown in
Figure~\ref{fig:approx-lower}:
Three very thin rectangles in different colors form an equilateral triangle with side length $\sqrt 3$.
The optimal fence uses the center of
the triangle as a Steiner vertex, whereas the 
fence $F_{\cal A}$ is restricted to follow the triangle edges.
This example, considered in isolation, gives only a ratio
$ {|F_{\cal A}|}/{|F^\star|} \approx ({4\sqrt3})/({3+2\sqrt3}) \approx
1.07$, because the outer boundary edges, which are common to both
fences, ``dilute'' the ratio.
 
So we 
  set $k = 1/\varepsilon$, and repeat the construction $k\times k$ times.
  We get $|F^\star| = 2k^2 \cdot 3 + 2k \cdot \sqrt 3$, versus $|F_{\cal A}| = 2k^2 \cdot 2 \sqrt 3 + 2k \cdot \sqrt 3$.
\end {proof}

\subsection{Finding a good fence in $\mathcal A$}\label{sec:approx_alg}

The problem of finding a small cut in a planar graph $G=(V,E)$ that separates $k$ different classes $T_1,\ldots,T_k\subset V$ of terminals was mentioned as a suggestion for future work by Dahlhaus et al.~\cite{dahlhaus1994complexity}, but we have not found any subsequent work on that except for the case $k=2$~\cite{borradaile2017multiple}.
We can, however, reduce the problem to the multiway cut problem in general graphs (also known as the multiterminal cut problem):
For each class $T_i$, we add an ``apex vertex'' $t_i$ which is connected to all vertices in $T_i$ by edges of infinite weight.
We then ask for the cut of minimum total weight that separates each pair $t_i,t_j$.
Dahlhaus et al.~gave a $(2-2/k)$-approximation algorithm for the problem.
In our setup, the running time will be $O(kn^8\log n)$.
The approximation ratio was since then improved to $3/2-1/k$ by C{\u{a}}linescu et al.~\cite{cualinescu2000improved}.
Finally, a randomized algorithm with approximation factor $1.3438$ was given by Karger et al.~\cite{karger2004rounding}, who also gave the best known bounds for various specific values of $k$.
Together with Theorem~\ref{thm:approx}, we obtain the following result.

\begin {theorem}
  \label{approximation-theorem}
  There is a randomized $4/3\cdot 1.3438$-approximation algorithm and
  a deterministic $(2-\frac{4}{3k})$-approximation algorithm for
  GEOMETRIC $k$-CUT, each of which runs in polynomial time.
  \qed
\end {theorem}


\iffull
\section*{Acknowledgements}
This work was initiated at the workshop
on \emph{Fixed-Parameter Computational Geometry}
at the Lorentz Center in Leiden
in May 2018. We thank the organizers and the Lorentz Center for a nice workshop and Michael Hoffmann for useful discussions during the workshop.
\fi

\bibliographystyle{plain}
\bibliography{bib}



\end{document}